\documentclass[12pt,british,nofootinbib,pra,reprint]{revtex4-1}
\usepackage[T1]{fontenc}
\usepackage[latin9]{inputenc}
\usepackage{amsmath}
\usepackage{amsthm}
\usepackage{amssymb}
\usepackage{graphicx}
\usepackage[raggedright]{subfigure}
\usepackage{esint}

\newtheorem{req}{Requirement}
  \theoremstyle{definition}
  \newtheorem{defn}{\protect\definitionname}
  \theoremstyle{plain}
  \newtheorem{ax}{\protect\axiomname}
  \theoremstyle{plain}
  \newtheorem{prop}{\protect\propositionname}
  \theoremstyle{plain}
  \newtheorem{lem}{\protect\lemmaname}
  \theoremstyle{plain}
  \newtheorem{cor}{\protect\corollaryname}
 \theoremstyle{definition}
  \newtheorem{example}{\protect\examplename}
   \theoremstyle{remark}
  
\theoremstyle{plain}
\newtheorem{thm}{\protect\theoremname}

\newcounter{biscompt}
\newtheorem{bis}[biscompt]{\protect\axiomname}

\newcounter{tercompt}
\newtheorem{ter}[tercompt]{\protect\axiomname}

\usepackage[unicode=true,pdfusetitle,
bookmarks=true,bookmarksnumbered=false,bookmarksopen=false,
breaklinks=true,pdfborder={0 0 0},pdfborderstyle={},backref=false,colorlinks=true]
{hyperref}

\makeatother
\usepackage[matrix,frame,arrow]{xy}
\usepackage{amsmath}
\newcommand{\bra}[1]{\left\langle{#1}\right\vert}
\newcommand{\ket}[1]{\left\vert{#1}\right\rangle}
\newcommand{\qw}[1][-1]{\ar @{-} [0,#1]}

\newcommand{\gate}[1]{*{\xy *+<.6em>{#1};p\save+LU;+RU **\dir{-}\restore\save+RU;+RD **\dir{-}\restore\save+RD;+LD **\dir{-}\restore\POS+LD;+LU **\dir{-}\endxy} \qw}

\newcommand{\measureD}[1]{*{\xy*+=+<.5em>{\vphantom{\rule{0em}{.1em}#1}}*\cir{r_l};p\save*!R{#1} \restore\save+UC;+UC-<.5em,0em>*!R{\hphantom{#1}}+L **\dir{-} \restore\save+DC;+DC-<.5em,0em>*!R{\hphantom{#1}}+L **\dir{-} \restore\POS+UC-<.5em,0em>*!R{\hphantom{#1}}+L;+DC-<.5em,0em>*!R{\hphantom{#1}}+L **\dir{-} \endxy} \qw}

\newcommand{\multigate}[2]{*+<1em,.9em>{\hphantom{#2}} \qw \POS[0,0].[#1,0];p !C *{#2},p \save+LU;+RU **\dir{-}\restore\save+RU;+RD **\dir{-}\restore\save+RD;+LD **\dir{-}\restore\save+LD;+LU **\dir{-}\restore}
\newcommand{\ghost}[1]{*+<1em,.9em>{\hphantom{#1}} \qw}

\newcommand{\Qcircuit}[1][0em]{\xymatrix @*[o] @*=<#1>}  
 \renewcommand{\Qcircuit}[1][0em]{\xymatrix @*=<#1>}

\newcommand{\pureghost}[1]{*+<1em,.9em>{\hphantom{#1}}}
\newcommand{\multiprepareC}[2]{*+<1em,.9em>{\hphantom{#2}}\save[0,0].[#1,0];p\save !C
  *{#2},p+RU+<0em,0em>;+LU+<+.8em,0em> **\dir{-}\restore\save +RD;+RU **\dir{-}\restore\save
  +RD;+LD+<.8em,0em> **\dir{-} \restore\save +LD+<0em,.8em>;+LU-<0em,.8em> **\dir{-} \restore \POS
  !UL*!UL{\cir<.9em>{u_r}};!DL*!DL{\cir<.9em>{l_u}}\restore}
\newcommand{\prepareC}[1]{*{\xy*+=+<.5em>{\vphantom{#1\rule{0em}{.1em}}}*\cir{l^r};p\save*!L{#1} \restore\save+UC;+UC+<.5em,0em>*!L{\hphantom{#1}}+R **\dir{-} \restore\save+DC;+DC+<.5em,0em>*!L{\hphantom{#1}}+R **\dir{-} \restore\POS+UC+<.5em,0em>*!L{\hphantom{#1}}+R;+DC+<.5em,0em>*!L{\hphantom{#1}}+R **\dir{-} \endxy}}
\newcommand{\poloFantasmaCn}[1]{{{}^{#1}_{\phantom{#1}}}}

\makeatletter

\newcommand{\R}{\mathbb{R}}

\newcommand{\set}[1]{\mathsf{#1}}
\newcommand{\grp}[1]{\mathsf{#1}}
\newcommand{\spc}[1]{\mathcal{#1}}

\def\d{{\rm d}}

\def\>{\rangle}
\def\<{\langle}

\newcommand{\st}[1]{\mathbf{#1}}

\newcommand{\map}[1]{\mathcal{#1}}

\newcommand{\Det}{{\mathsf{Det}}}
\newcommand{\St}{{\mathsf{St}}}
\newcommand{\Eff}{{\mathsf{Eff}}}
\newcommand{\Pur}{{\mathsf{Pur}}}
\newcommand{\Transf}{{\mathsf{Transf}}}

\newcommand{\rA}{\mathrm{A}}
\newcommand{\rB}{\mathrm{B}}
\newcommand{\rC}{\mathrm{C}}
\newcommand{\rE}{\mathrm{E}}
\newcommand{\rR}{\mathrm{R}}
\newcommand{\rS}{\mathrm{S}}

\newcommand{\Tr}{\mathrm{Tr}}
\newcommand{\cA}{\mathcal{A}}
\newcommand{\cB}{\mathcal{B}}
\newcommand{\cC}{\mathcal{C}}

\newcommand{\cU}{\mathcal{U}}

\makeatother

  \providecommand{\axiomname}{Axiom}
  \providecommand{\definitionname}{Definition}
  \providecommand{\examplename}{Example}
  \providecommand{\propositionname}{Proposition}
\providecommand{\corollaryname}{Corollary}
\providecommand{\lemmaname}{Lemma}
\providecommand{\remarkname}{Remark}
\providecommand{\theoremname}{Theorem}

\begin{document}

\title{Microcanonical thermodynamics in general physical theories}

\author{Giulio Chiribella}
\email{giulio@cs.hku.hk}
\affiliation{Department of Computer Science, University of Hong Kong, Pokfulam Road, Hong Kong}
\address{HKU Shenzhen Institute of Research and Innovation, Kejizhong 2$^{\rm nd}$ Road,  Shenzhen, China}
\affiliation{Canadian Institute for Advanced Research,
	CIFAR Program in Quantum Information Science, Toronto, ON M5G 1Z8}
\author{Carlo Maria Scandolo}
\email{carlomaria.scandolo@cs.ox.ac.uk}
\affiliation{Department of Computer Science, University of Oxford, Parks Road, Oxford, UK}

\begin{abstract}
 Microcanonical thermodynamics studies  the operations that can be performed on systems with well-defined energy. 
 So far, this  approach has been  applied to classical and quantum systems.  Here we extend it to arbitrary physical theories,   proposing two requirements for the development of a general microcanonical framework.    We then formulate  three resource theories, corresponding to three  different sets  of  basic operations:   \emph{i)} random reversible 
operations, resulting from  reversible dynamics with fluctuating parameters, \emph{ii)} noisy operations, generated by the interaction with  ancillas in the microcanonical state, and \emph{iii)} unital
operations, defined as the operations that preserve the microcanonical state. 
      We  focus our attention on  a class of physical theories, called sharp theories with purification, where these three sets of operations exhibit remarkable properties. Firstly,  each set is contained into the next. Secondly,   the convertibility of states by unital operations is completely characterised by  a majorisation criterion. Thirdly,  the three sets are equivalent in terms of state convertibility if and only if the dynamics allowed by  theory satisfy a suitable condition, which we call unrestricted reversibility.  Under this condition, we derive a duality between the resource theory of microcanonical thermodynamics and the resource theory of pure bipartite entanglement.  
        
       
\end{abstract}
\maketitle

\section{Introduction}

 In recent years, developments in
the field of nanotechnology have raised questions about thermodynamics
away from the thermodynamic limit \cite{delRio,Anders-thermo,Xuereb,Dahlsten-extractable,Aberg,Popescu-single-shot,Work-operational,Lostaglio-Muller,Gemmer-single-shot,Garner-one-shot1,Garner-one-shot2,Fluctuations1,Fluctuations2,Work-coherence,2ndlaw-control,2ndlaw-equality,Sparaciari-2}.
One way to address this new regime is to adopt a resource theoretic approach \cite{Quantum-resource-1,Quantum-resource-2}, which  starts from  a subset of operations that are regarded as  ``free'' or ``easy to implement''    
\cite{Resource-theories,Resource-monoid,Resource-knowledge,Resource-currencies}. 
A number of results in quantum thermodynamics have been obtained through this approach
 \cite{Athermality1,Horodecki-Oppenheim-2,2ndlaws,quantum2ndlaw,Lostaglio-Jennings-Rudolph,Lostaglio-coherence,Nicole-beyond,Nicole-non-commuting,Non-commuting-Bristol,David-non-commuting,Richens-2ndlaw,3rdlaw},   unveiling new connections between thermodynamics and information theory \cite{Athermality2,Egloff,Renes,Gibbs-preserving-maps,Alhambra-reversibility,Muller-generalization,Sparaciari}. 

The most basic instance of thermodynamics is for  systems with definite energy.
There, the natural choice of free state is the microcanonical state,~i.e., the uniform mixture of all states with the same energy.  In situations where the experimenter lacks control over the preparation of the system,  it is natural to expect that the system's state will fluctuate randomly from one experiment to the next, so that the overall statistics is described by the microcanonical state.  
The choice of free operations is less obvious.    The three main choices considered in the literature on quantum thermodynamics are:
\begin{enumerate}
\item \emph{random unitary channels} \cite{Uhlmann1,Uhlmann2,Uhlmann3},
arising from unitary dynamics with randomly fluctuating parameters;
\item  \emph{noisy operations} \cite{Local-Information,Horodecki-Oppenheim,Nicole}, generated by preparing ancillas
in the microcanonical state, turning on a    unitary dynamics, and discarding
the ancillas;
\item \emph{unital channels} \cite{Streater,Mendl-Wolf}, defined as the
quantum processes that preserve the microcanonical state.
\end{enumerate}
These three  sets are strictly different: the set of random unitary channels is strictly contained in the set of noisy
operations \cite{Shor} and the latter  is strictly contained in
the set of unital channels \cite{Haagerup-Musat}. In spite of this,
the three sets   are equivalent in terms of state convertibility
\cite{Nicole}. This means that  all the natural candidates for the
sets of free operations induce the same notion of resource. 
 This resource is generally called \emph{purity} and plays a fundamental role in many quantum protocols \cite{streltsov2016maximal}. 


In this paper we extend the paradigm of microcanonical thermodynamics from  quantum  theory to arbitrary physical theories  \cite{Hardy-informational-1,Barrett,Barnum-1,Chiribella-purification,Chiribella14,QuantumFromPrinciples,chiribella2016quantum,chiribella2017quantum}. We propose two  minimal requirements a probabilistic theory must satisfy in order to support a microcanonical description, and,   when these requirements are satisfied, we  provide  a general operational definition  of   random reversible, noisy, and unital operations.     We then focus on  a special class of theories, called \emph{sharp theories with purification}, which are appealing for the foundations of thermodynamics  \cite{TowardsThermo} and  have also been studied for their computational power \cite{Control-reversible,Lee-Selby-Grover} and interference properties \cite{HOP}.  
 In sharp theories with purification,  we show that the three sets of operations satisfy  the same inclusion relations as in quantum theory, with random reversible operations included in the set of noisy operations, and noisy operations  included in the set of unital operations.    For unital operations,  we characterise   the  convertibility of states completely in terms of  a suitable majorisation criterion. Thanks to this fact, one can take advantage of majorisation theory and develop quantitative measures of resourcefulness under unital operations. We call these measures {\em unital monotones} and show that they are in one-to-one correspondence with  Schur-convex functions \cite{Olkin}.  
 
Majorisation is a necessary and sufficient condition for state convertibility under unital operations.  For  random reversible and noisy operations, however, majorisation  is only necessary, as we illustrate explicitly with a counterexample.  Majorisation becomes sufficient   if and only if the dynamics allowed by the theory satisfy a suitable requirement, which we call {\em unrestricted reversibility}.   When this is the case, the sets of random reversible, noisy, and unital operations define the same notion of resource.
   Moreover,   one can prove the validity of an \emph{entanglement-thermodynamics duality}   \cite{Chiribella-Scandolo-entanglement}, which connects the three resource theories of purity and the resource theory of pure bipartite entanglement.  
All these results identify  sharp theories with purification and unrestricted reversibility as a privileged spot in the space of all possible physical theories. In this spot,  thermodynamic and information-theoretic features interact in a very similar way as they do in quantum mechanics. 

The paper is structured as follows. In section~\ref{sec:Framework}
we briefly  review  the framework, and in section~\ref{sec:constraints} we introduce constrained theories, a natural setting where to develop microcanonical thermodynamics.  In section~\ref{sec:microcanonical requirements} we propose two basic requirements for a well-posed microcanonical thermodynamics in general physical theories, and in section~\ref{sec:Resource-theories-of}
we extend  the three resource theories of random reversible, noisy, and unital operations from quantum theory to arbitrary probabilistic theories.   In section~\ref{sec:Axioms} we
 introduce the axioms and discuss their basic consequences for the class of theories we study. The implications of the axioms for microcanonical thermodynamics are examined in section~\ref{sec:microcanonical}. In section~\ref{sec:Purity-in-sharp} we establish majorisation as a necessary and sufficient condition for the convertibility of states under unital operations, and we characterise the corresponding monotones in terms of Schur-convex functions.  Remarkably, majorisation is not a sufficient criterion for convertibility under random reversible channels; we show a counterexample in section~\ref{sec:doubled}.
 In section~\ref{sec:sharp theories unrestricted} we determine when the three resource theories are equivalent in terms of state convertibility.
Finally, in section~\ref{sec:monotones} we establish the duality between the three resource theories of microcanonical thermodynamics and the resource theory of entanglement. Conclusions are
drawn in section~\ref{sec:Conclusions}.

\section{Framework\label{sec:Framework}}

Toy models of physical theories beyond classical and quantum mechanics can be formulated in the language of general probabilistic theories  \cite{Chiribella-purification,Chiribella-informational,hardy2011,Hardy-informational-2,hardy2013,Chiribella14,QuantumFromPrinciples,dariano2017quantum}, an umbrella term used to describe frameworks dealing with the notions of state, transformation, and measurement, along with a set of rules to assign probabilities to measurement outcomes. 
Specifically, this paper adopts  the framework  known as {\em operational-probabilistic
theories (OPTs)} \cite{Chiribella-purification,Chiribella-informational,hardy2011,Hardy-informational-2,hardy2013,Chiribella14,QuantumFromPrinciples,dariano2017quantum}. The OPT framework differs from other frameworks for general probabilistic theories, such as the convex set framework  \cite{Barrett,Wilce-formalism,Barnum-2,Barnum2016}, in the particular way it treats the composition of systems.   While in the convex set framework one generally starts  from convex sets associated with individual systems, and builds composites from them, the OPT framework takes  the composition of physical processes as primitive.  
   Mathematically, the ``top-down'' approach  of the OPT framework is underpinned by the graphical language of circuits \cite{Abramsky2004,Coecke-Kindergarten,Coecke-Picturalism,Selinger,Coecke2016,Coecke2017picturing}.    
In this section we give an informal presentation,  referring the reader to the original articles for a more in-depth discussion.

\subsection{States, transformations, and measurements}

OPTs describe the  experiments that can be performed on a given set of systems by a given set of physical processes. The framework is based on a primitive notion of composition, whereby every pair of physical systems $\rA$ and $\rB$ can be combined into a composite system, denoted by $\rA\otimes \rB$.   Physical processes can be connected in sequence or in parallel to build
circuits, such as
\begin{align}\label{circuit}
\begin{aligned}\Qcircuit @C=1em @R=.7em @!R { & \multiprepareC{1}{\rho} & \qw \poloFantasmaCn{\rA} & \gate{\cA} & \qw \poloFantasmaCn{\rA'} & \gate{\cA'} & \qw \poloFantasmaCn{\rA''} &\measureD{a} \\ & \pureghost{\rho} & \qw \poloFantasmaCn{\rB} & \gate{\cB} & \qw \poloFantasmaCn{\rB'} &\qw &\qw &\measureD{b} }\end{aligned}~.
\end{align}
In this example, $\mathrm{A}$, $\mathrm{A}', \rA'', \rB$, and $\rB'$ are \emph{systems}, $\rho$
is a bipartite \emph{state}, $\mathcal{A}$, $\mathcal{A}'$ and $\mathcal{B}$
are \emph{transformations}, $a$ and $b$ are \emph{effects}. Note that inputs are on the left and outputs are on the right.

For generic systems $\rA$ and $\rB$, we denote
by 
\begin{itemize}
\item $\mathsf{St}\left(\mathrm{A}\right)$ the set of states of system
$\mathrm{A}$,
\item $\mathsf{Eff}\left(\mathrm{A}\right)$ the set of effects on $\mathrm{A}$,
\item $\mathsf{Transf}\left(\mathrm{A},\mathrm{B}\right)$ the set of transformations
from $\mathrm{A}$ to $\mathrm{B}$, and by $\mathsf{Transf}\left(\mathrm{A}\right)$
the set of transformations from $\mathrm{A}$ to $\mathrm{A}$,
\item   $\mathcal B \circ \mathcal A$   (or $\mathcal B \mathcal A$, for short) the sequential composition of two transformations $\mathcal A$ and $\mathcal B$, with the input of $\mathcal B$ matching the output of $\mathcal A$,
\item  $\map I_{\rA}$  the identity transformation on system $\rA$, represented by   the plain wire   $\Qcircuit @C=1em @R=.7em @!R {  & \qw \poloFantasmaCn{\rA}  & \qw }$
\item $\mathcal{A}\otimes\mathcal{B}$ the parallel composition (or tensor
product) of the transformations $\mathcal{A}$ and $\mathcal{B}$.
\end{itemize}
Among the list of valid physical systems, every OPT includes the trivial system $\mathrm{I}$, corresponding to the degrees
of freedom ignored by  theory. The trivial system acts as a unit
for the composition of systems: for every system $\mathrm{A}$, one has $\mathrm{I}\otimes\mathrm{A}=\mathrm{A}\otimes\mathrm{I}=\mathrm{A}$.

States (resp.\ effects) are transformations with the trivial system
as input (resp.\ output). Circuits with no external wires, like the
circuit in Eq.~\eqref{circuit},  are called {\em scalars}.  
We will often use  the notation $\left(a\middle|\rho\right)$ to denote the scalar 
\begin{align}
\left(a|\rho\right)~:=\!\!\!\!\begin{aligned}\Qcircuit @C=1em @R=.7em @!R { & \prepareC{\rho}    & \qw \poloFantasmaCn{\rA}  &\measureD{a}}\end{aligned}~,
\end{align}
and of the notation $\left(a\middle|\mathcal{C}\middle|\rho\right)$
to denote the scalar \begin{align}
\left(a\right|\cC\left|\rho\right)~:=\!\!\!\!\begin{aligned}\Qcircuit @C=1em @R=.7em @!R { & \prepareC{\rho}    & \qw \poloFantasmaCn{\rA}  &\gate{\cC}  &\qw \poloFantasmaCn{\rB} &\measureD{a}}\end{aligned}~.
\end{align}
In the OPT framework,  the scalar $\left(a\middle|\rho\right)$  is identified with a real number in
the interval $\left[0,1\right]$, interpreted as the probability of a joint occurrence of the state $\rho$ and the effect $a$ in an experiment consisting of a state preparation  (containing $\rho$ as one of the possible states), followed by a measurement (containing $a$ as one of the possible effects).  

The fact that scalars are real numbers induces a notion of sum for
transformations, so that the sets $\mathsf{St}\left(\mathrm{A}\right)$,
$\mathsf{Transf}\left(\mathrm{A},\mathrm{B}\right)$, and $\mathsf{Eff}\left(\mathrm{A}\right)$
become spanning sets of \emph{real} vector spaces.   By {\em dimension of the state space}  $\St\left( \rA\right) $, we mean the dimension of the vector space spanned by the states of $\rA$.  


\subsection{Tests and channels}

In general, a physical process can be non-deterministic, i.e.\ it can result into a set of alternative transformations, heralded by different outcomes, which can (at least in principle) be accessed by an experimenter.
  General non-deterministic processes are described by 
 \emph{tests}:  a test from $\mathrm{A}$ to $\mathrm{B}$ is a collection
of transformations $\left\{ \mathcal{C}_{i}\right\} _{i\in\mathsf{X}}$
from $\mathrm{A}$ to $\mathrm{B}$, where  $\mathsf{X}$ is the set of outcomes. If $\mathrm{A}$ (resp.\ $\mathrm{B}$)
is the trivial system, the test is called a \emph{preparation-test}
(resp.\ \emph{observation-test}). If the set of outcomes $\mathsf{X}$ contains a single element, we say that the test is \emph{deterministic}, because one, and only one transformation can take place.    We will denote the sets of deterministic states, transformations, and effects as $\Det\St \left( \rA\right) $,  $\Det\Transf\left( \rA,\rB\right) $, and $\Det\Eff \left( \rA\right) $  respectively.   We refer
to deterministic transformations as \emph{channels}.   

A tranformation $\mathcal{U}$
from $\mathrm{A}$ to $\mathrm{B}$ is called \emph{reversible} if
there exists another  transformation $\mathcal{U}^{-1}$ from $\mathrm{B}$ to $\mathrm{A}$
such that $\mathcal{U}^{-1}\mathcal{U}=\mathcal{I}_{\mathrm{A}}$
and $\mathcal{U}\mathcal{U}^{-1}=\mathcal{I}_{\mathrm{B}}$.   It is not hard to see that reversible transformations are deterministic, i.e. they are ``channels''.  
If there exists a reversible transformation transforming $\mathrm{A}$ into
$\mathrm{B}$, we say that $\mathrm{A}$ and $\mathrm{B}$ are \emph{operationally
equivalent}, denoted by $\mathrm{A}\simeq\mathrm{B}$.  Physically, this means that every experiment performed on system $\rA$ can be (at least in principle) converted into an experiment on system $\rB$, and vice versa. The composition
of systems is required to be \emph{symmetric}, meaning that $\mathrm{A}\otimes\mathrm{B}\simeq\mathrm{B}\otimes\mathrm{A}$.
Physically, this means that for every pair of systems $\rA$ and $\rB$ there exists a reversible transformation  that swaps $\rA$ with $\rB$.

\subsection{Pure transformations}  
The  notion  of pure transformation plays centre stage in our work. Intuitively, pure transformations  represent the  most fine-grained processes allowed by the theory.    To make this intuition precise, we need a few definitions. 

The first definition is  \emph{coarse-graining}---the operation of joining two or more outcomes of a test into
a single outcome:  the test $\left\{ \mathcal{C}_{i}\right\} _{i\in\mathsf{X}}$
is a \emph{coarse-graining} of the test $\left\{ \mathcal{D}_{j}\right\} _{j\in\mathsf{Y}}$
if there exists a partition $\left\{ \mathsf{Y}_{i}\right\} _{i\in\mathsf{X}}$
of $\mathsf{Y}$ such that 
\begin{align} 
\mathcal{C}_{i}=\sum_{j\in\mathsf{Y}_{i}}\mathcal{D}_{j}  , \qquad \forall  i\in\mathsf   X \, .
\end{align}
We say that the test $\left\{ \mathcal{D}_{j}\right\} _{j\in\mathsf{Y}}$
is a \emph{refinement of  the test} $\left\{ \mathcal{C}_{i}\right\} _{i\in\mathsf{X}}$.   A transformation $\map D_j$ with $j$ in the set $\set Y_i$ is a {\em refinement of the transformation}  $\map C_i$.    

 Pure transformations are the most refined transformations:
\begin{defn}
A transformation
$\mathcal{C}\in\mathsf{Transf}\left(\mathrm{A},\mathrm{B}\right)$
is \emph{pure} if it has only trivial refinements, namely refinements
$\left\{ \mathcal{D}_{j}\right\}$ of the form $\mathcal{D}_{j}=p_{j}\mathcal{C}$,  
where $\left\{ p_{j}\right\} $ is a probability distribution.
\end{defn} 
We denote the sets of pure transformations, pure states, and pure effects as   $\mathsf {PurTransf}  \left( \rm A, \rm B\right)  $, $\mathsf{PurSt}  \left(   \rm A \right) $, and $\mathsf{PurEff} \left( \rm A\right) $  respectively.   As usual, non-pure
states are called \emph{mixed}.

\subsection{Purification}    

Another  key notion in our paper is the notion of purification  \cite{Chiribella-purification,QuantumFromPrinciples}.  Consider a bipartite system $\rA\otimes \rB$ in the state $\rho_{\rA\rB}$.  The state of system $\rA$ alone is obtained  by \emph{discarding} system $\rB$---that is, by applying a channel that transforms system $\rB$ into the trivial system.   Discarding operations are represented by  \emph{deterministic effects}, i.e.\ deterministic transformations with trivial output.    In quantum theory, every system has one and only one deterministic effect, corresponding to the partial trace on the Hilbert space of the system.
 
 Given a deterministic effect $  e  \in  \Det\Eff \left( \rB\right) $, the corresponding \emph{marginal state}  is  
\begin{equation}\label{marginal}
\begin{aligned}\Qcircuit @C=1em @R=.7em @!R { &\prepareC{\rho_\rA}  & \qw \poloFantasmaCn{\rA} &\qw }
\end{aligned}
  ~:= \!\!\!\! \begin{aligned}\Qcircuit @C=1em @R=.7em @!R { &\multiprepareC{1}{\rho_{\rA\rB}}  & \qw \poloFantasmaCn{\rA} &\qw \\ & \pureghost{\rho_{\rA\rB}}  & \qw \poloFantasmaCn{\rB} & \measureD{e} }\end{aligned}~,
\end{equation}
or, in formula, $\rho_\rA   :=   \left( \map I_\rA \otimes e\right)   \rho_{\rA\rB}$.   

When  $\rho_{\rA\rB}$  is pure and Eq.~\eqref{marginal} is satisfied for some deterministic effect  $e$, we say that $\rho_{\rA\rB}$ is a \emph{purification} of $\rho_\rA$ and we call $\rB$ the \emph{purifying system}  \cite{Chiribella-purification,QuantumFromPrinciples}.
\begin{defn}\label{def:uniqueness} A purification  $\Psi  \in  \Pur\St \left( \rA\otimes \rB\right) $ is \emph{essentially unique} \cite{QuantumFromPrinciples} if for every  pure state   $\Psi'  \in  \Pur\St \left( \rA\otimes \rB\right) $  and every deterministic effect $e' \in  \Det\Eff  \left( \rB\right)$   satisfying  the purification condition
\begin{equation}    \begin{aligned}\Qcircuit @C=1em @R=.7em @!R { &\multiprepareC{1}{\Psi'}  & \qw \poloFantasmaCn{\rA} &\qw \\ & \pureghost{\Psi'}  & \qw \poloFantasmaCn{\rB} & \measureD{e'} }  
  \end{aligned}~
     = \!\!\!\!\begin{aligned}\Qcircuit @C=1em @R=.7em @!R { &\prepareC{\rho_{\rA}}  & \qw \poloFantasmaCn{\rA} &\qw } \end{aligned}     \end{equation}
one has 
\begin{equation}\begin{aligned}\Qcircuit @C=1em @R=.7em @!R { & \multiprepareC{1}{\Psi'}    & \qw \poloFantasmaCn{\rA} &  \qw   \\  & \pureghost{\Psi'}    & \qw \poloFantasmaCn{\rB}  &   \qw }\end{aligned}~=\!\!\!\! \begin{aligned}\Qcircuit @C=1em @R=.7em @!R { & \multiprepareC{1}{\Psi}    & \qw \poloFantasmaCn{\rA} &  \qw &\qw &\qw   \\  & \pureghost{\Psi}    & \qw \poloFantasmaCn{\rB}  &   \gate{\cU} &\qw \poloFantasmaCn{\rB} &\qw }\end{aligned}~ ,\end{equation}
and 
\begin{equation}\begin{aligned}\Qcircuit @C=1em @R=.7em @!R {     & \qw \poloFantasmaCn{\rA} &  \measureD{e}   }\end{aligned} ~=~ \begin{aligned}\Qcircuit @C=1em @R=.7em @!R {   & \qw \poloFantasmaCn{\rA} &    \gate{\cU} &\qw \poloFantasmaCn{\rA} &\measureD{e'} }\end{aligned}~ ,\end{equation}
for some reversible transformation $\map U$.   
\end{defn}

In a completely general theory, there may be different ways to discard a system, corresponding to different deterministic effects.  The deterministic effect is unique in {\em causal} theories, that is, theories where no signal can be sent from the future to the past  \cite{Chiribella-purification}.

\subsection{Finiteness, closure, and convexity}  
In this paper we will make three standing assumptions. The first assumption is that our OPT describes {\em finite systems}, i.e.\ systems for which
the state space is finite-dimensional.    Operationally, this means that the state of each system is uniquely determined by the statistics of a finite number of finite-outcome measurements. 

Our second assumption is that the space of transformations (and the spaces of states and effects, in particular) is  closed under limits. Physically, this expresses the fact  that a limit of operational procedures is itself an operational procedure, whereby the target transformation can be implemented with arbitrary accuracy. Mathematically,   a sequence of transformations $ \left\{\mathcal{T}_n\right\}_{n\in\mathbb{N}} $ from $ \rA $ to $ \rB $ converges to the transformation $ \mathcal{T}\in\mathsf{Transf}\left( \rA,\rB\right)  $ if for every reference system $\rR$, every state $\rho  \in \St \left( \rA\otimes \rR\right) $, and every effect $E   \in  \Eff \left( \rB\otimes \rR\right) $   the probabilities $\left( E\middle|   \map T_n \otimes \map I_\rR  \middle|\rho\right) $ converge to the probability $\left( E  \middle|  \map T \otimes \map I_\rR  \middle| \rho\right) $. 

The third standing assumption  made throughout the paper is that the  space of transformations $\mathsf{Transf}\left( \rA,\rB\right) $ is convex  for every $\rA$ and $\rB$. Mathematically, this means that one has the implication  
\begin{align}
\nonumber &\map T,\map S \in  \mathsf{Transf}\left( \rA,\rB\right),\: p  \in  [0,1] \\
& \quad \Longrightarrow  \qquad  p \map T  +   \left( 1-p\right)   \map S  \in \mathsf{Transf}\left( \rA,\rB\right).
\end{align} 
Physically, this means that the experimenter can  perform arbitrary randomised operations. 
Note that convexity is a natural assumption   in every non-deterministic theory:  provided that {\em some} experiment yields random outcomes, one can always repeat that experiment many times and approximate every probability distribution \cite{Chiribella-purification}.  Then, the closure assumption guarantees that the limit probability distribution is also achievable within the theory.   Thanks to this fact,  the experimenter can  perform arbitrary  randomised tests.

\section{Theories of systems with constraints}\label{sec:constraints}  

The language of general probabilistic theories is largely interpretation-independent.  As such, it has the flexibility to model very different physical scenarios, or even to model different fragments of the same physical theory.  In this paper we use the framework to model scenarios of microcanonical thermodynamics, where the   systems under consideration have a well-defined energy.  More generally, the microcanonical approach can be applied to systems with additional constraints---e.g.\ to systems of particles confined in a given volume, or constrained to have a fixed value of the angular momentum.    In these scenarios, the microcanonical state is interpreted as  the state of ``minimum information'' compatible with the constraints.   In this section we outline how the OPT framework can be used to describe physical systems subject to constraints.

\subsection{Constrained  systems in quantum theory}   Before delving into general theories, it may help to analyse the example of  quantum theory.  Let us consider first the case of a quantum system  $\rS$ constrained to a fixed value of the energy.   
   The constraint is implemented by specifying the system's Hamiltonian $H_\rS$ and by restricting the allowed states to (mixtures of) eigenstates of the Hamiltonian for a fixed eigenvalue, say $E$.   
   The quantum states compatible with the constraint are the density matrices $\rho$ satisfying  the condition 
    \begin{align}\label{eneconstr}
    P_{E}\rho P_{E}  =  \rho ,
   \end{align} 
    where $P_{E}$ is the projector on the eigenspace of $H_\rS$ with eigenvalue $E$. 
   
   For example, the system $\rS$ could be an electron in a  hydrogen atom, in the absence of external fields.   In general, the  basis states of the electron are labelled as  $\ket{n,l,m,m_s}$, where  $n,l,m$, and $m_s$ are  the  principal, orbital, magnetic, and spin quantum number respectively.       The electron may be constrained to the lowest energy shell, corresponding to $n=1$ and $l=m=0$.  In this case,  the allowed states are contained in a two-dimensional subspace, spanned by the ``spin-up'' and ``spin-down'' states  $\ket{n=1, l=  0, m=0,  m_s=  1/2}$ and $\ket{n=1,  l=0,  m=0, m_s = -1/2}$,  
  Under this restriction, the electron's spin can be regarded as an {\em effective qubit}. 
  
 Constraints other than  energy preservation can be treated in a similar way.  
 We consider constraints of the form  
 \begin{align}\label{quantumlinear}
 \map L \left( \rho\right)  =  0  ,
 \end{align}
 where $\map L$ is a linear map on the state space of the system. This form is suggested by  Eq.~\eqref{eneconstr}, where  the linear map is $\map L \left( \cdot\right)    =    P_E \left(  \cdot \right)  P_E     -  \map I_\rS  \left( \cdot\right) $,  $\map I_\rS$ being the identity channel on system $\rS$.   
 
 Given a set of constraints $\left\{\map L_i\right\}$, one can define an {\em effective system}, whose states are the density matrices  $\rho$ satisfying the conditions $\map L_i  \left( \rho\right)   =  0$ for every $i$.    The constrained quantum system can be denoted as   $\rA   :  =   \left(  \rS_\rA , \left\{   \map L_{\rA i}  \right\}_{i=1}^k \right) $,  where $\rS_\rA$ is the original system and $\left\{\map L_{\rA i}  \right\}_{i=1}^k$ are the   linear maps representing  the constraints.   The physical transformations of the effective system are the physical transformations of $\rS$ that send every input state satisfying the constraint into an output state satisfying the constraints.   For the energy constraint~\eqref{eneconstr}, this means that the transformations of the effective systems should be energy-preserving.   
  
 A  familiar example of effective system is the polarisation of a single photon. At the fundamental level, the single photon is just  an excitation of  the electromagnetic field, e.g.\ corresponding to the wave vector $\st k$.     One can regard the single photon as an effective system by restricting the attention to the two-dimensional space spanned by the states $\ket{\st k, H,  1} $ and $\ket{\st k, V,  1}$,  corresponding to vertical and horizontal polarisation, respectively.    In this case, we can see two constraints working together: a constraint on the wave vector and a constraint on the energy of the field.  
 Note that the  effective description in terms of single photons is accurate  only as long as the dynamics of the field is confined into the ``single-photon subspace with wave vector $\st k$''.

So far, we have defined effective systems at the single-system level.   An important question is how to define the composition of effective systems.   
Consider two effective systems  $\rA  =  \left(  \rS_\rA , \left\{   \map L_{\rA i}  \right\}_{i=1}^k \right)$    and  $\rB  =  \left(  \rS_\rB , \left\{   \map L_{\rB j}  \right\}_{j=1}^l \right)$, where $\rS_\rA$ and $\rS_\rB$ are the original, unconstrained systems.      A natural way to define the effective composite system $\rA\otimes \rB$ is to select the states of the unconstrained composite system $\rS_\rA  \otimes \rS_\rB$ that satisfy both constraints---i.e.\ to select the density matrices $\rho$ such that 
\begin{align}
\nonumber \left(\map L_{\rA  i}\otimes \map I_{\rS_{\rB}}\right)  \left( \rho\right)  
  &=    0      \qquad \forall i \in \left\{1,\dots, k\right\}  \\
\left( \map I_{\rS_{\rB}}  \otimes  \map L_{\rB j}  \right)  \left( \rho\right) 
  &=    0      \qquad \forall j \in \left\{1,\dots, l\right\}   . \label{composeconstraint}
  \end{align} 

When the effective systems $\rA$ and $\rB$ result from an energy constraint,    the effective system $\rA\otimes \rB$ describes  a system consisting of two parts, each of which with its own, well-defined  energy.     In this case, the constraints~\eqref{composeconstraint} can be summarised in a single equation, namely  
\begin{align}\label{fixedlocal}
\left( P_{E_\rA}  \otimes  Q_{E_\rB}\right)   \rho  \left( P_{E_\rA}  \otimes  Q_{E_\rB}\right)    =   \rho ,
\end{align}
where $  E_\rA$ and $E_\rB$ are the energies of the two local systems, and  $P_{E_\rA}$ and $Q_{E_\rB}$ are the projectors on the corresponding eigenspaces.

One might be tempted to define the composite system  $\rA\otimes \rB$ in a different way, without imposing that each individual part has a definite energy.  Indeed, one could imagine that, when the two systems $\rS_\rA$ and $\rS_\rB$ are brought into contact, they start exchanging energy, with the only constraint that the {\em total} energy has to remain constant.    The resulting states would be density matrices that satisfy the (generally) weaker condition 
\begin{align}\label{totalE}
\Pi_{E_\rA +  E_\rB}   \rho   \Pi_{E_\rA  +  E_\rB}   =  \rho   ,
\end{align} 
where $\Pi_{E_\rA+E_\rB}$ is the projector on the eigenspace of $H_{\rS_\rA}  + H_{\rS_\rB}$ with eigenvalue $E_\rA  + E_\rB$.     The reason why we do not make such a choice can be illustrated with a simple example.  Suppose that $\rS_\rA$ and $\rS_\rB$ are  two spatial modes of the electromagnetic field, with wave vectors $\st k_\rA$ and $\st k_\rB$, respectively.  Systems $\rA$ and $\rB$ could be  single photons, i.e.\ effective systems corresponding to states of the field in the first excited level.    Now, if the energy of the two modes is the same, an energy-preserving evolution could transform the  initial state $\ket{\st k_\rA  ,  H,  1}  \ket{\st k_\rB,  H,1}$ into the state
\begin{align}
\ket{\Psi}  =  \frac{ \ket{\st k_\rA,  H,   2}  \ket{\st k_\rB ,  H, 0} + \ket{\st k_\rA,  H,   0}  \ket{\st k_\rB ,  H, 2}  }{\sqrt{2}}  . 
\end{align}
 States of this kind cannot be interpreted as states of two single photons.   Note that, instead, the constraint~\eqref{composeconstraint} correctly identifies the correct set of states---including, among others, entangled states such as the Bell state
\begin{align}
\ket{\Phi^+}  =  \frac{\ket{\st k_\rA,  H,   1}   \ket{\st k_\rB ,  H, 1} + \ket{\st k_\rA,  V,   1}   \ket{\st k_\rB ,  V, 1}  }{\sqrt 2} . 
\end{align}
  
Motivated by this and by similar examples, we reserve the notation $\rA\otimes \rB$ for effective systems defined by the constraint~\eqref{composeconstraint}.    Other effective systems, like the system defined by the constraint~\eqref{totalE}, can be treated in our framework, but will be  regarded as different from the product system $\rA\otimes \rB$.

\subsection{Constrained systems in general theories}  

The construction outlined in the quantum case can be easily extended to arbitrary physical theories.  
A constraint for system $\rS$ can be defined as an element  $\map L$ of the real vector space $\Transf_\R\left(   \rS\right) $ spanned by the physical transformations in  $\Transf \left( \rS\right) $.   The constraint is satisfied by the states $\rho$ such that 
\begin{align}
\map L  \left(  \rho\right)    = 0   .  
\end{align}   
For a given set of constraints  $\left\{ \map L_i\right\}_{i=1}^k$, one can define an {\em effective system} $\rA  :  =  \left( \rS  ,  \map L_1, \dots, \map L_k\right) $.       The states of the effective systems are defined as  
\begin{align}
\St \left( \rA\right)   =     \left\{ \rho  \in  \St \left( \rS\right)  \,\middle| \,  \map L_i  \left( \rho\right)     =  0  \, , i = 1,\dots, k  \right\}  .
\end{align}  

The transformations of the effective system $\rA$ are those transformations of  $\rS$ that send states of $\rA$ to states of $\rA$.     The measurements on $\rA$ are just the measurements on $\rS$, restricted to the states in $\St\left( \rA\right) $. 

For two effective systems,    $\rA  =  \left(  \rS_\rA , \left\{   \map L_{\rA i}  \right\}_{i=1}^k \right)$ and  $\rB  =  \left(  \rS_\rB , \left\{   \map L_{\rB j }  \right\}_{j=1}^l \right)$, we define the composite system $\rA\otimes \rB$ to be the effective system  
\begin{align}\label{compose}
\rA\otimes \rB   :  =  \left(  \rS_\rA \otimes \rS_\rB  ,    \left\{   \map L_{\rA i}  \otimes \map I_\rB \right\}_{i=1}^k \cup  \left\{   \map I_\rA\otimes \map L_{\rB j }  \right\}_{j=1}^l  \right) .
\end{align}  
  This definition is consistent with interpretation of system $\rA\otimes \rB$ as a composite system made of two, independently addressable parts $\rA$ and $\rB$. For example,  local measurements on the one side of a bipartite state of $\rA\otimes \rB$ induce states of the correct system  (either $\rA$ or $\rB$) on the other side.    

\subsection{Effective theories}   
 Given a theory and a set of constraints composed as in Eq.~\eqref{compose}, one can build a new {\em effective theory}, which consists only of effective systems.  For example,  one can build an effective theory where every system has definite energy, and where every composite systems consist of subsystems with definite energy.   For a given system $\rA$ in such a theory, all the states in    $\St \left( \rA\right) $ have---{\em by fiat}---the same energy.    Likewise, all  the transformations  in $\Transf \left( \rA\right) $ will be---{\em by fiat}---energy-preserving.    For every pair of systems $\rA$ and $\rB$,  the composite system $\rA\otimes \rB$ consists  of two parts, each of which with its own, well-defined  energy.   The joint transformations in $\Transf \left( \rA\otimes  \rB\right) $ will be interpreted as operations that preserve the energy of the first part   {\em and} the energy of the second part.   
 
One benefit of the effective picture is that one does not need to  specify the constraints---in principle, {\em every} linear constraint can fit into the framework. In this way, we can circumvent the thorny issue of defining the notion of  Hamiltonian in general probabilistic theories (cf.\ Ref.~\cite{TowardsThermo}): in the effective description,  we can simply regard each effective system as a system with trivial Hamiltonian, which assigns the same energy to all states of the system.     

 \section{The microcanonical framework}\label{sec:microcanonical requirements}

  In this section we build a microcanonical framework for general physical theories.  We will adopt the effective description, wherein  every system is interpreted as the result of a constraint---typically, but not essentially,  a constraint on the energy.  

\subsection{The principle of equal {\em a priori} probabilities}\label{sub:equalapriori} 

The starting point of the microcanonical approach is the principle of equal {\em a priori} probabilities, stating that one should assign the same probability to all the microstates of the system compatible with a given macrostate.  In our language, the ``microstates''  are the deterministic pure states, representing those preparations of the system that are both deterministic and maximally fine-grained.  The ``macrostate'' is specified by a constraint, such as the constraint of fixed energy.  The principle of equal {\em a priori} probabilities states that  the system should be described by a uniform mixture of all deterministic pure states satisfying the constraint. For example, the microcanonical state of a (finite-dimensional) quantum system at energy $E$ is described by the density matrix  
 \begin{align}\label{qmicro1}
 \chi_E : =   \int_{\set S_E}   p_E\left( \d \psi\right)           \ket{\psi}\bra{\psi}  ,
  \end{align} 
  where  $\set S_E$ is the  manifold of pure states in the eigenspace of the system's Hamiltonian corresponding to the  eigenvalue $E$,  and  $p_E  (\d \psi)$ is the uniform probability distribution over $\set S_E$. 
  In the effective picture,    the microcanonical state is nothing but  the {\em maximally mixed state}  
  \begin{align}\label{qmicro2}
  \chi_\rA   :=  \int  \d \psi  \,  |\psi\>\<\psi|  ,
  \end{align} 
where $\d \psi$ is the uniform probability distribution over the pure states of the system.  

A traditional problem in the foundations of statistical mechanics is to determine the conditions under which the principle of equal {\em a priori} probabilities holds.  Here we will not delve into this problem, which involves a great deal of detail about the  physics of  the system and of  its dynamics.   Instead, we will focus on the general conditions that must be satisfied in order to formulate the principle of equal {\em a priori} probabilities in  physical  theories  other than classical and quantum mechanics.

\begin{figure}[h!]
\centering
\subfigure[]
{\label{fig:semicerchio}
\includegraphics[width=0.4\linewidth]{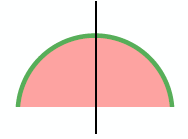}
}
\subfigure[]
{\label{fig:cerchio}
\includegraphics[width=0.4\linewidth]{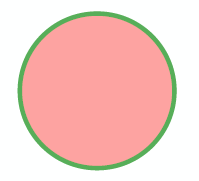}
}

\caption{ Two different sets of deterministic states.    For the set  in figure  \ref{fig:semicerchio}, the pure states form a half-circle (in green).   Due to the limited symmetry of the state space, there is no canonical notion of equal {\em a priori} probability on the manifold of pure states.  For the set in figure \ref{fig:cerchio}, the pure states form a full circle and the notion of uniform probability distribution is uniquely defined.  }
\end{figure}

In   general probabilistic theories,  the key problem  is to define what we mean by  ``equal {\em a priori} probabilities''. 
   In quantum mechanics, there is a canonical choice:  the unitarily invariant probability distribution on the pure states of the system.      The obvious extension to general theories is to consider  the probability distributions that are invariant under all reversible transformations.     
  The problem is, however, that there may be more than one invariant probability distribution.  This point is illustrated in the following example: 

\begin{example}  Consider a     toy theory where the space of  the deterministic states of one of the systems is  a half-disk in the two-dimensional plane, as  in Figure  \ref{fig:semicerchio}.      For this system, the pure states are the states on the half-circle (in green in the figure) and can be parametrised with a polar angle $\theta$ between 0 and  $\pi$.   Now, the reversible transformations send deterministic states into deterministic states and, therefore, must be symmetry transformations of the state space.  For the half-disk, the only symmetry  transformations are the identity transformation and the reflection around the symmetry axis (in black in the figure).    Hence, every probability distribution  that assigns the same probability distribution to the points $\theta$ and $\pi-\theta$ is guaranteed to be invariant under reversible transformations.    This means that the notion of  ``equal {\em a priori} probabilities'' is not uniquely defined.      The situation would be different if the state space of the system were a full disk, as illustrated in Figure \ref{fig:cerchio}.     In this case, every rotation of the disk could be  (at least in principle) a reversible transformation of the system.    The invariant probability distribution would be unique and given by  the probability density $p\left( \theta\right) = \frac{1}{2\pi}$. 

 Note that in the above examples we only specified the state space and the transformations of a single system, without  giving the full-blown OPT.  It is easy to see that  such a theory does indeed exist.  In general, one can always build a {\em ``Minimal OPT''} that includes a given system with a given state space, a given set of  transformations, and a given set of measurements.  The construction was shown in  Ref.~\cite[example 2]{chiribella2016bridging}.  In the  Minimal OPT, the composite systems are made of many copies of the given system, and their allowed  states, transformations, and operations are (mixtures of) product states, product transformations, and product operations.   
\end{example}

The above example shows that  there exist probabilistic theories where the notion of ``equal {\em a priori} probabilities'' on pure states is not uniquely defined. In order to formulate the principle of equal {\em a priori} probabilities, we put forward the following requirement:  
\begin{req}\label{req:uniquep}
For every (finite) system there exists a unique invariant probability distribution  on the deterministic pure states.   
\end{req}

This requirement is far from trivial.  In fact, it is equivalent to an important property, independently considered in the literature on the axiomatisation of quantum theory \cite{Hardy-informational-1,Brukner,masanes,Barnum-interference}: 
\begin{thm}\label{theo:uniquep}
For every finite system $\rA$, the following are equivalent: 
\begin{enumerate}
\item  There exists a unique invariant probability distribution on the  deterministic pure states of system $\rA$.
\item  
Every deterministic pure state of system $\rA$  can be obtained from every other deterministic pure state of the same system through a reversible transformation. 
\end{enumerate}
\end{thm}
\begin{proof}The  idea of the proof is that the set of deterministic pure states of system $\rA$ can be decomposed into a disjoint union of orbits generated by the group of reversible transformations. More formally, for every two (deterministic) pure states  $\alpha$ and $\alpha'$, one can define the equivalence relation $\alpha  \sim_{\rm rev} \alpha'$  if $\alpha'  = \map U \alpha$ for some reversible transformation $\map U$.   In this way, the set of deterministic pure states is partitioned into equivalence classes, known as homogeneous spaces. Moreover,  each homogeneous space is a closed set,  because the group of reversible transformations is closed \cite{Chiribella-purification} and has a finite-dimensional representation on the state space of system $\rA$. 
Now, one can define  an invariant probability distribution for every equivalence class. Indeed, it is enough to define the invariant measure on the pure states induced by the invariant measure on the group of reversible transformations.  Hence, the condition that there is only one invariant probability distribution implies that there must be only one equivalence class. In other words, every two pure states are connected by a reversible transformation. 

Conversely, if there is only one equivalence class for the relation $\sim_{\rm rev}$ there is only one invariant probability distribution. This is because the normalised invariant measure on a homogeneous space is uniquely defined.  \end{proof}      
  
The mutual convertibility of pure states under reversible transformations was  introduced by Hardy \cite{Hardy-informational-1} as an axiom for the derivation of quantum theory and has been assumed, either directly or indirectly, in all the recent derivations inspired by quantum information theory (see Refs. \cite{Brukner,masanes,Chiribella-informational,Hardy-informational-2} and the contributed volume \cite{chiribella2016quantum} for an overview).   Theorem~\ref{theo:uniquep} provides one more motivation for  the convertibility of pure states, identified as the necessary condition for the formulation of the principle of equal {\em a priori} probabilities. 

\subsection{The microcanonical state}

Every theory satisfying requirement~\ref{req:uniquep} has a canonical notion of ``uniform distribution over the pure states of the system''.  We can then apply the principle of equal {\em a priori} probabilities and define the microcanonical state as the uniform mixture
\begin{align}\label{chidef}
\chi_\rA   :  =  \int p_\rA(  \d \psi) ~   \psi  \,  , 
\end{align}  
where $p_\rA (\d \psi)$ the invariant probability distribution over the deterministic pure states of system $\rA$.  
 
The convexity of the state space guarantees that the microcanonical state is indeed a state.  Moreover, since the state space is finite-dimensional, it is possible to replace the integral  in Eq.~\eqref{chidef}  with a finite sum.    This means that the microcanonical state can  (in principle) be generated by picking deterministic pure states at random from a finite set.

The microcanonical state has two important properties, proved in appendix~\ref{app:propertieschi}: 
\begin{enumerate}
\item it is {\em invariant  under arbitrary  reversible dynamics} of the effective system;
\item it {\em can  be generated from every other deterministic pure state of the effective system  through a random reversible dynamics}.   
\end{enumerate}
Property 1 expresses the fact  that the microcanonical state is an {\em  equilibrium state}, in the sense that it does not evolve under any of the reversible dynamics compatible with the constraints.    Note that the notion of equilibrium here is different from  the notion of thermal equilibrium, which refers to interactions with an external bath.  
  Instead of thermal equilibrium, we consider  here a  {\em dynamical equilibrium}, consisting in the fact that  the probability assignments made by the microcanonical state are stable under all possible evolutions of the system.  
  
Property 2 refers to the fact that the system can---at least in principle---be {\em brought} to equilibrium.  Physically, we can imagine a situation where the experimenter has no control on the system's preparation, but has control on the system's dynamics  through some classical control fields.  In this picture, Property 2  guarantees that the experimenter can prepare the microcanonical state by drawing at random the parameters of her control fields.    Further along this line, one can also imagine scenarios where the randomisation occurs naturally as a result of fluctuations of the fields.   Property 2 is important from the resource-theoretic approach, where  the microcanonical state is often regarded as  free, or ``easy to prepare''.

 \subsection{Composition of microcanonical states} 
  
  At the level of single systems,  requirement~\ref{req:uniquep}  guarantees the existence of a microcanonical state.  
   But how does the microcanonical state behave under the composition of systems?  
   Traditionally, this question is not addressed in textbook presentations, where the microcanonical state is associated to {\em isolated} systems, i.e., systems that do not interact with other systems. 
  From the operational point of view, however, it is natural to consider scenarios where the experimenter has  more than  one system at her disposal.     
  
  Composition is especially important in the context  of resource theories, where it is natural  to ask how resources interact when combined together.    To illustrate this point, it is useful  to consider the quantum resource theory of noisy operations    \cite{Local-Information,Horodecki-Oppenheim,Nicole}.   There, the  microcanonical states are treated as  free.  
       Since the experimenter can generate the microcanonical states $\chi_\rA$ and $\chi_\rB$ at no cost, then she can generate the product state $\chi_\rA\otimes \chi_\rB$ at no cost, too. If we insist that  the microcanonical states are the {\em only} free states in the resource theory of noisy operations,  
          the product state  $\chi_\rA\otimes \chi_\rB$  must  be the microcanonical state of the composite system $\rA\otimes \rB$---in formula,  
          \begin{align}\label{microprod}
          \chi_\rA\otimes \chi_\rB   =   \chi_{\mathrm{AB}}.
          \end{align}
        Eq.~\eqref{microprod} is consistent with the intuitive interpretation of  the microcanonical state as ``the state of minimum information compatibly with the constraints''. 
   Indeed,  Eq.~\eqref{microprod}  amounts to saying that, if one has minimum information on the parts of a system, then one has minimum information about the whole.    This is indeed the case in quantum theory, where the product of two maximally mixed states is maximally mixed.  Recall that here we are dealing with effective systems,  which exist only as long the corresponding constraints are enforced.   For energy constraints, the composite of two effective systems $\rA$ and $\rB$  is defined as a system consisting of two parts, each constrained to a specific value of the energy.  Consistently with this interpretation, the microcanonical state of system $\rA\otimes \rB$  is the ``maximally mixed state'' in the manifold of quantum states with  fixed local energies, as defined in  Eq.~\eqref{fixedlocal}.   
         
 Following the example of quantum mechanics, we require that  minimum information about the parts imply minimum information about the whole:  
 \begin{req}\label{req:productchi}
The microcanonical state of a composite system is the product of the microcanonical states of its components.  In formula: 
\begin{equation}
\chi_{\mathrm{AB}}  =  \chi_{\mathrm{A}}\otimes\chi_{\mathrm{B}}  ,\label{eq:prodchi}
\end{equation}
for every pair of effective systems $\mathrm{A}$ and $\mathrm{B}$.
\end{req}
We call Eq.~\eqref{eq:prodchi} the \emph{condition of informational equilibrium}. Note that, again,  here we are not referring to thermal equilibrium between the two subsystems.  This is clear from the fact that  we do not allow  an energy flow between the two  systems $\rA$ and $\rB$. Instead, we allow a flow of information, implemented by the joint dynamics of the  composite  system $\rA\otimes \rB$. 


It is natural to ask which physical principles guarantee the condition of informational  equilibrium. One such principle is Local Tomography  \cite{Hardy-informational-1,Barrett,Chiribella-purification}, namely the requirement that the state of multipartite systems be determined by the joint statistics of local measurements. 
  However,  Local Tomography is not necessary for informational equilibrium.  For example, quantum theory on real Hilbert spaces violates Local Tomography, but still satisfies the condition of informational equilibrium.    In this paper, we will \emph{not} assume Local Tomography in our set of physical principles. Nevertheless, our principles will guarantee the validity of the condition of informational equilibrium.   

\subsection{Microcanonical theories}  

We are now ready to extend the microcanonical framework from quantum and classical theory to general physical theories.    

\begin{defn}
An operational-probabilistic  theory, interpreted as a theory  of effective systems,  is {\em   microcanonical} if requirements~\ref{req:uniquep} and \ref{req:productchi} are satisfied.  
\end{defn}  

Physically, a microcanonical theory is a theory where  {\em i)} every system has a well-defined notion of uniform mixture of all pure states, and {\em ii)} uniform mixtures are   stable under parallel composition of systems. 
  Microcanonical theories provide the foundation for the definition of three important resource theories,  analysed in the following sections.  
  
\section{Three resource theories\label{sec:Resource-theories-of}}

In this section we study three different notions of state convertibility in the microcanonical theories.  We adopt the resource-theoretic framework of Refs.  \cite{Resource-theories,Resource-monoid},   where one fixes  a  set of {\em free operations}, closed under sequential and parallel composition.  
A basic question in the resource-theoretic framework  is whether a given state  $\rho$
can be transformed into another state $\sigma$ by means of free operations. When
this is possible, $\rho$ is regarded as ``more resourceful''
than  $\sigma$, denoted as $\rho\succeq_{\set F} \sigma$, where $\set F$ is the set of free operations. 
Mathematically, the relation $\succeq_{\set F}$ is  a preorder on the states. 


In the following  we define three resource theories and their corresponding preorders.  


\subsection{The RaRe Resource Theory}
  
Our first resource theory is based on the notion of random reversible channel \cite{Chiribella-Scandolo-entanglement}:
\begin{defn}
A \emph{random reversible (RaRe) channel}  on system $\rA$ is a channel $\mathcal{R}$
of the form $\mathcal{R}=\sum_{i}p_{i}\mathcal{U}_{i}$, where $\left\{ p_{i}\right\} $
is a probability distribution  and, for every $i$,  $\mathcal{U}_{i}$ is a reversible
channel on system $\rA$. 
\end{defn}
Physically,  RaRe channels are the operations that can be implemented with limited control over the reversible dynamics of the system.   
 Mathematically, it is immediate to check that RaRe channels have all the properties
required of free operations: the identity channel is RaRe, the sequential composition of two RaRe channels is a  RaRe
channel, and so is the parallel composition.  We call the resulting resource theory the \emph{RaRe Resource
Theory} and we denote by $\succeq_{\mathsf{RaRe}}$ the corresponding preorder. 

Note that the RaRe Resource Theory  can be formulated in  every OPT, \emph{even in OPTs that do not satisfy   requirements~\ref{req:uniquep} and \ref{req:productchi}}.  
Such generality, however,  comes at a price: the RaRe Resource Theory has no free states. This is because states are operations
with trivial input, while the only free operations in the RaRe theory
are transformations where the input and the output coincide.  

Despite not having free states, the RaRe Resource Theory can have {\em minimally resourceful states}, defined as follows
\begin{defn}
In a resource theory with free operations $\set F$, a state $\rho$ is  {\em minimally  resourceful} if  the condition $\rho  \succeq_{\set F} \sigma$ implies $\sigma  =  \rho$.  
\end{defn}
In the RaRe Resource Theory, minimally resourceful states are easy to characterise:  
\begin{prop}\label{prop:minimallyresourceful} A state is minimally resourceful in the RaRe Resource Theory if and only if it is invariant under the action of reversible 
transformations.
\end{prop}
\begin{proof}  By definition, one has $\rho  \succeq_{\rm RaRe}  \map U  \rho$ for every state $\rho$ and for every reversible transformation  $\map U$.     If $\rho$ is minimally resourceful, one must  have $\rho  =  \map U  \rho$. Hence, $\rho$ must be invariant under arbitrary reversible transformations. 
	
	 Conversely, suppose that $\rho $ is invariant, and that $\rho  \succeq_{\rm RaRe}  \sigma$.   By definition, this means that $\sigma  =  \map R  \rho  $, for some RaRe channel $\map R$.  Since $\rho$ is invariant, it must satisfy the relation $\map R \rho  = \rho$.    Hence, $\sigma  = \rho$.  \end{proof}   

For theories satisfying requirement~\ref{req:uniquep},   proposition~\ref{prop:minimallyresourceful}  implies that the microcanonical state is  minimally resourceful: indeed, we know that the microcanonical state is invariant under reversible transformations.

\subsection{The  Noisy Resource Theory}

While the RaRe Resource Theory can be defined in every OPT, we  now discuss a second resource theory that  can only be defined in  physical theories satisfying requirements~\ref{req:uniquep} and \ref{req:productchi}. 
 In this  resource theory, the free operations are  generated by letting the system interact  with ancillas in the microcanonical state.   These operations, usually called   ``noisy'' \cite{Local-Information,Horodecki-Oppenheim,Nicole},
  are defined as follows:
\begin{defn}\label{def:basicnoisy}
A channel $\mathcal{B}$,  from system $\mathrm{A}$
to system $\mathrm{A}'$,  is a  \emph{basic noisy operation}  if it can be decomposed as
\begin{equation}\label{eq:noisy} \begin{aligned}\Qcircuit @C=1em @R=.7em @!R { & & \qw \poloFantasmaCn{\rA} & \gate{\cB} & \qw \poloFantasmaCn{\rA'} &\qw }\end{aligned} ~= \!\!\!\! \begin{aligned}\Qcircuit @C=1em @R=.7em @!R { & & \qw \poloFantasmaCn{\rA} & \multigate{1}{\cU} & \qw \poloFantasmaCn{\rA'} &\qw \\ & \prepareC{\chi} & \qw \poloFantasmaCn{\rE} &\ghost{\cU} & \qw \poloFantasmaCn{\rE'} & \measureD{e} }\end{aligned}~, \end{equation}where
$\mathrm{E}$ and $\mathrm{E}'$ are suitable systems such that $\mathrm{A}\otimes\mathrm{E}\simeq\mathrm{A'}\otimes\mathrm{E}'$, $\mathcal{U}$ is a reversible transformation, and $e$ is a deterministic effect, representing a possible way to discard system $\rE'$. 
\end{defn}

Note that here we only allow reversible transformations, instead of mixtures of reversible transformations.  In principle, one could  consider arbitrary RaRe channels, as in the previous subsection.  The main reason why we stick to the reversible transformations (without randomisation) is that we want to be consistent with the existing  literature \cite{Local-Information,Horodecki-Oppenheim,Nicole}. Note also that  definition~\ref{def:basicnoisy} is interesting {\em per se}, because it does not rely on the availability  of external sources of randomness: instead, all the randomness is accounted for in the preparation of the microcanonical state in the right-hand side of  Eq.~\eqref{eq:noisy}.

Definition~\ref{def:basicnoisy} has a slightly unpleasant aspect:  the set of basic noisy operations is generally not closed.    
  In quantum theory, for example,  there exist counterexamples where the limit of a sequence of basic
noisy operations is not  a basic noisy operation  \cite{Shor}.      It is then convenient to take the closure of  the set of basic noisy operations:
\begin{defn}\label{def:noisy}
A channel $\mathcal N$ is a  \emph{noisy
operation} if it is the limit of a sequence of basic noisy operations  $\left\{\mathcal  B_n\right\}$. 
\end{defn}

The set of noisy operations satisfies all the requirements for being a set of free operations:  the identity is a noisy operation, and the
parallel and sequential composition of two noisy operations are a noisy operations, thanks to the  condition of informational equilibrium~\eqref{eq:prodchi}.
The resource theory where the set of free operations
is the set of noisy operations will be called the \emph{Noisy Resource Theory}. The corresponding preorder on states will be denoted by $\succeq_{\mathsf{Noisy}}$.

\subsection{The Unital  Resource Theory}

 In the third  resource theory, the set of free operations includes
all the operations that transform microcanonical states into microcanonical  states.  The rationale for considering these transformations, called \emph{unital channels}, is their generality: if we insist that  the microcanonical states are the only free states, unital channels are the most general transformations that send free states into free states. In other words, they are the most general operations that do not create resources out of free states

Mathematically, the unital channels are defined as follows: 
\begin{defn}
A channel $\mathcal{D}$ from system $\mathrm{A}$ to system $\mathrm{A}'$
is called \emph{unital} if $\mathcal{D} \chi_{\mathrm{A}} =\chi_{\mathrm{A}'}$.
\end{defn}
Unital channels are the operational generalisation of  doubly
stochastic matrices in classical probability theory \cite{Olkin,Streater,Mendl-Wolf}.

  The set of unital channels enjoys all the properties required of a set of free operations: the identity
is a unital channel, and thanks to the condition of informational equilibrium, the sequential
and parallel composition of unital channels is a  unital channel.  The resource theory where free operations are unital
channels will be called the \emph{Unital Resource Theory}. The corresponding preorder on states will be denoted by $\succeq_{\mathsf{Unital}}$.

\subsection{Containment relations}

Let us highlight the relations between the three sets of   operations defined so far.      First, RaRe channels are examples of unital channels.    This is clear because every RaRe channel  can be decomposed as a mixture of reversible transformations,  each of which preserves the microcanonical state.  Hence, we have  the inclusion  
\begin{equation}\label{rareunital}
\mathsf{RaRe}\subseteq  \mathsf{Unital}  . 
\end{equation}   In classical probability theory,  the inclusion is actually an equality,    as a consequence of Birkhoff's theorem \cite{Birkhoff,Olkin}.   Remarkably,   in quantum theory there exist unital channels that are not random
unitary, meaning that the inclusion~\eqref{rareunital} is generally strict.  The simplest example is due by Landau and Streater \cite{Streater}: for a quantum particle of spin $j$, they  defined the map 
\begin{align}
\map D_j \left( \cdot\right)   =  \frac{J_x\left( \cdot \right) J_x   +  J_y\left( \cdot \right) J_y  +  J_z\left( \cdot\right)  J_z} {j\left( j+1\right) }   ,  
\end{align} 
where $J_x, J_y,  J_z$ are the three components of the spin operator. It is easy to see that the map $\map D_j$ is trace-preserving and identity-preserving---that is, it is a unital channel.  On the other hand, Landau and Streater showed that the map $\map D_j$ cannot be decomposed as a mixture of unitary channels unless $j=1/2$ \cite{Streater}.  

We have seen that all RaRe channels are unital.  Noisy operations are also unital, as shown by the following     
\begin{prop}
Every noisy operation is  unital. 
\end{prop}
\begin{proof}
Suppose that $\map B$ is a basic noisy operation,  decomposed as in Eq.~\eqref{eq:noisy}.  Then, one has
\begin{align} 
 \nonumber\begin{aligned}\Qcircuit @C=1em @R=.7em @!R { &\prepareC{\chi}  & \qw \poloFantasmaCn{\rA} & \gate{\cB} & \qw \poloFantasmaCn{\rA'} &\qw }\end{aligned} ~&= \!\!\!\! \begin{aligned}\Qcircuit @C=1em @R=.7em @!R { &\prepareC{\chi} & \qw \poloFantasmaCn{\rA} & \multigate{1}{\cU} & \qw \poloFantasmaCn{\rA'} &\qw \\ & \prepareC{\chi} & \qw \poloFantasmaCn{\rE} &\ghost{\cU} & \qw \poloFantasmaCn{\rE'} & \measureD{e} }\end{aligned} \\
  \nonumber& = \!\!\!\! \begin{aligned}\Qcircuit @C=1em @R=.7em @!R { &\multiprepareC{1}{\chi} & \qw \poloFantasmaCn{\rA} & \multigate{1}{\cU} & \qw \poloFantasmaCn{\rA'} &\qw \\ & \pureghost{\chi} & \qw \poloFantasmaCn{\rE} &\ghost{\cU} & \qw \poloFantasmaCn{\rE'} & \measureD{e} }
  \end{aligned}\\  
  \nonumber& = \!\!\!\! \begin{aligned}\Qcircuit @C=1em @R=.7em @!R { &\multiprepareC{1}{\chi} &    \qw \poloFantasmaCn{\rA'} &\qw \\ & \pureghost{\chi} &   \qw \poloFantasmaCn{\rE'} & \measureD{e} }
  \end{aligned}\\   \nonumber
  & = \!\!\!\! \begin{aligned}\Qcircuit @C=1em @R=.7em @!R { &\prepareC{\chi}  & \qw \poloFantasmaCn{\rA'} &\qw \\ & \prepareC{\chi}  & \qw \poloFantasmaCn{\rE'} & \measureD{e} }
  \end{aligned}\\   
  &  =  \!\!\!\!  \begin{aligned}\Qcircuit @C=1em @R=.7em @!R { &\prepareC{\chi}  & \qw \poloFantasmaCn{\rA'} &\qw }\end{aligned}
    ~, \end{align}
    having used the  condition of informational equilibrium~\eqref{eq:prodchi}, the invariance of the microcanonical  state $\chi_{\mathrm{AE}}$ under reversible transformations, and the condition $\left( e\middle|\chi_\rE\right) =1$, following from the fact that both $\chi_\rE$ and $e$ are deterministic.   Hence, every basic noisy operation is unital. Since the set of unital channels is closed under limits, all noisy operations are unital. 
\end{proof}

In summary, one has   the inclusion  
\begin{equation}\label{noisyunital}
\mathsf{Noisy}\subseteq\mathsf{Unital} \, . 
\end{equation}
The inclusion is  strict   in quantum theory, where Haagerup and Musat have found examples of unital channels that cannot be realised as noisy operations \cite{Haagerup-Musat}.   

It remains to understand the relation between RaRe channels and noisy operations.   
In quantum theory, the set of noisy operations (strictly) contains the set of
RaRe channels as a proper subset \cite{Shor}. In a generic theory, however,
this containment relation may not hold. As a  counterexample, consider the variant of quantum theory where only local operations are allowed: in this case, RaRe channels are not contained in the set of noisy operations, because all the interactions are trivial.    

The inclusions~\eqref{rareunital} and \eqref{noisyunital} are the most general result
one can derive from the definitions alone.   To go further, we need to introduce  axioms.
  In the next sections, we will introduce
a set of axioms that imply deeper relations between the RaRe, Noisy,
and Unital Resource Theories. In addition, the axioms will imply a  connection with the mathematical theory of majorisation and a connection with the resource theory of entanglement.

%

\section{Four axioms\label{sec:Axioms}}
In this section we review  the four  axioms used in this paper.   These axioms---Causality, Purity Preservation, Pure Sharpness, and Purification---define a special class of theories, which we call \emph{sharp theories with purification}. 

\subsection{Sharp theories with purification}
Sharp theories with purification are defined by the following  four axioms.   The first axiom---Causality---states that no signal can be sent from the future to the past: 
\begin{ax}[Causality \cite{Chiribella-purification,Chiribella-informational,QuantumFromPrinciples,dariano2017quantum}]\label{ax:causality}
The probability that a transformation  occurs in a test is independent of the settings of tests performed on the output. 
\end{ax}
  
The second axiom---Purity Preservation---states  that no information
can leak to the environment when two pure transformations are composed:
\begin{ax}[Purity Preservation \cite{Scandolo14}]
Sequential and parallel compositions of pure transformations yield
pure transformations.
\end{ax}
The third axiom---Pure Sharpness---guarantees that every system
possesses at least one elementary property, in the sense of Piron
\cite{PironBook}:
\begin{ax}[Pure Sharpness \cite{QPL15}]
For every system  there exists at least one pure effect
occurring with unit probability on some state.
\end{ax}
 Axioms 1--3 are satisfied  by both classical and quantum theory.      Our fourth axiom, Purification, characterises all physical
theories admitting a fundamental level of description where all deterministic
processes are pure and reversible.  

\begin{ax}[Purification \cite{Chiribella-purification,Chiribella-informational,QuantumFromPrinciples,dariano2017quantum}]\label{ax:purification}
Every state has a purification.  Purifications are essentially unique, in the sense of definition~\ref{def:uniqueness}.   
\end{ax}
Quantum theory, both on complex and real Hilbert spaces, satisfies Purification.  
 Remarkably, even classical theory can be regarded as a sub-theory  of a  larger physical theory where  Purification is satisfied~\cite{TowardsThermo}.  


\begin{defn}
An OPT is a {\em sharp theory with purification} if it satisfies Axioms~\ref{ax:causality}--\ref{ax:purification}. 
\end{defn}
 In the rest of the section  we will outline  the main  kinematic properties of sharp theories with purification.

\subsection{Well-defined  marginal states}

By definition, sharp theories with purification satisfy Causality, which in turn  is equivalent to the requirement that, for every system
$\mathrm{A}$, there exists a unique deterministic effect $u_{\mathrm{A}}\in \Eff (\rA)$
 (or simply $u$, when no ambiguity can arise) \cite{Chiribella-purification}.   
  The uniqueness of the deterministic effect implies that the marginals of a bipartite state are uniquely defined. 
  For a bipartite state  $\rho\in\St\left( \rA\otimes \rB\right) $, we will denote the marginal on system $\rA$ as 
  \begin{align}
  \mathrm{Tr}_{\mathrm{B}}\left[ \rho_{\mathrm{AB}}\right]   := \!\!\!\!\begin{aligned}\Qcircuit @C=1em @R=.7em @!R { &\multiprepareC{1}{\rho}  & \qw \poloFantasmaCn{\rA} &\qw \\ & \pureghost{\rho}  & \qw \poloFantasmaCn{\rB} & \measureD{u} }  
  \end{aligned}~, \end{align}in analogy with the notation used in quantum theory. 
  
 In a causal theory, it is immediate to see that a state $\rho$ can be prepared deterministically if and only if it is \emph{normalised}, namely 
 \begin{align}
 \Tr \left[ \rho \right]   :=  (u|\rho)   =  1 \, .
 \end{align} We denote  the set of normalised states of system $\rA$ as $\St_1 \left( \rA\right) : = \Det\St (\rA)$.

\subsection{Diagonalisation}  

In sharp theories with purification, 
one can prove that every state can be \emph{diagonalised}, that is, decomposed as a random mixture of perfectly distinguishable pure states.

\begin{thm}[\cite{QPL15,TowardsThermo}] Every normalised state $\rho \in \St_1 \left( \rA\right) $ of every system $\rA$ can be decomposed as 
\begin{equation} 
\rho  =  \sum_{i=1}^r p_i  \alpha_i  ,
\end{equation}
where  $r$ is an integer (called the \emph{rank} of the state),  $p_1\ge p_2\ge \ldots \ge p_r >  0$ are probabilities (called the \emph{eigenvalues}),  and  $\left\{\alpha_i\right\}_{i=1}^r$ is a set of perfectly distinguishable pure states (called the \emph{eigenstates}).
\end{thm}
 It follows from the axioms that the eigenvalues are uniquely defined by the state  (see \cite{TowardsThermo} for the proof). 
  The uniqueness of the spectrum is a non-trivial consequence of the axioms: notably, Refs.~\cite{Krumm-Muller,Krumm-thesis} exhibited  examples of  theories   (other than the sharp theories with purification considered here) where states can be diagonalised, but the same state can have two different  diagonalisations with two different spectra. 
    \subsection{State-effect duality}
    Sharp theories with purification  exhibit  a duality between
normalised pure states and normalised pure effects---a normalised effect being an effect $a$ such that $\left( a\middle|\rho\right) =1$ for some state.  Denoting the set of normalised pure effects by $\Pur\Eff_1 \left( \rA\right) $, the duality reads as follows: 
\begin{prop}[\cite{QPL15}]
\label{prop:duality states-effects}There is a bijective correspondence
between normalised pure states and normalised pure effects. Specifically,
if $\alpha\in\mathsf{PurSt}_{1}\left(\mathrm{A}\right)$, there exists
a unique $\alpha^{\dagger}\in\mathsf{PurEff}_{1}\left(\mathrm{A}\right)$
such that $\left(\alpha^{\dagger}\middle|\alpha\right)=1$.
\end{prop}
Physically, the meaning of the duality is that every pure state can be \emph{certified}  by a (unique) pure effect, which occurs with unit probability only on that particular state.  The duality between pure states and pure effects can be lifted to a duality between maximal sets of perfectly distinguishable pure states and perfectly distinguishing observation-tests, defined as follows:  
\begin{defn}
An observation-test $\left\{a_i\right\}_{i\in \set X} $ is  called \emph{perfectly distinguishing} if there exists a set of states $\left\{  \rho_i\right\}_{i\in\set X}$, such that  $\left( a_i  \middle|  \rho_j\right)     =  \delta_{ij}$ for all $i$ and $j$ in  $\set X$. In this case the states $\left\{  \rho_i\right\}_{i\in\set X}$ are said \emph{perfectly distinguishable}.
\end{defn}

\begin{defn}
A set of perfectly distinguishable states $\left\{  \rho_i\right\}_{i\in\set X}$ is \emph{maximal} if there is \emph{no} state $ \rho_0 $ such that the states $\left\{  \rho_i\right\}_{i\in\set X}\cup\left\lbrace\rho_0 \right\rbrace $ are perfectly distinguishable.
 \end{defn}
A maximal set of perfectly distinguishable pure states will be called \emph{pure maximal set} for short. With this notation, the duality reads  
\begin{prop}[\cite{TowardsThermo}]\label{prop:dagobs}
The pure states $\left\{   \alpha_i \right\}_{i\in \set X}$  are a maximal set  if and only if the  pure effects $\left\{\alpha_i^\dag \right\}_{i\in\set X}$ form a perfectly distinguishing observation-test.  
\end{prop}

As a consequence,  the product of two pure maximal sets  is a pure maximal set for the composite system:   
\begin{prop}\label{prop:information locality}
If $\left\{\alpha_i\right\}_{i=1}^{d_\rA}$ is a pure maximal set  for system $\rA$ and  $\left\{\beta_j\right\}_{j=1}^{d_\rB}$ is a pure maximal set for system $\rB$, then 
$\left\{\alpha_i\otimes \beta_j \right\}_{  i  \in  \left\{  1,\dots, d_\rA\right\}  ,  j \in  \left\{  1,\dots , d_\rB\right\}}$ is a pure maximal set for the composite system $\rA\otimes\rB$. 
\end{prop}
\begin{proof}
By proposition \ref{prop:dagobs},   $\left\{\alpha_i^\dag\right\}_{i=1}^{d_\rA}$  and   $\left\{\beta^\dag_j\right\}_{j=1}^{d_\rB}$ are two observation-tests for systems $\rA$ and $\rB$, respectively.  Now, the product of two-observation tests is an observation-test (physically, corresponding to two measurements performed in parallel). Hence, the product  $\left\{\alpha^\dag_i\otimes \beta^\dag_j \right\}_{  i  \in  \left\{  1,\dots, d_\rA\right\}  ,  j \in  \left\{  1,\dots , d_\rB\right\}}$
  is an observation-test on the composite system $\rA\otimes \rB$.   Moreover, each effect $\alpha_i^{\dagger}\otimes \beta_j^{\dagger}$ is pure, due to Purity Preservation.  Using proposition~\ref{prop:dagobs} again,   we obtain that  $\left\{\alpha_i\otimes \beta_j \right\}_{  i  \in  \left\{  1,\dots, d_\rA\right\}  ,  j \in  \left\{  1,\dots , d_\rB\right\} }$ is a pure maximal set. \end{proof} 
 It is possible to show that all pure maximal sets in a given system have the same cardinality \cite{TowardsThermo}.  
 For a generic system $\rA$, we will denote the cardinality of the maximal sets by $d_\rA$. We will refer to $d_\rA$ as the \emph{dimension} of system $\rA$. We stress that the dimension $d_\rA$ should not be confused with the dimension of the normalised state space $\St_1\left( \rA\right) $: in quantum theory, the dimension  $d_\rA$ is the dimension of the system's Hilbert space, while the dimension of the space of density matrices is $d_\rA^2  -1$.

 
 Proposition~\ref{prop:information locality} shows that the dimension of a composite system is the  product of the dimensions of the components, namely 
 \begin{align}
 d_{\mathrm{AB}}=d_{\mathrm{A}}d_{\mathrm{B}} \, ,
 \end{align} 
 for every  pair of systems $\mathrm{A}$ and $\mathrm{B}$.  This property has been dubbed  \emph{information locality} by Hardy \cite{Hardy-informational-2,hardy2013}.

 \section{Microcanonical thermodynamics in sharp theories with purification}\label{sec:microcanonical}

Here we show that sharp theories with purification satisfy our requirements~\ref{req:uniquep} and \ref{req:productchi} for the construction of the microcanonical framework.  Moreover,  we will show that sharp theories with purification exhibit a simple inclusion  relation between RaRe and  noisy  operations.

\subsection{The microcanonical state}

We start by showing that every sharp theory with purification satisfies requirement~\ref{req:uniquep}, which enables the formulation of the principle of equal {\em a priori} probabilities. 
Thanks to theorem~\ref{theo:uniquep},  we only need to show that every two pure states of the same system are connected by a reversible transformation.  This fact is an immediate consequence of Purification:  
\begin{prop}[\cite{Chiribella-purification}]\label{prop:transitivity}
For every  theory satisfying Purification, for every system $\rA$ in the theory,  and for every pair of deterministic pure states $\alpha$ and $\alpha'$ of system $\rA$, there exists a reversible transformation $\map U$ such that $\alpha'  =  \map  U \alpha$. 
\end{prop}  

Proposition~\ref{prop:transitivity}  guarantees that for every system $\rA$  there exists a unique  probability distribution  $p_\rA  \left( \d \psi\right) $, which is invariant under all reversible dynamics. 
   In turn, the probability distribution  $p_\rA \left( \d \psi\right) $ can be used to define the microcanonical state $\chi_\rA$.   
   
In sharp theories with purification, the microcanonical state enjoys a remarkable property: the state can be decomposed into a uniform mixture of perfectly distinguishable pure states.

\begin{prop}[\cite{TowardsThermo}]
\label{prop:chi invariant} For sharp theory with purification,  every system $\rA$, and every  pure maximal set  $\left\{ \alpha_{i}\right\} _{i=1}^{d_\rA}$ in $\rA$, one has 
the decomposition
\begin{equation}\label{chidecomp}     \chi_\rA  = \frac{1}{d_\rA}\sum_{i=1}^{d_\rA}\alpha_{i}    \, .
\end{equation}  
\end{prop}

In quantum theory, the decomposition  of Eq. (\ref{chidecomp}) is nothing but the expression  
\begin{align}\label{identityoverd}
\chi_\rA  =  \frac{ I_{\rA}} {d_\rA}  ,
\end{align}
where $d_\rA$ is the dimension of the system's Hilbert space, and $I_\rA$ is the  $d_\rA\times d_\rA$ identity matrix.   Recall that here we are interpreting the systems in our theory as  effective systems.  
  For a system with definite energy,   the decomposition of Eq.~\eqref{identityoverd} reads 
\begin{align}\label{chiE}
\chi_E   =  \frac 1 {d_E}   \sum_{i=1}^{d_E}    \ket{E,n}\bra{E,n} ,
\end{align}  
 where $ \left\{   \ket{E,n} \right\}_{n=1}^{d_E}$ is any orthonormal basis for the eigenspace of the Hamiltonian with eigenvalue $E$.  
 
 It is worth noting that Eq.~\eqref{chiE} is often chosen as  the {\em definition} of the microcanonical state in quantum statistical mechanics.   Proposition~\ref{prop:chi invariant} shows that a similar definition is possible in every sharp theory with purification. 
 One may be tempted to use Eq.~\eqref{chidecomp}  to define the microcanonical state  in arbitrary physical theories.  However, the fact that the state is independent  of the choice of maximal set is not guaranteed to hold in every theory.  For this reason, we prefer to define the microcanonical state as the uniform mixture of {\em all} pure states with a given energy, rather than the uniform mixture of a particular maximal set of pure states.  Physically, the uniform mixture of all pure states  represents the result of fully uncontrolled, but energy conserving fluctuations in the experimental setup. From a subjective point of view, the uniform mixture represents the complete lack of knowledge besides the knowledge of the value of the energy:  not even the ``energy eigenbasis'' is assumed to be known. 
  
 \subsection{The condition of informational equilibrium} 

We have seen that sharp theories with purification satisfy requirement~\ref{req:uniquep}---the uniqueness of the uniform distribution over the pure states.    We now show that requirement~\ref{req:productchi}---the condition of informational equilibrium---is satisfied too.    

\begin{prop}
For every pair of systems $\rA$ and $\rB$, one has $\chi_{\rA\rB}   =  \chi_\rA\otimes \chi_{\rB}$. 
\end{prop}
\begin{proof}
Pick two pure maximal sets for $\rA$ and $\rB$, say $\left\{\alpha_i\right\}_{i=1}^{d_\rA}$ and   $\left\{\beta^\dag_j\right\}_{j=1}^{d_\rB}$. Then, the product set    $\left\{\alpha_i\otimes \beta_j \right\}_{  i  \in  \left\{  1,\dots, d_\rA\right\}  ,  j \in  \left\{  1,\dots , d_\rB\right\}}$   is maximal for the composite system $\rA\otimes \rB$, by proposition~\ref{prop:information locality}.  Using the decomposition~\eqref{chidecomp}, we obtain  
\begin{align}
\nonumber \chi_{\mathrm{AB}}  
\nonumber  &  =  \frac 1 {d_{\mathrm{AB}}}    \sum_{i=1}^{d_\rA}  \sum_{j=1}^{d_\rB}   \left(  \alpha_i\otimes \beta_j \right) \\ 
\nonumber  &  =  \frac 1 {d_{\rA}  d_{\rB} }     \left(  \sum_{i=1}^{d_\rA}      \alpha_i\right )  \otimes \left(  \sum_{j=1}^{d_\rB}   \beta_j\right)    \\
&  =  \chi_\rA \otimes \chi_\rB  , 
\end{align}  
having used the information locality condition $d_{\mathrm{AB}}  =  d_\rA d_\rB$. \end{proof}

In summary, sharp theories with purification satisfy our two requirements for the general microcanonical framework.    In the following, we will show that sharp theories with purification also  guarantee an important inclusion relation  between the set of RaRe channels and the set of noisy operations. 


\subsection{Inclusion of RaRe into Noisy}  

In sharp theories with purification, one can establish an inclusion between   RaRe channels and  noisy operations.   
   To obtain this result, we first restrict our attention to {\em rational} RaRe channels, i.e.\ RaRe channels of the form $\map R  =  \sum_{i}   p_i   \map U_i$ where each $p_i$ is a  \emph{rational} number.  With this definition, we have  the following lemma:

\begin{lem}\label{prop:rarenoisy}
In every sharp theory with purification, every rational RaRe channel is a basic noisy operation.
 \end{lem}

In quantum theory, this statement is quite immediate, as pointed out in Ref. \cite{Nicole}: a generic  RaRe channel with rational probabilities $\left\{ p_i  =  n_i/n\right\}_{i=1}^r$  and unitary gates $\left\lbrace U_i\right\rbrace _{i  = 1}^r$  can be realised  as the basic noisy operation  
\begin{align}\label{trivial1}
\map B  \left( \rho\right) 
:  =      \Tr_{\rm anc}  \left[   U   \left(\rho \otimes  \frac{I_n}{n} \right) U^\dag\right]  ,   
 \end{align}
where  $\Tr_{\rm anc}$ is the partial trace on the $n$-dimensional system used as ancilla, and $U$ is the control-unitary gate  
\begin{align}\label{trivial2}
U  :=    \sum_{k=1}^n   V_k  \otimes \ket{k}\bra{k} ,    
\end{align} 
$\left\lbrace   V_k\right\rbrace _{ k=1}^n$ being a list of unitary gates, $n_1$ of which are equal to $U_1$, $n_2$ equal to $U_2$, and so on.  

The situation is in general more complicated in sharp theories with purification. The reason is that the simple construction of Eqs.~\eqref{trivial1}  and \eqref{trivial2} cannot be reproduced.   The analogue of the control unitary  $U$ is a control-reversible transformation, which performs a reversible transformation on the target system depending on the state of a control system.  However, later in the paper we will show that not every sharp theory with purification admits control-reversible transformations.  In fact, we will show that the existence of control-reversible transformation is equivalent to a non-trivial property of the dynamics, which we will call ``unrestricted reversibility''.  

The non-trivial content of lemma~\ref{prop:rarenoisy} is that the inclusion  ${\sf RationalRaRe}  \subseteq {\sf Noisy}$ holds {\em in every sharp theory with purification}, without the need of assuming unrestricted reversibility.  To prove such a result we need a new construction, more elaborate than the simple  construction of Eqs.~\eqref{trivial1}  and \eqref{trivial2}. The details can be found in   appendix~\ref{app:rarenoisy}.

Now, since rational RaRe channels are dense in the set of RaRe channels, and since the set of noisy operations is closed (see definition~\ref{def:noisy}), we obtain the following theorem: 
\begin{thm}\label{theo:rarenoisy}
In every sharp theory with purification, RaRe channels are noisy operations.   
\end{thm}

The inclusion of RaRe channels in the set of noisy operations is generally strict: for example, in quantum theory there
exist noisy operations that are not RaRe channels \cite{Shor}. In
summary, we have the inclusions 
\begin{equation}
\mathsf{RaRe}\subseteq\mathsf{Noisy}\subseteq\mathsf{Unital},\label{eq:inclusions}
\end{equation}
illustrated in Fig.~\ref{fig:Inclusions}.

\begin{figure}
\begin{centering}
\includegraphics[scale=0.7]{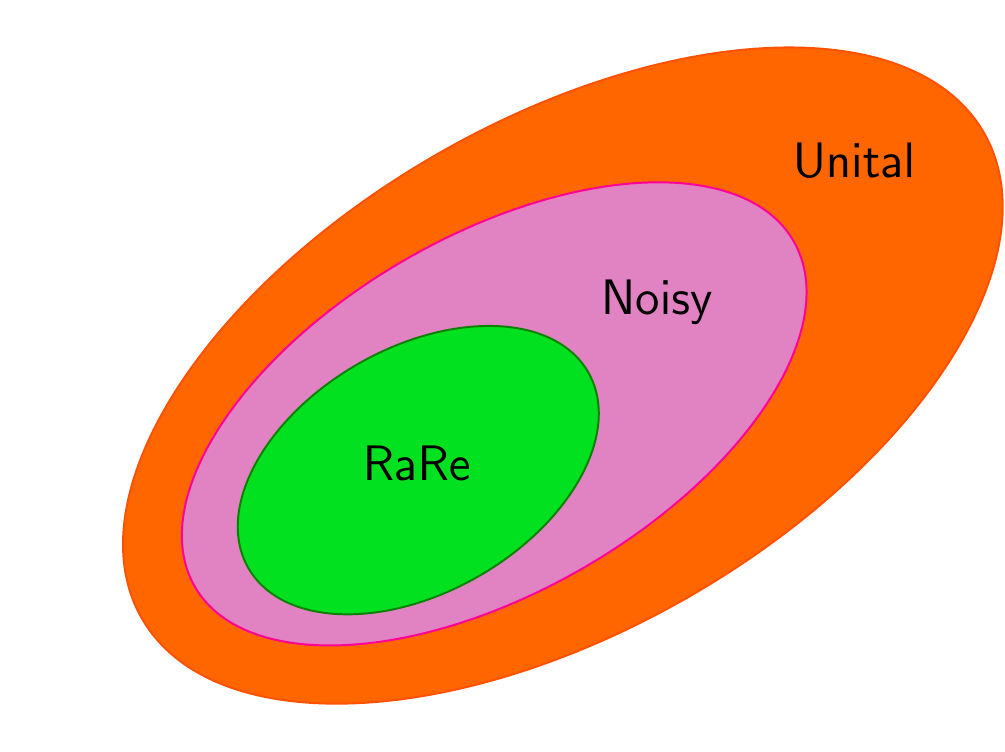}
\par\end{centering}

\caption{\label{fig:Inclusions}  Inclusion relations between the sets of free operations
in the three  resource theories of purity.}
\end{figure}

\section{State convertibility and majorisation\label{sec:Purity-in-sharp}}

In this section we investigate the convertibility of states in the RaRe, Noisy, and Unital Resource Theories.  The main result is that, in every
sharp theory with purification, an input state can be converted into an output state {\em by a unital channel} if and only
if the vector of eigenvalues of the output state is majorised by the vector  of eigenvalues of 
the input state. 
Since the set of unital channels contains the sets of noisy operations, and RaRe channels, our results establishes majorisation as a { necessary} condition for convertibility  under noisy operations and Rare channels. 
Later in the paper, we will determine the physical condition under which majorisation is  also sufficient.


\subsection{State convertibility}
In sharp theories with purification,  the inclusions~\eqref{eq:inclusions} imply the relations
\begin{align}\label{ConvertibilityImplications}
\rho\succeq_{\mathsf{RaRe}}\sigma  \quad \Longrightarrow \quad   \rho\succeq_{\mathsf{Noisy}}\sigma     \quad \Longrightarrow \quad  \rho\succeq_{\mathsf{Unital}}\sigma \, ,
\end{align}
  valid for every pair of states $\rho$ and $\sigma$ of the same system.     Note that  the unital relation
$\succeq_{\mathsf{Unital}}$ is the weakest, i.e.\ the easiest to
satisfy. In the following we will provide a necessary and sufficient condition for
the unital preorder. 

\subsection{Unital channels and doubly stochastic matrices}  
In a broad sense, unital channels are the generalisation of doubly stochastic matrices.  In sharp theories with purification,  there is also a more explicit connection: 
\begin{lem}\label{lem:channelmatrix}
Let $\map D$ be a unital channel acting  on system $\rA$ and let $\left\{\alpha_i\right\}_{i=1}^d$  and $\left\{\alpha_i'\right\}_{i=1}^d$  be two pure maximal sets of system $\rA$.  Then, the matrix $D$ with entries 
\begin{align}
D_{ij}   :  =  \left(  \alpha_i^\dag \middle|  \map D \middle|  \alpha_j' \right)  
\end{align}
is doubly stochastic. 
\end{lem}
\begin{proof}
Every entry $D_{ij}$ is a probability and therefore it is non-negative. Moreover, one has 
\begin{align}
\nonumber 
\sum_{i=1}^{d}   D_{ij}  &= \sum_{i=1}^{d} \left(    \alpha_i^\dag  \middle |    \map D  \middle |   \alpha'_j\right) \\
\nonumber 
  &  = \left(u\middle|\mathcal{D}\middle|\alpha_{j}\right)\\
\nonumber 
  &  = \Tr\left[ \alpha_{j}\right] \\
 \label{bisto1} &  = 1 \, , \qquad \forall j\in  \{1,\dots, d\}
\end{align}
having used the fact that the effects $\left\{ \alpha_{i}'^{\dagger}\right\} _{i=1}^{d}$  form an observation-test and that $\map D$ is a channel, and  therefore $u\mathcal{D}=u$ \cite{Chiribella-purification}. 
On the other hand, one has 
\begin{align}
\nonumber 
\sum_{j=1}^d  \,   D_{ij}   &=  \sum_{j=1}^{d}\left(\alpha_{i}^{\dagger}\middle|\mathcal{D}\middle|\alpha'_{j}\right)\\
\nonumber 
 & =d\left(\alpha_{i}^{\dagger}\middle|\mathcal{D}\middle|\chi\right)\\
\nonumber 
  &=d\left(\alpha_{i}^{\dagger}\middle|\chi\right)\\
  \nonumber 
  &=d\cdot\frac{1}{d}\\
 \label{bisto2} &=1 \, ,\qquad \forall i\in\left\{1,\ldots, d\right\}\, ,
\end{align}
having used  proposition~\ref{prop:chi invariant} and the fact that unital channels leave $\chi$ invariant.   In conclusion, Eqs.~\eqref{bisto1} and \eqref{bisto2} show that the matrix $D$ is doubly-stochastic. \end{proof}
Vice-versa, every doubly stochastic matrix defines a unital channel:  
 \begin{lem}\label{lem:matrixchannel}
Let $D$ be a $d\times d$ doubly stochastic matrix, and let $\left\{\alpha_i\right\}_{i=1}^d$  and $\left\{\alpha_i'\right\}_{i=1}^d$  be two pure maximal sets of system $\rA$.  Then, the channel  defined by 
\begin{equation}
\map D   :  =  \sum_{j=1}^{d}        \rho_j      \alpha_j^{\dag}   \, , \qquad  {\rm with}  \qquad  \rho_j :  =  \sum_{i=1}^{d}   D_{ij}  \alpha'_i ,
\end{equation}
 is unital. 
\end{lem}
 \begin{proof} The transformation $\map D$ is a channel of the measure-and-prepare form: it can be implemented by performing the observation-test $\left\{\alpha_j^{\dag}\right\}_{j=1}^{d}$ and by preparing the state $\rho_j$ conditionally on outcome $j$.  Moreover, one has \begin{align}
\nonumber 
\map D   \chi  &  =  \sum_{j=1}^{d}   \rho_j   \left(\alpha_j^{\dag}\middle|\chi\right)  \\
\nonumber 
 &  =  \frac 1 d  \sum_{j=1}^{d}    \sum_{i=1}^{d}   D_{ij}   \alpha'_i     \\
\nonumber 
 &  =  \frac 1d   \sum_{i=1}^{d} \alpha'_i\\
& =\chi ,  
\end{align} the third equality following from the definition of doubly stochastic matrix, and the fourth equality following from the diagonalisation of the  state  $\chi$ (proposition \ref{prop:chi invariant}).
\end{proof}
Lemmas~\ref{lem:channelmatrix} and \ref{lem:matrixchannel} establish a direct connection between unital channels and doubly stochastic matrices.  Using this connection, in the following we  establish a relation between the Unital Resource Theory and the theory of majorisation.   
\subsection{Majorisation criterion for state convertibility under unital channels}
Here we show that the ability to convert states in the Unital Resource Theory is completely determined by a suitable majorisation criterion. 
Let us start by recalling the definition of majorisation \cite{Olkin}:
\begin{defn}
Let $\mathbf{x}$ and $\mathbf{y}$ be two  generic vectors in $\mathbb{R}^{d}$.  One says that  $\mathbf{x}$ \emph{majorises}   $\mathbf{y}$, denoted ${\bf x}  \succeq  {\bf y}$,  if, when the entries of $\mathbf{x}$ and $\mathbf{y}$ are rearranged in decreasing order, one has
\begin{align*} 
\sum_{i=1}^{k}x_{i}\geq\sum_{i=1}^{k}y_{i}   , \quad \forall  k  <d   \qquad {\rm and}  \qquad    \sum_{i=1}^{d}x_{i}=\sum_{i=1}^{d}y_{i}  .
\end{align*}
\end{defn}
Majorisation can be equivalently characterised in terms of doubly
stochastic matrices: one has $\mathbf{x}\succeq\mathbf{y}$ if and only if
$\mathbf{y}=D\mathbf{x}$, where $D$ is a doubly stochastic matrix
\cite{Hardy-Littlewood-Polya1929,Olkin}.

In every sharp theory with purification, majorisation is a necessary and sufficient condition for convertibility under unital channels:  
\begin{thm}
\label{prop:unital channels}Let $\rho$ and $\sigma$ be normalised
states, and let $\mathbf{p}$ and $\mathbf{q}$ be 
the vectors of their eigenvalues, respectively. The state $\rho$ can be converted into the state $\sigma$ by a
unital channel if and only if ${\bf p}$ majorises ${\bf q}$. In formula:
\begin{equation}
\rho\succeq_{\mathsf{Unital}}\sigma\quad\Longleftrightarrow\quad\mathbf{p}\succeq\mathbf{q}.
\end{equation}
 \end{thm}
The proof is provided in appendix~\ref{app:unital channels}. Note that since RaRe channels and noisy operations are special cases of unital channels, majorisation is a {\em necessary} condition for convertibility the RaRe and Noisy Resource Theories.

\subsection{Characterisation of unital monotones}

The majorisation criterion determines whether a state is more resourceful than another. To be more quantitative, one can introduce monotones \cite{Resource-theories,Quantum-resource-2,Resource-monoid}---i.e.\ functions that are non-increasing under free operations: 
\begin{defn}
A \emph{monotone  under the free operations $\set F$}   for system $\rA$ is a function $P:   \St\left( \rA\right)   \to \R$ satisfying the
condition 
\begin{equation}
P\left( \rho\right)   \ge P\left( \sigma\right)   \qquad  \forall   \rho,\sigma\in\St\left( \rA\right)  ,    \rho  \succeq_{\set F}   \sigma   .
\end{equation} 
\end{defn}   
When $\set F$ is the set of unital operations,  we refer to $ P $ as {\em unital monotones}.  
In sharp theories with purification, unital monotones have an elegant mathematical characterisation: 
\begin{prop}
	A function on the state space $P:\mathsf{St}_{1}\left(\mathrm{A}\right)\to\mathbb{R}$
	is a unital monotone if and only if  $P\left(\rho\right)=f\left(\mathbf{p}\right)$, where 
	$\mathbf{p}$
	is the vector of eigenvalues of $\rho$ and 
	$f:\mathbb{R}^{d_\rA}\to\mathbb{R}$ is a {\em Schur-convex function}---that is, a function satisfying the condition $f\left( \st p\right)   \ge f \left( \st q\right) $ whenever $\st p \succeq \st q$. 
	\end{prop}
\begin{proof}
	Theorem~\ref{prop:unital channels} shows that the convertibility of states under unital channels is fully captured by their eigenvalues. Consequently, a unital monotone will be a function only of the eigenvalues of a state: there exists 
a function $f:\mathbb{R}^{d_\rA}\to\mathbb{R}$
	such that $P\left(\rho\right)=f\left(\mathbf{p}\right)$, for every
	normalised state $\rho$. Now, suppose that $\mathbf{p}$ and $\mathbf{q}$ are
	two probability distributions satisfying $\mathbf{p}\succeq\mathbf{q}$.
	Then, theorem~\ref{prop:unital channels} 	implies that there is a unital channel transforming the state $\rho=\sum_{i=1}^{d}p_{i}\alpha_{i}$ 
	into the state $\sigma=\sum_{i=1}^{d}q_{i}\alpha_{i}$, for any pure
	maximal set $\left\{ \alpha_{i}\right\}_{i=1}^{d}$. As a result,
	we obtain the relation
	\begin{equation}
	f\left(\mathbf{p}\right)=P\left(\rho\right)\geq P\left(\sigma\right)=f\left(\mathbf{q}\right) \, .
	\end{equation}
This means  that $f$ is Schur-convex.  Conversely, given a Schur-convex function $f$ one can define a function $P_f $ on the state space, as $P_f  \left( \rho\right)   :  =  f \left( \st p\right) $,  $\st p$ being the spectrum of $\rho$. This function  is easily proved to be a unital monotone, thanks to theorem \ref{prop:unital channels}.
\end{proof}

A  canonical example of Schur-convex function is the negative of the Shannon entropy, namely the function  
\begin{equation} 
f\left( \st p\right)   := -H\left( \mathbf{p}\right) \, ,\qquad  
H\left( \st p\right)    :  =  -\sum_{i=1}^{d}  p_i \log p_i  . \end{equation}  
 The corresponding purity monotone is   the negative of the Shannon-von Neumann entropy \cite{Scandolo,Chiribellatalk,TowardsThermo,Colleagues}
 \begin{equation}  P\left( \rho\right) :  =  -S\left( \rho\right)   , \qquad   S  \left( \rho\right)   :  =    H \left( \st p\right)       .\end{equation}   
Other important examples are the negatives of the Rényi entropies \cite{TowardsThermo,Colleagues}.   

\section{The counterexample of Doubled Quantum Theory}\label{sec:doubled} 
We have seen that  majorisation  is a necessary and sufficient condition for state convertibility in the Unital Resource Theory.    Is majorisation sufficient also for  convertibility  in the  RaRe Resource Theory?  Now we show that the answer is negative by constructing a counterexample, which we call ``Doubled Quantum Theory''.

\subsection{Individual systems}  
Consider a  theory  where every non-trivial system is the direct sum
of two identical quantum systems with Hilbert spaces  $\spc H_0$ and $\spc H_1$, respectively.  Physically, we can think of the two Hilbert spaces as two superselection sectors.    
 We associate  each ``doubled quantum system'' with a pair of isomorphic Hilbert spaces $  \left(   \spc H_0,  \spc H_1\right) $,  with  $\spc H_0  \simeq \spc H_1$.  
 We define the states of the doubled quantum system to be of the form
  \begin{equation}\label{eq:state doubled}\rho   =     p  \rho_0   \oplus    \left( 1-p\right)   \rho_1 \end{equation}
where $\rho_0$ and $\rho_1$ are two density matrices in the two sectors, respectively, and $p$ is a probability.  Likewise, we define the effects to be all quantum effects of the form  
$e  =   e_0  \oplus e_1$, where $e_0$ and $e_1$ are two quantum effects in the two sectors.  The allowed channels  from the input system $  \left(   \spc H_0,  \spc H_1\right) $ to the output system $ \left(   \spc K_0,  \spc K_1\right) $ are the quantum channels (completely positive trace-preserving maps) that 
\begin{enumerate}
\item send operators on  $\spc H_0 \oplus \spc H_1$  to operators on  $\spc K_0 \oplus \spc K_1$ 
\item map block diagonal operators to block diagonal operators. 
\end{enumerate}

The set of allowed tests is defined as the set of quantum instruments $\left\{\map C_i\right\}_{i\in\set X}$ where each quantum operation $\map C_i$  sends operators on  $\spc H_0 \oplus \spc H_1$  to operators on  $\spc K_0 \oplus \spc K_1$, mapping  block diagonal operators to block diagonal operators.

{\em Remark.} Note that the set of allowed  channels includes quantum channels of the form  $\map C   =   \map C_0  \oplus \map C_1$, where $\map C_0$ and $\map C_1$ are quantum channels acting on the individual sectors. However, not all allowed  channels are of this form.  
For example, our definition includes unitary channels, of the form $\map U  \left( \cdot\right)   =  U  \cdot U^\dag$,  with 
\begin{equation}\label{eq:allowed unitaries}
U= U_{0}\oplus U_{1} ,
\end{equation}
where  $U_{0}$ and $U_{1}$ are unitary operators   
acting on the subspaces $\spc H_0$ and $\spc H_1$, respectively.

The unitary channel $\map U$ is {\em different} from the non-unitary channel  $\map C  =  \map C_0  \oplus \map C_1  $, with $\map C_0  \left( \cdot\right)   =  U_0  \cdot  U_0^\dag$ and  $\map C_1  \left( \cdot\right)   =  U_1  \cdot  U_1^\dag$, even though the two channels act in the same way on every input state of the form~\eqref{eq:state doubled}.   Operationally, the difference between  the channels $\map U$ and $\map C$ will become visible when the channels are applied locally on one part of a composite system.

The existence of different physical transformations that are indistinguishable at the single-system level  is made possible by the fact that Doubled Quantum Theory does \emph{not} satisfy the Local Tomography axiom, as we show  in appendix~\ref{app:double local tomography}.   The inability to distinguish transformations at the single-system level is a fairly generic trait  of theories where the Local Tomography axiom does not hold.   Quantum theory on real Hilbert spaces also exhibit this trait \cite{Chiribella-purification,QuantumFromPrinciples} (but only for non-pure transformations \cite{Chiribella-purification}).

\subsection{Composite systems}

The peculiarity of Doubled Quantum Theory is the way systems are composed.  
The product of two doubled quantum systems $\left(   \spc H^\rA_0,  \spc H^\rA_1\right) $ and $\left(   \spc H^\rB_0,  \spc H^\rB_1\right) $ is the doubled 	quantum system $\left(   \spc H^{\rA\rB}_0,  \spc H^{\rA\rB}_1\right) $ defined by  
\begin{align}\label{hhahhb}
\nonumber \spc H_0^{\rA\rB}    &:  =   \left( \spc H_0^\rA  \otimes \spc H_0^\rB\right)   \oplus     \left( \spc H_1^\rA  \otimes \spc H_1^\rB \right)  \\
\spc H_1^{\rA\rB}   & :  =  \left(  \spc H_0^\rA  \otimes \spc H_1^\rB \right)  \oplus   \left(   \spc H_1^\rA  \otimes \spc H_0^\rB \right)  \, .
\end{align}
As an example,   consider the composite system of  two doubled qubits, corresponding to  $\spc H^\rA_0  \simeq \spc H^\rA_1 \simeq \spc H^\rB_0 \simeq \spc H^\rB_1  \simeq  \mathbb C^2$. An example of state of the composite system is the pure state
\begin{align}\label{example0011}
\ket{\Psi}   =     \frac{   \ket{0,0 }_\rA   \ket{0,0}_\rB  +   \ket{1,0}_\rA  \ket{1,0}_\rB     }{\sqrt 2}   ,
\end{align} 
where $\left\{\ket{0,0}  ,  \ket{0,1}\right\}$ is an orthonormal basis for $\spc H_0$ and $\left\{  \ket{1,0}  , \ket{1,1}\right\}$ is an orthonormal basis for $\spc H_1$.    Note that, when one of the two systems is traced out, the remaining local state has the block diagonal form $\rho  =  \frac{1}{2}  \ket{0,0}\bra{0,0}  \oplus  \frac{1}{2}  \ket{1,0}\bra{1,0}$.       This means that the coherence between the two summands in the state~\eqref{example0011} is invisible   at the single-system level.

From a physical point of view, Doubled Quantum Theory can be thought of as ordinary quantum theory with a superselection rule on the {\em total} parity.    Every system is  split into two identical sectors of even and odd parity, respectively.  When systems are composed, the sectors are grouped together based on the total parity, so that superpositions between subspaces with the same parity are allowed.

\subsection{In Doubled Quantum Theory, majorisation is not sufficient for convertibility under RaRe channels}  

In appendix~\ref{app:operationaldoubled} we summarise a few  operational features of  Doubled Quantum Theory. In particular, we show that Doubled Quantum Theory is a sharp theory with purification.
Nevertheless, now we show that  majorisation  does \emph{not}  guarantee
the convertibility of states under RaRe channels. 

Consider the following states of a doubled qubit:
\begin{equation}\label{eq:rho}
\rho=\frac{1}{2}   \Big(   \ket{0,0}\bra{0,0}  +  \ket{0,1}\bra{0,1} \Big) 
\end{equation}
and 
\begin{equation}\label{eq:sigma}
\sigma=\frac{1}{2}     \ket{0,0}\bra{0,0}  \oplus   \frac 12  \ket{1,0}\bra{1,0}  ,
\end{equation}
where $\left\{  \ket{0,0}  ,   \ket{0,1}\right\}$ is an orthonormal basis for $\spc H_0$ and $\left\{  \ket{1,0}  ,   \ket{1,1}\right\}$ an orthonormal basis for $\spc H_1$.  The key point here is that the state $\rho$ is fully contained in one sector (the even parity sector), while the state $\sigma$ is a mixture of two states in two different sectors.  

The two states have the same spectrum, and therefore they are equivalent  in terms of majorisation.    However, there is no RaRe channel transforming one state into the other.   To see this, we use the following lemmas: 
\begin{lem}[\cite{Muller3D}]
If any two states   $\rho $ and $\sigma$ are interconvertible with Rare channels, then there exists a reversible transformation $\map U$ such that $\sigma  = \map U \left(  \rho\right) $. 
\end{lem}

\begin{lem}  
 No unitary matrices in Doubled Quantum Theory are such that   $\sigma =  U  \rho U^\dag$, where $ \rho $ and $ \sigma $ are defined in Eqs.~\eqref{eq:rho} and \eqref{eq:sigma} respectively.
\end{lem}
\begin{proof}
The proof is by contradiction. Suppose that one has $\sigma =  U  \rho U^\dag$. Then,  define  the vectors $  \ket{\varphi_0}  :   =   U \ket{ 0,0} $ and $ \ket{\varphi_1}  :  =  U  \ket{0,1}$.   With this definition, we have 
\begin{align}
U \rho U^\dag  =    \frac 1 2  \left(    \ket{\varphi_0}\bra{\varphi_0}  +  \ket{\varphi_1}\bra{\varphi_1}\right) \, .
\end{align}
  Now,  $U\rho U^\dag$ must be an allowed state in Double Quantum Theory. This means that there are only two  possibilities: either $\ket{\varphi_0}$ and $\ket{\varphi_1}$ belong to the same sector, or they do not.  But $\sigma$ is a mixture of states in both sectors. Hence,  $\ket{\varphi_0}$ and $\ket{\varphi_1}$  must belong to different sectors, if the relation  $U\rho U^\dag = \sigma$ is to hold.     At this point, there are only two possibilities:  either  
\begin{align}\label{0000}
U \ket{0,0}   =   \ket{0,0}   \qquad \textrm{and} \qquad U  \ket{0,1}   =  \ket{1,0}    ,
\end{align}  
 or 
\begin{align}\label{0010}
U \ket{0,0}   =   \ket{1,0}   \qquad \textrm{and} \qquad U  \ket{0,1}   =  \ket{0,0}   .
\end{align}  
However, none of these conditions can be satisfied by a unitary in Double Quantum Theory:  every unitary matrix satisfying either  condition would map the valid state $  \ket{0  ,  +}  =   \frac{1}{\sqrt{2}}\left(    \ket{0,0}  +  \ket{0,1}\right)$ into the \emph{invalid} state $   \frac{1}{\sqrt{2}}\left(   \ket{0,0}   +  \ket{1,0}\right)$, which is forbidden by the parity superselection rule.   \end{proof}  

 Since unitary channels are the only reversible transformations in Doubled Quantum Theory, we conclude that no Rare channel can convert $\rho$ into $\sigma$.  Summarising:  majorisation is \emph{not} sufficient for the convertibility via RaRe channels.

\section{Equivalence of the three resource theories}\label{sec:sharp theories unrestricted}  
In this section we will determine when the  RaRe, Noisy, and Unital Resource Theories are equivalent in terms of state convertibility. 

\subsection{Unrestricted reversibility}  
The condition for the equivalence of the RaRe, Noisy, and Unital Resource Theories  can be expressed in three, mutually equivalent ways, corresponding to three axioms independently introduced by different authors: 
 
\begin{ax}[Permutability \cite{Hardy-informational-2,hardy2013}]
Every permutation of every pure maximal set can be implemented by a reversible transformation.
\end{ax}

\begin{bis}[Strong Symmetry \cite{Barnum-interference}]
For every two  pure maximal sets, there exists  a reversible transformation that converts the states in one set into the states in the other.  
\end{bis}

\begin{ter}[Reversible Controllability \cite{Control-reversible}]
For every pair of systems $\rA$ and  $\rB$, every  pure maximal set $\left\{  \alpha_i\right\}_{i=1}^d$ of system $\rA$ and every  set of reversible transformations  $\left\{\map U_i\right\}_{i=1}^d$ on system $\rB$, there exists a reversible transformation $\map U$ on the composite system $\rA\otimes \rB$ such that 
\begin{align}
\begin{aligned} \Qcircuit @C=1em @R=.7em @!R { & \prepareC{\alpha_i} & \qw \poloFantasmaCn{\rA} & \multigate{1}{\mathcal U} & \qw \poloFantasmaCn{\rA} & \qw \\ &  & \qw \poloFantasmaCn{\rB} & \ghost{\mathcal U} & \qw \poloFantasmaCn{\rB} & \qw  } \end{aligned} =
\begin{aligned} \Qcircuit @C=1em @R=.7em @!R { & \prepareC{\alpha_i} & \qw \poloFantasmaCn{\rA} & \qw  &\qw &\qw \\
&  & \qw \poloFantasmaCn{\rB}   &  \gate{\map U_i}  &  \qw \poloFantasmaCn{\rB} &\qw} \end{aligned}
\end{align}
for every $i \in  \left\{1,\dots,  d\right\}$.
\end{ter}

Permutability, Strong Symmetry, and Reversible Controllability are logically distinct requirements.  
For example, Strong Symmetry implies Permutability, but the converse is not true in general, as shown by the  example of the \emph{square bit} \cite{Barrett} in Fig.~\ref{fig:square} (see  appendix~\ref{app:square} for more details).    

\begin{figure}
	\begin{centering}
	\includegraphics[bb=0bp 0bp 120bp 115bp]{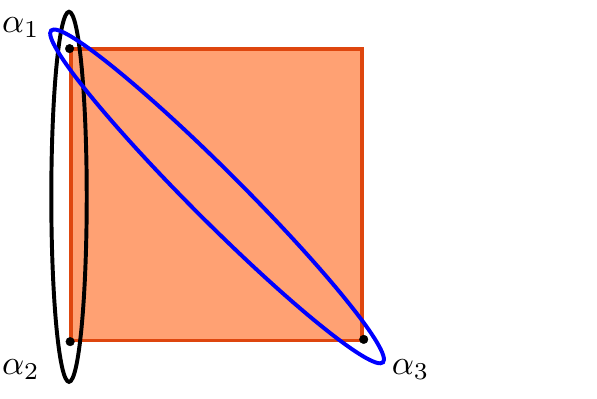}
	\par\end{centering}

	\caption{\label{fig:square}Normalised states of the square bit. The two sets $\left\{ \alpha_{1},\alpha_{2}\right\} $
		(circled in black) and  $\left\{ \alpha_{1},\alpha_{3}\right\} $ (circled in blue) consist of perfectly distinguishable pure states. Permutability holds,  because every permutation of every pair  of perfectly
		distinguishable pure states can be implemented by a reversible transformation, corresponding to a symmetry of the square.  
		However, no reversible transformation can transform $\alpha_2$ into $\alpha_3$ while leaving $\alpha_1$  unchanged. Hence, Strong Symmetry cannot hold for the square bit.}
\end{figure}

Although different in general, Permutability, Strong Symmetry, and Reversible Controllability become equivalent in sharp theories with purification:
\begin{prop}\label{prop:permstrong}
In every sharp theory with purification, Permutability, Strong Symmetry, and Reversible Controllability are equivalent requirements. 
\end{prop}
The proof is presented in appendix~\ref{app:permstrong}. The fact that three desirable properties become equivalent under our axioms gives a further evidence that the axioms capture an important structure of physical theories. 

Since Permutability,  Strong Symmetry, and Reversible Controllability are equivalent in the present
context, we conflate them into a single notion: 
\begin{defn}
A sharp theory with purification has \emph{unrestricted reversibility} if the theory satisfies Permutability, or Strong Symmetry, or Reversible Controllability. 
\end{defn}

\subsection{When the three resource theories of purity are equivalent}    

We now  characterise exactly when  the RaRe, Noisy, and Unital Resource theories are equivalent in terms of state convertibility.   Owing to the inclusions ${\sf RaRe}  \subseteq {\sf Noisy}  \subseteq {\sf Unital}$, a sufficient condition for the equivalence is that the convertibility under unital channels implies the convertibility under RaRe channels.  The characterisation is as follows:     

\begin{thm}\label{theo:unitalrare}  In every sharp theory with purification,
the following statements are equivalent:
\begin{enumerate}
\item   the RaRe, Noisy, and Unital Resource Theories are equivalent in terms of state convertibility
\item  the theory has unrestricted reversibility.  
 \end{enumerate}
\end{thm} \begin{proof}
The implication $2\Rightarrow 1$ was already proven in  Ref.~\cite{QPL15}
 To prove the implication $1\Rightarrow 2$, we show that condition $1$ implies the validity of Strong Symmetry. Let $\left\{\alpha_i\right\}_{i=1}^d$ and  $\left\{\alpha'_i\right\}_{i=1}^d$ be two pure maximal sets, and let  $\left\{p_i\right\}_{i=1}^d$ be a probability distribution, with $p_{1}  > p_{2}  >\dots  >p_{d}>0$.    
Consider the two states $\rho$ and $\sigma$ defined by $\rho=\sum_{i=1}^{d}p_{i}\alpha_{i}$,
and $\sigma=\sum_{i=1}^{d}p_{i}\alpha_{i}'$. Since the two states  $\rho $ and $\sigma$ have the same eigenvalues, the majorisation criterion guarantees that $\rho$ can be converted into $\sigma$ by a unital  channel, and vice versa (theorem~\ref{prop:unital channels}).    Now, our hypothesis is that convertibility under unital channels implies convertibility under RaRe channels.       The mutual convertibility of $\rho$ and $\sigma$ under RaRe channels implies that  there exists a reversible transformation $\mathcal{U}$ such that
$\sigma=\mathcal{U}\rho$ \cite{Muller3D,Chiribella-Scandolo-entanglement}.  
Applying the effect $\alpha_{1}'^{\dagger}$ to both sides of the
equality  $\sigma=\mathcal{U}\rho$, we obtain
\begin{align}
\nonumber p_{1} &=\left(\alpha_{1}'^{\dagger}\middle|\sigma\right)\\
\nonumber  &=\sum_{j=1}^{d}p_{j}\left(\alpha_{1}'^{\dagger}  \middle|\mathcal{U}\middle|\alpha_{j}\right)\\
\nonumber  &=\sum_{j=1}^{d}D_{1j} p_{j}\\
&\le p_{1},
\end{align}
having used the fact that  $D_{ij}:=\left(\alpha_{i}'^{\dagger}\middle|\mathcal{U}\middle|\alpha_{j}\right)$
are the entries of a doubly stochastic matrix (lemma~\ref{lem:channelmatrix}).    The above
condition is satisfied only if $\left(\alpha_{1}'^{\dagger}\middle|\mathcal{U}\middle|\alpha_{1}\right)=1$.     By the state-effect duality (proposition~\ref{prop:duality states-effects}), this  condition is equivalent to the condition 
\begin{align}\label{ualpha}
\mathcal{U}\alpha_{1}=\alpha_{1}' .
\end{align}   
 
 Now, decompose the states $\rho$ and $\sigma$ as 
\begin{align}
\rho   =  p_1  \alpha_1  +  \left( 1-p_1\right)     \rho_1    ,  \qquad  \rho_1  :  =    \frac{\sum_{i=2}^{d}p_{i}\alpha_{i}}{\sum_{i=2}^{d}p_{i}}
\end{align}
and 
\begin{align}
\sigma   =  p_1  \alpha'_1  +  \left( 1-p_1\right)     \sigma_1 ,  \qquad  \sigma_1  :  =    \frac{\sum_{i=2}^{d}p_{i}\alpha'_{i}}{\sum_{i=2}^{d}p_{i}}  .
\end{align}
Combining Eq.~\eqref{ualpha} with  the equality $\map U  \rho  =  \sigma$, we obtain the condition  $\map U  \rho_1  =  \sigma_1$.  Applying to $\rho_1$ and $\sigma_1$ the same argument we used for $\rho$ and $\sigma$,  we obtain the equality $\mathcal{U}\alpha_{2}=\alpha'_{2}$.
Iterating the procedure $d-1$ times, we finally obtain the 
equality $\mathcal{U}\alpha_{i}=\alpha'_{i}$ for every $i\in\left\lbrace 1,\dots,d\right\rbrace $. Hence,
every two maximal sets of perfectly distinguishable pure states are
connected by a reversible transformation. 
\end{proof}
Theorem~\ref{theo:unitalrare} gives  necessary and sufficient conditions for the equivalence of the three resource theories of microcanonical thermodynamics. In addition, it provides a thermodynamic motivation for the condition of unrestricted reversibility. 

\subsection{The equivalence in a nutshell}

The results of this Section can be summed up in the following theorem: 
\begin{thm}\label{theo:equivalence}
In every sharp theory with purification and unrestricted reversibility, the following are
equivalent
\begin{enumerate}
\item $\rho\succeq_{\mathsf{RaRe}}\sigma$
\item $\rho\succeq_{\mathsf{Noisy}}\sigma$
\item $\rho\succeq_{\mathsf{Unital}}\sigma$
\item $\mathbf{p}\succeq\mathbf{q}$
\end{enumerate}
for arbitrary normalised states $\rho$ and $\sigma$, where $\mathbf{p}$ and
$\mathbf{q}$ are the vectors of eigenvalues of $\rho$ and $\sigma$, respectively. 
\end{thm}
\begin{proof}

The implications $1\Rightarrow 2$ and $2\Rightarrow 3$ follow  from the inclusions~\eqref{eq:inclusions}.   The implication  
 $3 \Rightarrow 4$ follows from theorem~\ref{prop:unital channels}.   The implication $4\Rightarrow 1$ follows from the equivalence between majorisation and unital convertibility, combined with theorem~\ref{theo:unitalrare}.
\end{proof}

Theorem~\ref{theo:equivalence} tells us that the Rare, Noisy, and Unital Resource Theories  are all  equivalent in terms of state convertibility. It is
important to stress that the equivalence holds despite the fact that
the three sets of operations are generally  different.

An important consequence of the equivalence is that the RaRe, Noisy, and Unital Resource Theories have the same quantitative measures of resourcefulness:  
\begin{prop}
Let $P:\St_1 \left( \rA\right)   \to \R$ be a real-valued function on the state space of system $\rA$.  If $P$ is a monotone under one of the sets $\mathsf{RaRe}$, $\mathsf{Noisy}$ and $\mathsf{Unital}$,  then it is a monotone under all the other sets. 
\end{prop}

Since the preorders $\rho\succeq_{\mathsf{RaRe}}$,  $\rho\succeq_{\mathsf{Noisy}}$,    and $\succeq_{\mathsf{Unital}}$ coincide, we can say that the   RaRe, Noisy, and Unital Resource Theories define the  same notion of resource, which one may   call {\em purity}.    Accordingly, we will talk about ``the resource theory of purity'' without specifying the set of free operations.  

\section{The entanglement-thermodynamics duality  \label{sec:monotones}}  

We conclude the paper by showing that sharp theories with purification and unrestricted reversibility exhibit a fundamental duality between the resource theory of purity and the resource theory of entanglement  \cite{Chiribella-Scandolo-entanglement}.    
The entanglement-thermodynamics duality is a duality between two resource theories: the resource theory of purity (with RaRe, or Noisy, or Unital channels as free operations) and the resource theory of pure bipartite entanglement (with local operations and classical communication as  free operations). The content of the duality is that a pure bipartite state is more entangled than another if and only if the marginal states of the latter are purer than the marginal states of the former.   
More formally, the duality can be stated as follows \cite{Chiribella-Scandolo-entanglement}: 
\begin{defn}
A theory satisfies the \emph{entanglement-thermodynamics duality} if for every pair of systems $\rA$ and $\rB$, and every pair of pure states $  \Phi, \Psi \in  \Pur\St_1 \left( \rA\otimes \rB\right) $ the following are equivalent  
\begin{enumerate}
\item $\Phi$ can be converted into $\Psi$ by local operations and classical communication
\item the marginal of $\Psi$  on system $\rA$ can be converted into the marginal of $\Phi$ on system $\rA$ by a RaRe channel
\item the marginal of $\Psi$  on system $\rB$ can be converted into the marginal of $\Phi$ on system $\rB$ by a RaRe channel. 
\end{enumerate}
 \end{defn}

Our earlier work  \cite{Chiribella-Scandolo-entanglement}  showed that  the entanglement-thermodynamics duality can be proved from four axioms: Causality, Purity Preservation, Purification, and Local Exchangeability---the latter defined as follows: 
\begin{defn} 
A theory satisfies \emph{Local Exchangeability} if for every pair of systems $\rA$ and $\rB$, and for every pure state  $\Psi  \in \Pur\St \left( \rA\otimes \rB\right) $  there exist two channels $\map C \in  \Det\Transf  \left( \rA,\rB\right) $ and $\map D\in\Det\Transf\left( \rB,\rA\right) $ such that
\begin{equation}    \begin{aligned}\Qcircuit @C=1em @R=.7em @!R { &\multiprepareC{1}{\Psi}  & \qw \poloFantasmaCn{\rA} &\gate{\map C}  &   \qw \poloFantasmaCn{\rB}  & \qw   \\ & \pureghost{\Psi}  & \qw \poloFantasmaCn{\rB} & \gate{\map D}  &   \qw \poloFantasmaCn{\rA}  & \qw   }  
  \end{aligned}   
   ~=\!\!\!\!  
 \begin{aligned}\Qcircuit @C=1em @R=.7em @!R { &\multiprepareC{1}{\Psi}  & \qw \poloFantasmaCn{\rA} &\multigate{1}{\tt SWAP}  &   \qw \poloFantasmaCn{\rB}  & \qw   \\ & \pureghost{\Psi}  & \qw \poloFantasmaCn{\rB} & \ghost{\tt SWAP}  &   \qw \poloFantasmaCn{\rA}  & \qw   }  
  \end{aligned}   ~,   \end{equation}
  where $\tt SWAP$ is the channel that exchanges system $\rA$ and system $\rB$. 
\end{defn}  

Since Causality, Purity Preservation, and Purification are already assumed among our axioms,   proving the entanglement-thermodynamics duality  is reduced to proving the validity of Local Exchangeability.   The proof is presented in appendix~\ref{app:localexchange}, which backs the following claim: 
\begin{thm}
Every sharp theory with purification and unrestricted reversibility satisfies the entanglement-thermodynamics duality. 
\end{thm}
As a consequence of the duality, the purity monotones characterised in the previous subsection are in one-to-one correspondence with measures of pure bipartite entanglement.  For example, Shannon-von Neumann entropy of the marginals of a pure bipartite state can be regarded as the \emph{entanglement entropy} \cite{Entanglement-entropy1,Entanglement-entropy2,Janzing2009}, an entropic measure of entanglement that is playing an increasingly important role in quantum field theory \cite{Ryu1,Ryu2} and condensed matter \cite{Area-law}.

\section{Conclusions\label{sec:Conclusions}}

In this work we developed a microcanonical framework for general physical theories. The framework is based on two requirements: the  uniqueness of the invariant probability distribution over pure states, needed to define the microcanonical state,   and the stability of the microcanonical  state under composition.  Under these requirements, we defined three resource theories, where free operations are random reversible channels, noisy operations, and unital channels, respectively.    
 We explored the connections between these three sets of operations  in a special class of physical theories, called sharp theories with purification, which  enable a fundamentally reversible description of every process.     In sharp theories with purification,  the sets of random reversible channels is contained in the set of noisy operations, which in turn is contained in the set of unital channels.   Convertibility under unital channels is equivalent to majorisation, which is a necessary condition for convertibility under the other sets of operations.  
    Majorisation becomes a sufficient condition for convertibility under \emph{all} sets of operations if and only if the dynamics allowed by the  theory have a property, called unrestricted reversibility.   In this case, one obtains the  entanglement-thermodynamics duality, which connects the entanglement of pure bipartite states with the purity of their marginals.   
    
    Our results identify  sharp theories with purification and unrestricted reversibility as the natural candidate for the information-theoretic foundation of microcanonical thermodynamics.    At the same time, it is interesting to go beyond the microcanonical scenario and to develop a general probabilistic framework for the canonical ensemble.    Some steps in this direction can be found  in a companion paper \cite{TowardsThermo}, where we give an operational definition of the Gibbs state and use it in an information-theoretic derivation of Landauer's principle.  These results are only the surface of a deep operational structure, where thermodynamic and information-theoretic features are interwoven at the level of  fundamental principles.   Many interesting directions of research  remain open, including, for example, an extension of the notion of thermomajorisation \cite{Horodecki-Oppenheim-2},  a derivation of the monotonicity of the relative entropy \cite{Petz}, and a derivation of  the ``second laws of thermodynamics''  \cite{2ndlaws} from operational axioms.

\begin{acknowledgments}
We thank the Referees for insightful comments that stimulated a substantial extension of our original manuscript. 

This work has been supported by the Foundational Questions Institute (FQXi-RFP3-1325 and FQXi-RFP-1608),   by the Canadian Institute for Advanced Research (CIFAR), by the 1000 Youth Fellowship Program of China and by the Hong Kong Research Grant Council through grant 17326616. 
 GC acknowledges the hospitality of the Simons Institute for the Theory of Computation and of Perimeter Institute.  
CMS  thanks
J Barrett, L Hardy, R Spekkens, M Hoban, C Lee, J Selby, and A Tosini for valuable discussions, acknowledges the support by EPSRC doctoral training grant and
by Oxford-Google Deepmind Graduate Scholarship, and  the hospitality of  Perimeter Institute.     Research at Perimeter Institute Theoretical Physics is supported in part by the Government of Canada through NSERC
and by the Province of Ontario through MRI.
\end{acknowledgments}

\bibliographystyle{apsrev4-1}
\bibliography{bibliographyQPL16}

\begin{thebibliography}{102}%
\makeatletter
\providecommand \@ifxundefined [1]{%
 \@ifx{#1\undefined}
}%
\providecommand \@ifnum [1]{%
 \ifnum #1\expandafter \@firstoftwo
 \else \expandafter \@secondoftwo
 \fi
}%
\providecommand \@ifx [1]{%
 \ifx #1\expandafter \@firstoftwo
 \else \expandafter \@secondoftwo
 \fi
}%
\providecommand \natexlab [1]{#1}%
\providecommand \enquote  [1]{``#1''}%
\providecommand \bibnamefont  [1]{#1}%
\providecommand \bibfnamefont [1]{#1}%
\providecommand \citenamefont [1]{#1}%
\providecommand \href@noop [0]{\@secondoftwo}%
\providecommand \href [0]{\begingroup \@sanitize@url \@href}%
\providecommand \@href[1]{\@@startlink{#1}\@@href}%
\providecommand \@@href[1]{\endgroup#1\@@endlink}%
\providecommand \@sanitize@url [0]{\catcode `\\12\catcode `\$12\catcode
  `\&12\catcode `\#12\catcode `\^12\catcode `\_12\catcode `\%12\relax}%
\providecommand \@@startlink[1]{}%
\providecommand \@@endlink[0]{}%
\providecommand \url  [0]{\begingroup\@sanitize@url \@url }%
\providecommand \@url [1]{\endgroup\@href {#1}{\urlprefix }}%
\providecommand \urlprefix  [0]{URL }%
\providecommand \Eprint [0]{\href }%
\providecommand \doibase [0]{http://dx.doi.org/}%
\providecommand \selectlanguage [0]{\@gobble}%
\providecommand \bibinfo  [0]{\@secondoftwo}%
\providecommand \bibfield  [0]{\@secondoftwo}%
\providecommand \translation [1]{[#1]}%
\providecommand \BibitemOpen [0]{}%
\providecommand \bibitemStop [0]{}%
\providecommand \bibitemNoStop [0]{.\EOS\space}%
\providecommand \EOS [0]{\spacefactor3000\relax}%
\providecommand \BibitemShut  [1]{\csname bibitem#1\endcsname}%
\let\auto@bib@innerbib\@empty
\bibitem [{\citenamefont {Goold}\ \emph {et~al.}(2016)\citenamefont {Goold},
  \citenamefont {Huber}, \citenamefont {Riera}, \citenamefont {del Rio},\ and\
  \citenamefont {Skrzypczyk}}]{delRio}%
  \BibitemOpen
  \bibfield  {author} {\bibinfo {author} {\bibfnamefont {J.}~\bibnamefont
  {Goold}}, \bibinfo {author} {\bibfnamefont {M.}~\bibnamefont {Huber}},
  \bibinfo {author} {\bibfnamefont {A.}~\bibnamefont {Riera}}, \bibinfo
  {author} {\bibfnamefont {L.}~\bibnamefont {del Rio}}, \ and\ \bibinfo
  {author} {\bibfnamefont {P.}~\bibnamefont {Skrzypczyk}},\ }\href {\doibase
  10.1088/1751-8113/49/14/143001} {\bibfield  {journal} {\bibinfo  {journal}
  {J. Phys. A}\ }\textbf {\bibinfo {volume} {49}},\ \bibinfo {pages} {143001}
  (\bibinfo {year} {2016})}\BibitemShut {NoStop}%
\bibitem [{\citenamefont {Vinjanampathy}\ and\ \citenamefont
  {Anders}(2016)}]{Anders-thermo}%
  \BibitemOpen
  \bibfield  {author} {\bibinfo {author} {\bibfnamefont {S.}~\bibnamefont
  {Vinjanampathy}}\ and\ \bibinfo {author} {\bibfnamefont {J.}~\bibnamefont
  {Anders}},\ }\href {\doibase 10.1080/00107514.2016.1201896} {\bibfield
  {journal} {\bibinfo  {journal} {Contemp. Phys.}\ }\textbf {\bibinfo {volume}
  {57}},\ \bibinfo {pages} {1} (\bibinfo {year} {2016})}\BibitemShut {NoStop}%
\bibitem [{\citenamefont {Millen}\ and\ \citenamefont {Xuereb}(2016)}]{Xuereb}%
  \BibitemOpen
  \bibfield  {author} {\bibinfo {author} {\bibfnamefont {J.}~\bibnamefont
  {Millen}}\ and\ \bibinfo {author} {\bibfnamefont {A.}~\bibnamefont
  {Xuereb}},\ }\href {http://stacks.iop.org/1367-2630/18/i=1/a=011002}
  {\bibfield  {journal} {\bibinfo  {journal} {New J. Phys.}\ }\textbf {\bibinfo
  {volume} {18}},\ \bibinfo {pages} {011002} (\bibinfo {year}
  {2016})}\BibitemShut {NoStop}%
\bibitem [{\citenamefont {Dahlsten}\ \emph {et~al.}(2011)\citenamefont
  {Dahlsten}, \citenamefont {Renner}, \citenamefont {Rieper},\ and\
  \citenamefont {Vedral}}]{Dahlsten-extractable}%
  \BibitemOpen
  \bibfield  {author} {\bibinfo {author} {\bibfnamefont {O.~C.~O.}\
  \bibnamefont {Dahlsten}}, \bibinfo {author} {\bibfnamefont {R.}~\bibnamefont
  {Renner}}, \bibinfo {author} {\bibfnamefont {E.}~\bibnamefont {Rieper}}, \
  and\ \bibinfo {author} {\bibfnamefont {V.}~\bibnamefont {Vedral}},\ }\href
  {\doibase 10.1088/1367-2630/13/5/053015} {\bibfield  {journal} {\bibinfo
  {journal} {New J. Phys.}\ }\textbf {\bibinfo {volume} {13}},\ \bibinfo
  {pages} {053015} (\bibinfo {year} {2011})}\BibitemShut {NoStop}%
\bibitem [{\citenamefont {\r{A}berg}(2013)}]{Aberg}%
  \BibitemOpen
  \bibfield  {author} {\bibinfo {author} {\bibfnamefont {J.}~\bibnamefont
  {\r{A}berg}},\ }\href {\doibase 10.1038/ncomms2712} {\bibfield  {journal}
  {\bibinfo  {journal} {Nat. Commun.}\ }\textbf {\bibinfo {volume} {4}},\
  \bibinfo {pages} {1925} (\bibinfo {year} {2013})}\BibitemShut {NoStop}%
\bibitem [{\citenamefont {Skrzypczyk}\ \emph {et~al.}(2014)\citenamefont
  {Skrzypczyk}, \citenamefont {Short},\ and\ \citenamefont
  {Popescu}}]{Popescu-single-shot}%
  \BibitemOpen
  \bibfield  {author} {\bibinfo {author} {\bibfnamefont {P.}~\bibnamefont
  {Skrzypczyk}}, \bibinfo {author} {\bibfnamefont {A.~J.}\ \bibnamefont
  {Short}}, \ and\ \bibinfo {author} {\bibfnamefont {S.}~\bibnamefont
  {Popescu}},\ }\href {\doibase 10.1038/ncomms5185} {\bibfield  {journal}
  {\bibinfo  {journal} {Nat. Commun.}\ }\textbf {\bibinfo {volume} {5}},\
  \bibinfo {pages} {4185} (\bibinfo {year} {2014})}\BibitemShut {NoStop}%
\bibitem [{\citenamefont {Gallego}\ \emph {et~al.}(2016)\citenamefont
  {Gallego}, \citenamefont {Eisert},\ and\ \citenamefont
  {Wilming}}]{Work-operational}%
  \BibitemOpen
  \bibfield  {author} {\bibinfo {author} {\bibfnamefont {R.}~\bibnamefont
  {Gallego}}, \bibinfo {author} {\bibfnamefont {J.}~\bibnamefont {Eisert}}, \
  and\ \bibinfo {author} {\bibfnamefont {H.}~\bibnamefont {Wilming}},\ }\href
  {\doibase 10.1088/1367-2630/18/10/103017} {\bibfield  {journal} {\bibinfo
  {journal} {New J. Phys.}\ }\textbf {\bibinfo {volume} {18}},\ \bibinfo
  {pages} {103017} (\bibinfo {year} {2016})}\BibitemShut {NoStop}%
\bibitem [{\citenamefont {Lostaglio}\ \emph
  {et~al.}(2015{\natexlab{a}})\citenamefont {Lostaglio}, \citenamefont
  {M\"uller},\ and\ \citenamefont {Pastena}}]{Lostaglio-Muller}%
  \BibitemOpen
  \bibfield  {author} {\bibinfo {author} {\bibfnamefont {M.}~\bibnamefont
  {Lostaglio}}, \bibinfo {author} {\bibfnamefont {M.~P.}\ \bibnamefont
  {M\"uller}}, \ and\ \bibinfo {author} {\bibfnamefont {M.}~\bibnamefont
  {Pastena}},\ }\href {\doibase 10.1103/PhysRevLett.115.150402} {\bibfield
  {journal} {\bibinfo  {journal} {Phys. Rev. Lett.}\ }\textbf {\bibinfo
  {volume} {115}},\ \bibinfo {pages} {150402} (\bibinfo {year}
  {2015}{\natexlab{a}})}\BibitemShut {NoStop}%
\bibitem [{\citenamefont {Gemmer}\ and\ \citenamefont
  {Anders}(2015)}]{Gemmer-single-shot}%
  \BibitemOpen
  \bibfield  {author} {\bibinfo {author} {\bibfnamefont {J.}~\bibnamefont
  {Gemmer}}\ and\ \bibinfo {author} {\bibfnamefont {J.}~\bibnamefont
  {Anders}},\ }\href {\doibase 10.1088/1367-2630/17/8/085006} {\bibfield
  {journal} {\bibinfo  {journal} {New J. Phys.}\ }\textbf {\bibinfo {volume}
  {17}},\ \bibinfo {pages} {085006} (\bibinfo {year} {2015})}\BibitemShut
  {NoStop}%
\bibitem [{\citenamefont {Yunger~Halpern}\ \emph
  {et~al.}(2015{\natexlab{a}})\citenamefont {Yunger~Halpern}, \citenamefont
  {Garner}, \citenamefont {Dahlsten},\ and\ \citenamefont
  {Vedral}}]{Garner-one-shot1}%
  \BibitemOpen
  \bibfield  {author} {\bibinfo {author} {\bibfnamefont {N.}~\bibnamefont
  {Yunger~Halpern}}, \bibinfo {author} {\bibfnamefont {A.~J.~P.}\ \bibnamefont
  {Garner}}, \bibinfo {author} {\bibfnamefont {O.~C.~O.}\ \bibnamefont
  {Dahlsten}}, \ and\ \bibinfo {author} {\bibfnamefont {V.}~\bibnamefont
  {Vedral}},\ }\href {\doibase 10.1088/1367-2630/17/9/095003} {\bibfield
  {journal} {\bibinfo  {journal} {New J. Phys.}\ }\textbf {\bibinfo {volume}
  {17}},\ \bibinfo {pages} {095003} (\bibinfo {year}
  {2015}{\natexlab{a}})}\BibitemShut {NoStop}%
\bibitem [{\citenamefont {Yunger~Halpern}\ \emph
  {et~al.}(2015{\natexlab{b}})\citenamefont {Yunger~Halpern}, \citenamefont
  {Garner}, \citenamefont {Dahlsten},\ and\ \citenamefont
  {Vedral}}]{Garner-one-shot2}%
  \BibitemOpen
  \bibfield  {author} {\bibinfo {author} {\bibfnamefont {N.}~\bibnamefont
  {Yunger~Halpern}}, \bibinfo {author} {\bibfnamefont {A.~J.~P.}\ \bibnamefont
  {Garner}}, \bibinfo {author} {\bibfnamefont {O.~C.~O.}\ \bibnamefont
  {Dahlsten}}, \ and\ \bibinfo {author} {\bibfnamefont {V.}~\bibnamefont
  {Vedral}},\ }\href {http://arxiv.org/abs/1505.06217} {\bibfield  {journal}
  {\bibinfo  {journal} {arXiv:1505.06217 [cond-mat.stat-mech]}\ } (\bibinfo
  {year} {2015}{\natexlab{b}})}\BibitemShut {NoStop}%
\bibitem [{\citenamefont {Salek}\ and\ \citenamefont
  {Wiesner}(2015)}]{Fluctuations1}%
  \BibitemOpen
  \bibfield  {author} {\bibinfo {author} {\bibfnamefont {S.}~\bibnamefont
  {Salek}}\ and\ \bibinfo {author} {\bibfnamefont {K.}~\bibnamefont
  {Wiesner}},\ }\href {http://arxiv.org/abs/1504.05111} {\bibfield  {journal}
  {\bibinfo  {journal} {arXiv:1504.05111 [quant-ph]}\ } (\bibinfo {year}
  {2015})}\BibitemShut {NoStop}%
\bibitem [{\citenamefont {Alhambra}\ \emph
  {et~al.}(2016{\natexlab{a}})\citenamefont {Alhambra}, \citenamefont
  {Oppenheim},\ and\ \citenamefont {Perry}}]{Fluctuations2}%
  \BibitemOpen
  \bibfield  {author} {\bibinfo {author} {\bibfnamefont {A.~M.}\ \bibnamefont
  {Alhambra}}, \bibinfo {author} {\bibfnamefont {J.}~\bibnamefont {Oppenheim}},
  \ and\ \bibinfo {author} {\bibfnamefont {C.}~\bibnamefont {Perry}},\ }\href
  {\doibase 10.1103/PhysRevX.6.041016} {\bibfield  {journal} {\bibinfo
  {journal} {Phys. Rev. X}\ }\textbf {\bibinfo {volume} {6}},\ \bibinfo {pages}
  {041016} (\bibinfo {year} {2016}{\natexlab{a}})}\BibitemShut {NoStop}%
\bibitem [{\citenamefont {Korzekwa}\ \emph {et~al.}(2016)\citenamefont
  {Korzekwa}, \citenamefont {Lostaglio}, \citenamefont {Oppenheim},\ and\
  \citenamefont {Jennings}}]{Work-coherence}%
  \BibitemOpen
  \bibfield  {author} {\bibinfo {author} {\bibfnamefont {K.}~\bibnamefont
  {Korzekwa}}, \bibinfo {author} {\bibfnamefont {M.}~\bibnamefont {Lostaglio}},
  \bibinfo {author} {\bibfnamefont {J.}~\bibnamefont {Oppenheim}}, \ and\
  \bibinfo {author} {\bibfnamefont {D.}~\bibnamefont {Jennings}},\ }\href
  {\doibase 10.1088/1367-2630/18/2/023045} {\bibfield  {journal} {\bibinfo
  {journal} {New J. Phys.}\ }\textbf {\bibinfo {volume} {18}},\ \bibinfo
  {pages} {023045} (\bibinfo {year} {2016})}\BibitemShut {NoStop}%
\bibitem [{\citenamefont {Wilming}\ \emph {et~al.}(2016)\citenamefont
  {Wilming}, \citenamefont {Gallego},\ and\ \citenamefont
  {Eisert}}]{2ndlaw-control}%
  \BibitemOpen
  \bibfield  {author} {\bibinfo {author} {\bibfnamefont {H.}~\bibnamefont
  {Wilming}}, \bibinfo {author} {\bibfnamefont {R.}~\bibnamefont {Gallego}}, \
  and\ \bibinfo {author} {\bibfnamefont {J.}~\bibnamefont {Eisert}},\ }\href
  {\doibase 10.1103/PhysRevE.93.042126} {\bibfield  {journal} {\bibinfo
  {journal} {Phys. Rev. E}\ }\textbf {\bibinfo {volume} {93}},\ \bibinfo
  {pages} {042126} (\bibinfo {year} {2016})}\BibitemShut {NoStop}%
\bibitem [{\citenamefont {Alhambra}\ \emph
  {et~al.}(2016{\natexlab{b}})\citenamefont {Alhambra}, \citenamefont
  {Masanes}, \citenamefont {Oppenheim},\ and\ \citenamefont
  {Perry}}]{2ndlaw-equality}%
  \BibitemOpen
  \bibfield  {author} {\bibinfo {author} {\bibfnamefont {A.~M.}\ \bibnamefont
  {Alhambra}}, \bibinfo {author} {\bibfnamefont {L.}~\bibnamefont {Masanes}},
  \bibinfo {author} {\bibfnamefont {J.}~\bibnamefont {Oppenheim}}, \ and\
  \bibinfo {author} {\bibfnamefont {C.}~\bibnamefont {Perry}},\ }\href
  {\doibase 10.1103/PhysRevX.6.041017} {\bibfield  {journal} {\bibinfo
  {journal} {Phys. Rev. X}\ }\textbf {\bibinfo {volume} {6}},\ \bibinfo {pages}
  {041017} (\bibinfo {year} {2016}{\natexlab{b}})}\BibitemShut {NoStop}%
\bibitem [{\citenamefont {Sparaciari}\ \emph {et~al.}(2017)\citenamefont
  {Sparaciari}, \citenamefont {Jennings},\ and\ \citenamefont
  {Oppenheim}}]{Sparaciari-2}%
  \BibitemOpen
  \bibfield  {author} {\bibinfo {author} {\bibfnamefont {C.}~\bibnamefont
  {Sparaciari}}, \bibinfo {author} {\bibfnamefont {D.}~\bibnamefont
  {Jennings}}, \ and\ \bibinfo {author} {\bibfnamefont {J.}~\bibnamefont
  {Oppenheim}},\ }\href {http://arxiv.org/abs/1701.01703} {\bibfield  {journal}
  {\bibinfo  {journal} {arXiv:1701.01703 [quant-ph]}\ } (\bibinfo {year}
  {2017})}\BibitemShut {NoStop}%
\bibitem [{\citenamefont {Horodecki}\ and\ \citenamefont
  {Oppenheim}(2013{\natexlab{a}})}]{Quantum-resource-1}%
  \BibitemOpen
  \bibfield  {author} {\bibinfo {author} {\bibfnamefont {M.}~\bibnamefont
  {Horodecki}}\ and\ \bibinfo {author} {\bibfnamefont {J.}~\bibnamefont
  {Oppenheim}},\ }\href {\doibase 10.1142/S0217979213450197} {\bibfield
  {journal} {\bibinfo  {journal} {Int. J. Mod. Phys. B}\ }\textbf {\bibinfo
  {volume} {27}},\ \bibinfo {pages} {1345019} (\bibinfo {year}
  {2013}{\natexlab{a}})}\BibitemShut {NoStop}%
\bibitem [{\citenamefont {Brand\~ao}\ and\ \citenamefont
  {Gour}(2015)}]{Quantum-resource-2}%
  \BibitemOpen
  \bibfield  {author} {\bibinfo {author} {\bibfnamefont {F.~G. S.~L.}\
  \bibnamefont {Brand\~ao}}\ and\ \bibinfo {author} {\bibfnamefont
  {G.}~\bibnamefont {Gour}},\ }\href {\doibase 10.1103/PhysRevLett.115.070503}
  {\bibfield  {journal} {\bibinfo  {journal} {Phys. Rev. Lett.}\ }\textbf
  {\bibinfo {volume} {115}},\ \bibinfo {pages} {070503} (\bibinfo {year}
  {2015})}\BibitemShut {NoStop}%
\bibitem [{\citenamefont {Coecke}\ \emph
  {et~al.}(2016{\natexlab{a}})\citenamefont {Coecke}, \citenamefont {Fritz},\
  and\ \citenamefont {Spekkens}}]{Resource-theories}%
  \BibitemOpen
  \bibfield  {author} {\bibinfo {author} {\bibfnamefont {B.}~\bibnamefont
  {Coecke}}, \bibinfo {author} {\bibfnamefont {T.}~\bibnamefont {Fritz}}, \
  and\ \bibinfo {author} {\bibfnamefont {R.~W.}\ \bibnamefont {Spekkens}},\
  }\href {\doibase http://dx.doi.org/10.1016/j.ic.2016.02.008} {\bibfield
  {journal} {\bibinfo  {journal} {Inform. Comput.}\ }\textbf {\bibinfo {volume}
  {250}},\ \bibinfo {pages} {59} (\bibinfo {year} {2016}{\natexlab{a}})},\
  \bibinfo {note} {quantum Physics and Logic}\BibitemShut {NoStop}%
\bibitem [{\citenamefont {Fritz}(2017)}]{Resource-monoid}%
  \BibitemOpen
  \bibfield  {author} {\bibinfo {author} {\bibfnamefont {T.}~\bibnamefont
  {Fritz}},\ }\href {\doibase 10.1017/S0960129515000444} {\bibfield  {journal}
  {\bibinfo  {journal} {Math. Structures Comput. Sci.}\ }\textbf {\bibinfo
  {volume} {27}},\ \bibinfo {pages} {850} (\bibinfo {year} {2017})}\BibitemShut
  {NoStop}%
\bibitem [{\citenamefont {del Rio}\ \emph {et~al.}(2015)\citenamefont {del
  Rio}, \citenamefont {Kr\"amer},\ and\ \citenamefont
  {Renner}}]{Resource-knowledge}%
  \BibitemOpen
  \bibfield  {author} {\bibinfo {author} {\bibfnamefont {L.}~\bibnamefont {del
  Rio}}, \bibinfo {author} {\bibfnamefont {L.}~\bibnamefont {Kr\"amer}}, \ and\
  \bibinfo {author} {\bibfnamefont {R.}~\bibnamefont {Renner}},\ }\href
  {http://arxiv.org/abs/1511.08818} {\bibfield  {journal} {\bibinfo  {journal}
  {arXiv:1511.08818 [quant-ph]}\ } (\bibinfo {year} {2015})}\BibitemShut
  {NoStop}%
\bibitem [{\citenamefont {Kr\"amer}\ and\ \citenamefont {del
  Rio}(2016)}]{Resource-currencies}%
  \BibitemOpen
  \bibfield  {author} {\bibinfo {author} {\bibfnamefont {L.}~\bibnamefont
  {Kr\"amer}}\ and\ \bibinfo {author} {\bibfnamefont {L.}~\bibnamefont {del
  Rio}},\ }\href {http://arxiv.org/abs/1605.01064} {\bibfield  {journal}
  {\bibinfo  {journal} {arXiv:1605.01064 [quant-ph]}\ } (\bibinfo {year}
  {2016})}\BibitemShut {NoStop}%
\bibitem [{\citenamefont {Janzing}\ \emph {et~al.}(2000)\citenamefont
  {Janzing}, \citenamefont {Wocjan}, \citenamefont {Zeier}, \citenamefont
  {Geiss},\ and\ \citenamefont {Beth}}]{Athermality1}%
  \BibitemOpen
  \bibfield  {author} {\bibinfo {author} {\bibfnamefont {D.}~\bibnamefont
  {Janzing}}, \bibinfo {author} {\bibfnamefont {P.}~\bibnamefont {Wocjan}},
  \bibinfo {author} {\bibfnamefont {R.}~\bibnamefont {Zeier}}, \bibinfo
  {author} {\bibfnamefont {R.}~\bibnamefont {Geiss}}, \ and\ \bibinfo {author}
  {\bibfnamefont {T.}~\bibnamefont {Beth}},\ }\href {\doibase
  10.1023/A:1026422630734} {\bibfield  {journal} {\bibinfo  {journal} {Int. J.
  Theor. Phys.}\ }\textbf {\bibinfo {volume} {39}},\ \bibinfo {pages} {2717}
  (\bibinfo {year} {2000})}\BibitemShut {NoStop}%
\bibitem [{\citenamefont {Horodecki}\ and\ \citenamefont
  {Oppenheim}(2013{\natexlab{b}})}]{Horodecki-Oppenheim-2}%
  \BibitemOpen
  \bibfield  {author} {\bibinfo {author} {\bibfnamefont {M.}~\bibnamefont
  {Horodecki}}\ and\ \bibinfo {author} {\bibfnamefont {J.}~\bibnamefont
  {Oppenheim}},\ }\href {\doibase 10.1038/ncomms3059} {\bibfield  {journal}
  {\bibinfo  {journal} {Nat. Commun.}\ }\textbf {\bibinfo {volume} {4}},\
  \bibinfo {pages} {2059} (\bibinfo {year} {2013}{\natexlab{b}})}\BibitemShut
  {NoStop}%
\bibitem [{\citenamefont {Brand\~ao}\ \emph {et~al.}(2015)\citenamefont
  {Brand\~ao}, \citenamefont {Horodecki}, \citenamefont {Ng}, \citenamefont
  {Oppenheim},\ and\ \citenamefont {Wehner}}]{2ndlaws}%
  \BibitemOpen
  \bibfield  {author} {\bibinfo {author} {\bibfnamefont {F.~G. S.~L.}\
  \bibnamefont {Brand\~ao}}, \bibinfo {author} {\bibfnamefont {M.}~\bibnamefont
  {Horodecki}}, \bibinfo {author} {\bibfnamefont {N.}~\bibnamefont {Ng}},
  \bibinfo {author} {\bibfnamefont {J.}~\bibnamefont {Oppenheim}}, \ and\
  \bibinfo {author} {\bibfnamefont {S.}~\bibnamefont {Wehner}},\ }\href
  {\doibase 10.1073/pnas.1411728112} {\bibfield  {journal} {\bibinfo  {journal}
  {Proc. Natl. Acad. Sci.}\ }\textbf {\bibinfo {volume} {112}},\ \bibinfo
  {pages} {3275} (\bibinfo {year} {2015})}\BibitemShut {NoStop}%
\bibitem [{\citenamefont {\'{C}wikli\'{n}ski}\ \emph
  {et~al.}(2015)\citenamefont {\'{C}wikli\'{n}ski}, \citenamefont
  {Studzi\'{n}ski}, \citenamefont {Horodecki},\ and\ \citenamefont
  {Oppenheim}}]{quantum2ndlaw}%
  \BibitemOpen
  \bibfield  {author} {\bibinfo {author} {\bibfnamefont {P.}~\bibnamefont
  {\'{C}wikli\'{n}ski}}, \bibinfo {author} {\bibfnamefont {M.}~\bibnamefont
  {Studzi\'{n}ski}}, \bibinfo {author} {\bibfnamefont {M.}~\bibnamefont
  {Horodecki}}, \ and\ \bibinfo {author} {\bibfnamefont {J.}~\bibnamefont
  {Oppenheim}},\ }\href {\doibase 10.1103/PhysRevLett.115.210403} {\bibfield
  {journal} {\bibinfo  {journal} {Phys. Rev. Lett.}\ }\textbf {\bibinfo
  {volume} {115}},\ \bibinfo {pages} {210403} (\bibinfo {year}
  {2015})}\BibitemShut {NoStop}%
\bibitem [{\citenamefont {Lostaglio}\ \emph
  {et~al.}(2015{\natexlab{b}})\citenamefont {Lostaglio}, \citenamefont
  {Jennings},\ and\ \citenamefont {Rudolph}}]{Lostaglio-Jennings-Rudolph}%
  \BibitemOpen
  \bibfield  {author} {\bibinfo {author} {\bibfnamefont {M.}~\bibnamefont
  {Lostaglio}}, \bibinfo {author} {\bibfnamefont {D.}~\bibnamefont {Jennings}},
  \ and\ \bibinfo {author} {\bibfnamefont {T.}~\bibnamefont {Rudolph}},\ }\href
  {\doibase 10.1038/ncomms7383} {\bibfield  {journal} {\bibinfo  {journal}
  {Nat. Commun.}\ }\textbf {\bibinfo {volume} {6}},\ \bibinfo {pages} {6383}
  (\bibinfo {year} {2015}{\natexlab{b}})}\BibitemShut {NoStop}%
\bibitem [{\citenamefont {Lostaglio}\ \emph
  {et~al.}(2015{\natexlab{c}})\citenamefont {Lostaglio}, \citenamefont
  {Korzekwa}, \citenamefont {Jennings},\ and\ \citenamefont
  {Rudolph}}]{Lostaglio-coherence}%
  \BibitemOpen
  \bibfield  {author} {\bibinfo {author} {\bibfnamefont {M.}~\bibnamefont
  {Lostaglio}}, \bibinfo {author} {\bibfnamefont {K.}~\bibnamefont {Korzekwa}},
  \bibinfo {author} {\bibfnamefont {D.}~\bibnamefont {Jennings}}, \ and\
  \bibinfo {author} {\bibfnamefont {T.}~\bibnamefont {Rudolph}},\ }\href
  {\doibase 10.1103/PhysRevX.5.021001} {\bibfield  {journal} {\bibinfo
  {journal} {Phys. Rev. X}\ }\textbf {\bibinfo {volume} {5}},\ \bibinfo {pages}
  {021001} (\bibinfo {year} {2015}{\natexlab{c}})}\BibitemShut {NoStop}%
\bibitem [{\citenamefont {Yunger~Halpern}\ and\ \citenamefont
  {Renes}(2016)}]{Nicole-beyond}%
  \BibitemOpen
  \bibfield  {author} {\bibinfo {author} {\bibfnamefont {N.}~\bibnamefont
  {Yunger~Halpern}}\ and\ \bibinfo {author} {\bibfnamefont {J.~M.}\
  \bibnamefont {Renes}},\ }\href {\doibase 10.1103/PhysRevE.93.022126}
  {\bibfield  {journal} {\bibinfo  {journal} {Phys. Rev. E}\ }\textbf {\bibinfo
  {volume} {93}},\ \bibinfo {pages} {022126} (\bibinfo {year}
  {2016})}\BibitemShut {NoStop}%
\bibitem [{\citenamefont {Yunger~Halpern}\ \emph {et~al.}(2016)\citenamefont
  {Yunger~Halpern}, \citenamefont {Faist}, \citenamefont {Oppenheim},\ and\
  \citenamefont {Winter}}]{Nicole-non-commuting}%
  \BibitemOpen
  \bibfield  {author} {\bibinfo {author} {\bibfnamefont {N.}~\bibnamefont
  {Yunger~Halpern}}, \bibinfo {author} {\bibfnamefont {P.}~\bibnamefont
  {Faist}}, \bibinfo {author} {\bibfnamefont {J.}~\bibnamefont {Oppenheim}}, \
  and\ \bibinfo {author} {\bibfnamefont {A.}~\bibnamefont {Winter}},\ }\href
  {\doibase 10.1038/ncomms12051} {\bibfield  {journal} {\bibinfo  {journal}
  {Nat. Commun.}\ }\textbf {\bibinfo {volume} {7}},\ \bibinfo {pages} {12051}
  (\bibinfo {year} {2016})}\BibitemShut {NoStop}%
\bibitem [{\citenamefont {Guryanova}\ \emph {et~al.}(2016)\citenamefont
  {Guryanova}, \citenamefont {Popescu}, \citenamefont {Short}, \citenamefont
  {Silva},\ and\ \citenamefont {Skrzypczyk}}]{Non-commuting-Bristol}%
  \BibitemOpen
  \bibfield  {author} {\bibinfo {author} {\bibfnamefont {Y.}~\bibnamefont
  {Guryanova}}, \bibinfo {author} {\bibfnamefont {S.}~\bibnamefont {Popescu}},
  \bibinfo {author} {\bibfnamefont {A.~J.}\ \bibnamefont {Short}}, \bibinfo
  {author} {\bibfnamefont {R.}~\bibnamefont {Silva}}, \ and\ \bibinfo {author}
  {\bibfnamefont {P.}~\bibnamefont {Skrzypczyk}},\ }\href {\doibase
  10.1038/ncomms12049} {\bibfield  {journal} {\bibinfo  {journal} {Nat.
  Commun.}\ }\textbf {\bibinfo {volume} {7}},\ \bibinfo {pages} {12049}
  (\bibinfo {year} {2016})}\BibitemShut {NoStop}%
\bibitem [{\citenamefont {Lostaglio}\ \emph {et~al.}(2017)\citenamefont
  {Lostaglio}, \citenamefont {Jennings},\ and\ \citenamefont
  {Rudolph}}]{David-non-commuting}%
  \BibitemOpen
  \bibfield  {author} {\bibinfo {author} {\bibfnamefont {M.}~\bibnamefont
  {Lostaglio}}, \bibinfo {author} {\bibfnamefont {D.}~\bibnamefont {Jennings}},
  \ and\ \bibinfo {author} {\bibfnamefont {T.}~\bibnamefont {Rudolph}},\ }\href
  {\doibase 10.1088/1367-2630/aa617f} {\bibfield  {journal} {\bibinfo
  {journal} {New J. Phys.}\ }\textbf {\bibinfo {volume} {19}},\ \bibinfo
  {pages} {043008} (\bibinfo {year} {2017})}\BibitemShut {NoStop}%
\bibitem [{\citenamefont {Richens}\ \emph {et~al.}(2017)\citenamefont
  {Richens}, \citenamefont {Alhambra},\ and\ \citenamefont
  {Masanes}}]{Richens-2ndlaw}%
  \BibitemOpen
  \bibfield  {author} {\bibinfo {author} {\bibfnamefont {J.~G.}\ \bibnamefont
  {Richens}}, \bibinfo {author} {\bibfnamefont {A.~M.}\ \bibnamefont
  {Alhambra}}, \ and\ \bibinfo {author} {\bibfnamefont {L.}~\bibnamefont
  {Masanes}},\ }\href {http://arxiv.org/abs/1702.03357} {\bibfield  {journal}
  {\bibinfo  {journal} {arXiv:1702.03357 [quant-ph]}\ } (\bibinfo {year}
  {2017})}\BibitemShut {NoStop}%
\bibitem [{\citenamefont {Masanes}\ and\ \citenamefont
  {Oppenheim}(2017)}]{3rdlaw}%
  \BibitemOpen
  \bibfield  {author} {\bibinfo {author} {\bibfnamefont {L.}~\bibnamefont
  {Masanes}}\ and\ \bibinfo {author} {\bibfnamefont {J.}~\bibnamefont
  {Oppenheim}},\ }\href {\doibase 10.1038/ncomms14538} {\bibfield  {journal}
  {\bibinfo  {journal} {Nat. Commun.}\ }\textbf {\bibinfo {volume} {8}},\
  \bibinfo {pages} {14538} (\bibinfo {year} {2017})}\BibitemShut {NoStop}%
\bibitem [{\citenamefont {Brand\~ao}\ \emph {et~al.}(2013)\citenamefont
  {Brand\~ao}, \citenamefont {Horodecki}, \citenamefont {Oppenheim},
  \citenamefont {Renes},\ and\ \citenamefont {Spekkens}}]{Athermality2}%
  \BibitemOpen
  \bibfield  {author} {\bibinfo {author} {\bibfnamefont {F.~G. S.~L.}\
  \bibnamefont {Brand\~ao}}, \bibinfo {author} {\bibfnamefont {M.}~\bibnamefont
  {Horodecki}}, \bibinfo {author} {\bibfnamefont {J.}~\bibnamefont
  {Oppenheim}}, \bibinfo {author} {\bibfnamefont {J.~M.}\ \bibnamefont
  {Renes}}, \ and\ \bibinfo {author} {\bibfnamefont {R.~W.}\ \bibnamefont
  {Spekkens}},\ }\href {\doibase 10.1103/PhysRevLett.111.250404} {\bibfield
  {journal} {\bibinfo  {journal} {Phys. Rev. Lett.}\ }\textbf {\bibinfo
  {volume} {111}},\ \bibinfo {pages} {250404} (\bibinfo {year}
  {2013})}\BibitemShut {NoStop}%
\bibitem [{\citenamefont {Egloff}\ \emph {et~al.}(2015)\citenamefont {Egloff},
  \citenamefont {Dahlsten}, \citenamefont {Renner},\ and\ \citenamefont
  {Vedral}}]{Egloff}%
  \BibitemOpen
  \bibfield  {author} {\bibinfo {author} {\bibfnamefont {D.}~\bibnamefont
  {Egloff}}, \bibinfo {author} {\bibfnamefont {O.~C.~O.}\ \bibnamefont
  {Dahlsten}}, \bibinfo {author} {\bibfnamefont {R.}~\bibnamefont {Renner}}, \
  and\ \bibinfo {author} {\bibfnamefont {V.}~\bibnamefont {Vedral}},\ }\href
  {\doibase 10.1088/1367-2630/17/7/073001} {\bibfield  {journal} {\bibinfo
  {journal} {New J. Phys.}\ }\textbf {\bibinfo {volume} {17}},\ \bibinfo
  {pages} {073001} (\bibinfo {year} {2015})}\BibitemShut {NoStop}%
\bibitem [{\citenamefont {Renes}(2014)}]{Renes}%
  \BibitemOpen
  \bibfield  {author} {\bibinfo {author} {\bibfnamefont {J.~M.}\ \bibnamefont
  {Renes}},\ }\href {\doibase 10.1140/epjp/i2014-14153-8} {\bibfield  {journal}
  {\bibinfo  {journal} {Eur. Phys. J. Plus}\ }\textbf {\bibinfo {volume}
  {129}},\ \bibinfo {pages} {1} (\bibinfo {year} {2014})}\BibitemShut {NoStop}%
\bibitem [{\citenamefont {Faist}\ \emph {et~al.}(2015)\citenamefont {Faist},
  \citenamefont {Oppenheim},\ and\ \citenamefont
  {Renner}}]{Gibbs-preserving-maps}%
  \BibitemOpen
  \bibfield  {author} {\bibinfo {author} {\bibfnamefont {P.}~\bibnamefont
  {Faist}}, \bibinfo {author} {\bibfnamefont {J.}~\bibnamefont {Oppenheim}}, \
  and\ \bibinfo {author} {\bibfnamefont {R.}~\bibnamefont {Renner}},\ }\href
  {\doibase 10.1088/1367-2630/17/4/043003} {\bibfield  {journal} {\bibinfo
  {journal} {New J. Phys.}\ }\textbf {\bibinfo {volume} {17}},\ \bibinfo
  {pages} {043003} (\bibinfo {year} {2015})}\BibitemShut {NoStop}%
\bibitem [{\citenamefont {Alhambra}\ \emph {et~al.}(2015)\citenamefont
  {Alhambra}, \citenamefont {Wehner}, \citenamefont {Wilde},\ and\
  \citenamefont {Woods}}]{Alhambra-reversibility}%
  \BibitemOpen
  \bibfield  {author} {\bibinfo {author} {\bibfnamefont {A.~M.}\ \bibnamefont
  {Alhambra}}, \bibinfo {author} {\bibfnamefont {S.}~\bibnamefont {Wehner}},
  \bibinfo {author} {\bibfnamefont {M.~M.}\ \bibnamefont {Wilde}}, \ and\
  \bibinfo {author} {\bibfnamefont {M.~P.}\ \bibnamefont {Woods}},\ }\href
  {http://arxiv.org/abs/1506.08145} {\bibfield  {journal} {\bibinfo  {journal}
  {arXiv:1506.08145 [quant-ph]}\ } (\bibinfo {year} {2015})}\BibitemShut
  {NoStop}%
\bibitem [{\citenamefont {M\"uller}\ and\ \citenamefont
  {Pastena}(2016)}]{Muller-generalization}%
  \BibitemOpen
  \bibfield  {author} {\bibinfo {author} {\bibfnamefont {M.~P.}\ \bibnamefont
  {M\"uller}}\ and\ \bibinfo {author} {\bibfnamefont {M.}~\bibnamefont
  {Pastena}},\ }\href {\doibase 10.1109/TIT.2016.2528285} {\bibfield  {journal}
  {\bibinfo  {journal} {IEEE Trans. Inf. Theory}\ }\textbf {\bibinfo {volume}
  {62}},\ \bibinfo {pages} {1711} (\bibinfo {year} {2016})}\BibitemShut
  {NoStop}%
\bibitem [{\citenamefont {Sparaciari}\ \emph {et~al.}(2016)\citenamefont
  {Sparaciari}, \citenamefont {Oppenheim},\ and\ \citenamefont
  {Fritz}}]{Sparaciari}%
  \BibitemOpen
  \bibfield  {author} {\bibinfo {author} {\bibfnamefont {C.}~\bibnamefont
  {Sparaciari}}, \bibinfo {author} {\bibfnamefont {J.}~\bibnamefont
  {Oppenheim}}, \ and\ \bibinfo {author} {\bibfnamefont {T.}~\bibnamefont
  {Fritz}},\ }\href {http://arxiv.org/abs/1607.01302} {\bibfield  {journal}
  {\bibinfo  {journal} {arXiv:1607.01302 [quant-ph]}\ } (\bibinfo {year}
  {2016})}\BibitemShut {NoStop}%
\bibitem [{\citenamefont {Uhlmann}(1971)}]{Uhlmann1}%
  \BibitemOpen
  \bibfield  {author} {\bibinfo {author} {\bibfnamefont {A.}~\bibnamefont
  {Uhlmann}},\ }\href@noop {} {\bibfield  {journal} {\bibinfo  {journal} {Wiss.
  Z. Karl-Marx Univ. Leipzig}\ }\textbf {\bibinfo {volume} {20}},\ \bibinfo
  {pages} {633} (\bibinfo {year} {1971})}\BibitemShut {NoStop}%
\bibitem [{\citenamefont {Uhlmann}(1972)}]{Uhlmann2}%
  \BibitemOpen
  \bibfield  {author} {\bibinfo {author} {\bibfnamefont {A.}~\bibnamefont
  {Uhlmann}},\ }\href@noop {} {\bibfield  {journal} {\bibinfo  {journal} {Wiss.
  Z. Karl-Marx Univ. Leipzig.}\ }\textbf {\bibinfo {volume} {21}},\ \bibinfo
  {pages} {421} (\bibinfo {year} {1972})}\BibitemShut {NoStop}%
\bibitem [{\citenamefont {Uhlmann}(1973)}]{Uhlmann3}%
  \BibitemOpen
  \bibfield  {author} {\bibinfo {author} {\bibfnamefont {A.}~\bibnamefont
  {Uhlmann}},\ }\href@noop {} {\bibfield  {journal} {\bibinfo  {journal} {Wiss.
  Z. Karl-Marx Univ. Leipzig}\ }\textbf {\bibinfo {volume} {22}},\ \bibinfo
  {pages} {139} (\bibinfo {year} {1973})}\BibitemShut {NoStop}%
\bibitem [{\citenamefont {Horodecki}\ \emph
  {et~al.}(2003{\natexlab{a}})\citenamefont {Horodecki}, \citenamefont
  {Horodecki}, \citenamefont {Horodecki}, \citenamefont {Horodecki},
  \citenamefont {Oppenheim}, \citenamefont {Sen(De)},\ and\ \citenamefont
  {Sen}}]{Local-Information}%
  \BibitemOpen
  \bibfield  {author} {\bibinfo {author} {\bibfnamefont {M.}~\bibnamefont
  {Horodecki}}, \bibinfo {author} {\bibfnamefont {K.}~\bibnamefont
  {Horodecki}}, \bibinfo {author} {\bibfnamefont {P.}~\bibnamefont
  {Horodecki}}, \bibinfo {author} {\bibfnamefont {R.}~\bibnamefont
  {Horodecki}}, \bibinfo {author} {\bibfnamefont {J.}~\bibnamefont
  {Oppenheim}}, \bibinfo {author} {\bibfnamefont {A.}~\bibnamefont {Sen(De)}},
  \ and\ \bibinfo {author} {\bibfnamefont {U.}~\bibnamefont {Sen}},\ }\href
  {\doibase 10.1103/PhysRevLett.90.100402} {\bibfield  {journal} {\bibinfo
  {journal} {Phys. Rev. Lett.}\ }\textbf {\bibinfo {volume} {90}},\ \bibinfo
  {pages} {100402} (\bibinfo {year} {2003}{\natexlab{a}})}\BibitemShut
  {NoStop}%
\bibitem [{\citenamefont {Horodecki}\ \emph
  {et~al.}(2003{\natexlab{b}})\citenamefont {Horodecki}, \citenamefont
  {Horodecki},\ and\ \citenamefont {Oppenheim}}]{Horodecki-Oppenheim}%
  \BibitemOpen
  \bibfield  {author} {\bibinfo {author} {\bibfnamefont {M.}~\bibnamefont
  {Horodecki}}, \bibinfo {author} {\bibfnamefont {P.}~\bibnamefont
  {Horodecki}}, \ and\ \bibinfo {author} {\bibfnamefont {J.}~\bibnamefont
  {Oppenheim}},\ }\href {\doibase 10.1103/PhysRevA.67.062104} {\bibfield
  {journal} {\bibinfo  {journal} {Phys. Rev. A}\ }\textbf {\bibinfo {volume}
  {67}},\ \bibinfo {pages} {062104} (\bibinfo {year}
  {2003}{\natexlab{b}})}\BibitemShut {NoStop}%
\bibitem [{\citenamefont {Gour}\ \emph {et~al.}(2015)\citenamefont {Gour},
  \citenamefont {M\"uller}, \citenamefont {Narasimhachar}, \citenamefont
  {Spekkens},\ and\ \citenamefont {Yunger~Halpern}}]{Nicole}%
  \BibitemOpen
  \bibfield  {author} {\bibinfo {author} {\bibfnamefont {G.}~\bibnamefont
  {Gour}}, \bibinfo {author} {\bibfnamefont {M.~P.}\ \bibnamefont {M\"uller}},
  \bibinfo {author} {\bibfnamefont {V.}~\bibnamefont {Narasimhachar}}, \bibinfo
  {author} {\bibfnamefont {R.~W.}\ \bibnamefont {Spekkens}}, \ and\ \bibinfo
  {author} {\bibfnamefont {N.}~\bibnamefont {Yunger~Halpern}},\ }\href
  {\doibase 10.1016/j.physrep.2015.04.003} {\bibfield  {journal} {\bibinfo
  {journal} {Phys. Rep.}\ }\textbf {\bibinfo {volume} {583}},\ \bibinfo {pages}
  {1} (\bibinfo {year} {2015})}\BibitemShut {NoStop}%
\bibitem [{\citenamefont {Landau}\ and\ \citenamefont
  {Streater}(1993)}]{Streater}%
  \BibitemOpen
  \bibfield  {author} {\bibinfo {author} {\bibfnamefont {L.~J.}\ \bibnamefont
  {Landau}}\ and\ \bibinfo {author} {\bibfnamefont {R.~F.}\ \bibnamefont
  {Streater}},\ }\href {\doibase 10.1016/0024-3795(93)90274-R} {\bibfield
  {journal} {\bibinfo  {journal} {Linear Algebra Appl.}\ }\textbf {\bibinfo
  {volume} {193}},\ \bibinfo {pages} {107} (\bibinfo {year}
  {1993})}\BibitemShut {NoStop}%
\bibitem [{\citenamefont {Mendl}\ and\ \citenamefont
  {Wolf}(2009)}]{Mendl-Wolf}%
  \BibitemOpen
  \bibfield  {author} {\bibinfo {author} {\bibfnamefont {C.~B.}\ \bibnamefont
  {Mendl}}\ and\ \bibinfo {author} {\bibfnamefont {M.~M.}\ \bibnamefont
  {Wolf}},\ }\href {\doibase 10.1007/s00220-009-0824-2} {\bibfield  {journal}
  {\bibinfo  {journal} {Commun. Math. Phys.}\ }\textbf {\bibinfo {volume}
  {289}},\ \bibinfo {pages} {1057} (\bibinfo {year} {2009})}\BibitemShut
  {NoStop}%
\bibitem [{\citenamefont {Shor}(2010)}]{Shor}%
  \BibitemOpen
  \bibfield  {author} {\bibinfo {author} {\bibfnamefont {P.~W.}\ \bibnamefont
  {Shor}},\ }\href@noop {} {\enquote {\bibinfo {title} {Structure of unital
  maps and the asymptotic quantum birkhoff conjecture},}\ }\bibinfo
  {howpublished} {presentation} (\bibinfo {year} {2010})\BibitemShut {NoStop}%
\bibitem [{\citenamefont {Haagerup}\ and\ \citenamefont
  {Musat}(2011)}]{Haagerup-Musat}%
  \BibitemOpen
  \bibfield  {author} {\bibinfo {author} {\bibfnamefont {U.}~\bibnamefont
  {Haagerup}}\ and\ \bibinfo {author} {\bibfnamefont {M.}~\bibnamefont
  {Musat}},\ }\href {\doibase 10.1007/s00220-011-1216-y} {\bibfield  {journal}
  {\bibinfo  {journal} {Commun. Math. Phys.}\ }\textbf {\bibinfo {volume}
  {303}},\ \bibinfo {pages} {555} (\bibinfo {year} {2011})}\BibitemShut
  {NoStop}%
\bibitem [{\citenamefont {Streltsov}\ \emph {et~al.}(2016)\citenamefont
  {Streltsov}, \citenamefont {Kampermann}, \citenamefont {W{\"o}lk},
  \citenamefont {Gessner},\ and\ \citenamefont
  {Bru{\ss}}}]{streltsov2016maximal}%
  \BibitemOpen
  \bibfield  {author} {\bibinfo {author} {\bibfnamefont {A.}~\bibnamefont
  {Streltsov}}, \bibinfo {author} {\bibfnamefont {H.}~\bibnamefont
  {Kampermann}}, \bibinfo {author} {\bibfnamefont {S.}~\bibnamefont
  {W{\"o}lk}}, \bibinfo {author} {\bibfnamefont {M.}~\bibnamefont {Gessner}}, \
  and\ \bibinfo {author} {\bibfnamefont {D.}~\bibnamefont {Bru{\ss}}},\ }\href
  {http://arxiv.org/abs/1612.07570} {\bibfield  {journal} {\bibinfo  {journal}
  {arXiv:1612.07570 [quant-ph]}\ } (\bibinfo {year} {2016})}\BibitemShut
  {NoStop}%
\bibitem [{\citenamefont {{Hardy}}(2001)}]{Hardy-informational-1}%
  \BibitemOpen
  \bibfield  {author} {\bibinfo {author} {\bibfnamefont {L.}~\bibnamefont
  {{Hardy}}},\ }\href {http://arxiv.org/abs/quant-ph/0101012} {\bibfield
  {journal} {\bibinfo  {journal} {arXiv quant-ph/0101012}\ } (\bibinfo {year}
  {2001})}\BibitemShut {NoStop}%
\bibitem [{\citenamefont {Barrett}(2007)}]{Barrett}%
  \BibitemOpen
  \bibfield  {author} {\bibinfo {author} {\bibfnamefont {J.}~\bibnamefont
  {Barrett}},\ }\href {\doibase 10.1103/PhysRevA.75.032304} {\bibfield
  {journal} {\bibinfo  {journal} {Phys. Rev. A}\ }\textbf {\bibinfo {volume}
  {75}},\ \bibinfo {pages} {032304} (\bibinfo {year} {2007})}\BibitemShut
  {NoStop}%
\bibitem [{\citenamefont {Barnum}\ \emph {et~al.}(2007)\citenamefont {Barnum},
  \citenamefont {Barrett}, \citenamefont {Leifer},\ and\ \citenamefont
  {Wilce}}]{Barnum-1}%
  \BibitemOpen
  \bibfield  {author} {\bibinfo {author} {\bibfnamefont {H.}~\bibnamefont
  {Barnum}}, \bibinfo {author} {\bibfnamefont {J.}~\bibnamefont {Barrett}},
  \bibinfo {author} {\bibfnamefont {M.}~\bibnamefont {Leifer}}, \ and\ \bibinfo
  {author} {\bibfnamefont {A.}~\bibnamefont {Wilce}},\ }\href {\doibase
  10.1103/PhysRevLett.99.240501} {\bibfield  {journal} {\bibinfo  {journal}
  {Phys. Rev. Lett.}\ }\textbf {\bibinfo {volume} {99}},\ \bibinfo {pages}
  {240501} (\bibinfo {year} {2007})}\BibitemShut {NoStop}%
\bibitem [{\citenamefont {Chiribella}\ \emph {et~al.}(2010)\citenamefont
  {Chiribella}, \citenamefont {D'Ariano},\ and\ \citenamefont
  {Perinotti}}]{Chiribella-purification}%
  \BibitemOpen
  \bibfield  {author} {\bibinfo {author} {\bibfnamefont {G.}~\bibnamefont
  {Chiribella}}, \bibinfo {author} {\bibfnamefont {G.~M.}\ \bibnamefont
  {D'Ariano}}, \ and\ \bibinfo {author} {\bibfnamefont {P.}~\bibnamefont
  {Perinotti}},\ }\href {\doibase 10.1103/PhysRevA.81.062348} {\bibfield
  {journal} {\bibinfo  {journal} {Phys. Rev. A}\ }\textbf {\bibinfo {volume}
  {81}},\ \bibinfo {pages} {062348} (\bibinfo {year} {2010})}\BibitemShut
  {NoStop}%
\bibitem [{\citenamefont {Chiribella}(2014)}]{Chiribella14}%
  \BibitemOpen
  \bibfield  {author} {\bibinfo {author} {\bibfnamefont {G.}~\bibnamefont
  {Chiribella}},\ }in\ \href {\doibase 10.4204/EPTCS.172.1} {\emph {\bibinfo
  {booktitle} {{\rm Proceedings 11th workshop on} Quantum Physics and Logic,
  {\rm Kyoto, Japan, 4-6th June 2014}}}},\ \bibinfo {series} {Electronic
  Proceedings in Theoretical Computer Science}, Vol.\ \bibinfo {volume} {172},\
  \bibinfo {editor} {edited by\ \bibinfo {editor} {\bibfnamefont
  {B.}~\bibnamefont {Coecke}}, \bibinfo {editor} {\bibfnamefont
  {I.}~\bibnamefont {Hasuo}}, \ and\ \bibinfo {editor} {\bibfnamefont
  {P.}~\bibnamefont {Panangaden}}}\ (\bibinfo {year} {2014})\ pp.\ \bibinfo
  {pages} {1--14}\BibitemShut {NoStop}%
\bibitem [{\citenamefont {Chiribella}\ \emph {et~al.}(2016)\citenamefont
  {Chiribella}, \citenamefont {D'Ariano},\ and\ \citenamefont
  {Perinotti}}]{QuantumFromPrinciples}%
  \BibitemOpen
  \bibfield  {author} {\bibinfo {author} {\bibfnamefont {G.}~\bibnamefont
  {Chiribella}}, \bibinfo {author} {\bibfnamefont {G.~M.}\ \bibnamefont
  {D'Ariano}}, \ and\ \bibinfo {author} {\bibfnamefont {P.}~\bibnamefont
  {Perinotti}},\ }\enquote {\bibinfo {title} {Quantum theory: Informational
  foundations and foils},}\ \ (\bibinfo  {publisher} {Springer Netherlands},\
  \bibinfo {address} {Dordrecht},\ \bibinfo {year} {2016})\ Chap.\ \bibinfo
  {chapter} {Quantum from Principles}, pp.\ \bibinfo {pages}
  {171--221}\BibitemShut {NoStop}%
\bibitem [{\citenamefont {Chiribella}\ and\ \citenamefont
  {Spekkens}(2016)}]{chiribella2016quantum}%
  \BibitemOpen
  \bibinfo {editor} {\bibfnamefont {G.}~\bibnamefont {Chiribella}}\ and\
  \bibinfo {editor} {\bibfnamefont {R.~W.}\ \bibnamefont {Spekkens}},\ eds.,\
  \href {\doibase 10.1007/978-94-017-7303-4} {\emph {\bibinfo {title} {Quantum
  Theory: Informational Foundations and Foils}}},\ \bibinfo {series}
  {Fundamental Theories of Physics}, Vol.\ \bibinfo {volume} {181}\ (\bibinfo
  {publisher} {Springer Netherlands},\ \bibinfo {address} {Dordrecht},\
  \bibinfo {year} {2016})\BibitemShut {NoStop}%
\bibitem [{\citenamefont {D'Ariano}\ \emph {et~al.}(2017)\citenamefont
  {D'Ariano}, \citenamefont {Chiribella},\ and\ \citenamefont
  {Perinotti}}]{chiribella2017quantum}%
  \BibitemOpen
  \bibfield  {author} {\bibinfo {author} {\bibfnamefont {G.~M.}\ \bibnamefont
  {D'Ariano}}, \bibinfo {author} {\bibfnamefont {G.}~\bibnamefont
  {Chiribella}}, \ and\ \bibinfo {author} {\bibfnamefont {P.}~\bibnamefont
  {Perinotti}},\ }\href@noop {} {\emph {\bibinfo {title} {Quantum Theory from
  First Principles: An Informational Approach}}}\ (\bibinfo  {publisher}
  {Cambridge University Press},\ \bibinfo {address} {Cambridge},\ \bibinfo
  {year} {2017})\BibitemShut {NoStop}%
\bibitem [{\citenamefont {Chiribella}\ and\ \citenamefont
  {Scandolo}(2016)}]{TowardsThermo}%
  \BibitemOpen
  \bibfield  {author} {\bibinfo {author} {\bibfnamefont {G.}~\bibnamefont
  {Chiribella}}\ and\ \bibinfo {author} {\bibfnamefont {C.~M.}\ \bibnamefont
  {Scandolo}},\ }\href {http://arxiv.org/abs/1608.04459} {\bibfield  {journal}
  {\bibinfo  {journal} {arXiv:1608.04459 [quant-ph]}\ } (\bibinfo {year}
  {2016})}\BibitemShut {NoStop}%
\bibitem [{\citenamefont {Lee}\ and\ \citenamefont
  {Selby}(2016{\natexlab{a}})}]{Control-reversible}%
  \BibitemOpen
  \bibfield  {author} {\bibinfo {author} {\bibfnamefont {C.~M.}\ \bibnamefont
  {Lee}}\ and\ \bibinfo {author} {\bibfnamefont {J.~H.}\ \bibnamefont
  {Selby}},\ }\href {\doibase 10.1088/1367-2630/18/3/033023} {\bibfield
  {journal} {\bibinfo  {journal} {New J. Phys.}\ }\textbf {\bibinfo {volume}
  {18}},\ \bibinfo {pages} {033023} (\bibinfo {year}
  {2016}{\natexlab{a}})}\BibitemShut {NoStop}%
\bibitem [{\citenamefont {Lee}\ and\ \citenamefont
  {Selby}(2016{\natexlab{b}})}]{Lee-Selby-Grover}%
  \BibitemOpen
  \bibfield  {author} {\bibinfo {author} {\bibfnamefont {C.~M.}\ \bibnamefont
  {Lee}}\ and\ \bibinfo {author} {\bibfnamefont {J.~H.}\ \bibnamefont
  {Selby}},\ }\href {\doibase 10.1088/1367-2630/18/9/093047} {\bibfield
  {journal} {\bibinfo  {journal} {New J. Phys.}\ }\textbf {\bibinfo {volume}
  {18}},\ \bibinfo {pages} {093047} (\bibinfo {year}
  {2016}{\natexlab{b}})}\BibitemShut {NoStop}%
\bibitem [{\citenamefont {Barnum}\ \emph {et~al.}(2017)\citenamefont {Barnum},
  \citenamefont {Lee}, \citenamefont {Scandolo},\ and\ \citenamefont
  {Selby}}]{HOP}%
  \BibitemOpen
  \bibfield  {author} {\bibinfo {author} {\bibfnamefont {H.}~\bibnamefont
  {Barnum}}, \bibinfo {author} {\bibfnamefont {C.}~\bibnamefont {Lee}},
  \bibinfo {author} {\bibfnamefont {C.~M.}\ \bibnamefont {Scandolo}}, \ and\
  \bibinfo {author} {\bibfnamefont {J.}~\bibnamefont {Selby}},\ }\href
  {\doibase 10.3390/e19060253} {\bibfield  {journal} {\bibinfo  {journal}
  {Entropy}\ }\textbf {\bibinfo {volume} {19}},\ \bibinfo {pages} {253}
  (\bibinfo {year} {2017})}\BibitemShut {NoStop}%
\bibitem [{\citenamefont {Marshall}\ \emph {et~al.}(2011)\citenamefont
  {Marshall}, \citenamefont {Olkin},\ and\ \citenamefont {Arnold}}]{Olkin}%
  \BibitemOpen
  \bibfield  {author} {\bibinfo {author} {\bibfnamefont {A.~W.}\ \bibnamefont
  {Marshall}}, \bibinfo {author} {\bibfnamefont {I.}~\bibnamefont {Olkin}}, \
  and\ \bibinfo {author} {\bibfnamefont {B.~C.}\ \bibnamefont {Arnold}},\
  }\href {\doibase 10.1007/978-0-387-68276-1} {\emph {\bibinfo {title}
  {Inequalities: Theory of Majorization and Its Applications}}},\ Springer
  Series in Statistics\ (\bibinfo  {publisher} {Springer},\ \bibinfo {address}
  {New York},\ \bibinfo {year} {2011})\BibitemShut {NoStop}%
\bibitem [{\citenamefont {Chiribella}\ and\ \citenamefont
  {Scandolo}(2015{\natexlab{a}})}]{Chiribella-Scandolo-entanglement}%
  \BibitemOpen
  \bibfield  {author} {\bibinfo {author} {\bibfnamefont {G.}~\bibnamefont
  {Chiribella}}\ and\ \bibinfo {author} {\bibfnamefont {C.~M.}\ \bibnamefont
  {Scandolo}},\ }\href {\doibase 10.1088/1367-2630/17/10/103027} {\bibfield
  {journal} {\bibinfo  {journal} {New J. Phys.}\ }\textbf {\bibinfo {volume}
  {17}},\ \bibinfo {pages} {103027} (\bibinfo {year}
  {2015}{\natexlab{a}})}\BibitemShut {NoStop}%
\bibitem [{\citenamefont {Chiribella}\ \emph {et~al.}(2011)\citenamefont
  {Chiribella}, \citenamefont {D'Ariano},\ and\ \citenamefont
  {Perinotti}}]{Chiribella-informational}%
  \BibitemOpen
  \bibfield  {author} {\bibinfo {author} {\bibfnamefont {G.}~\bibnamefont
  {Chiribella}}, \bibinfo {author} {\bibfnamefont {G.~M.}\ \bibnamefont
  {D'Ariano}}, \ and\ \bibinfo {author} {\bibfnamefont {P.}~\bibnamefont
  {Perinotti}},\ }\href {\doibase 10.1103/PhysRevA.84.012311} {\bibfield
  {journal} {\bibinfo  {journal} {Phys. Rev. A}\ }\textbf {\bibinfo {volume}
  {84}},\ \bibinfo {pages} {012311} (\bibinfo {year} {2011})}\BibitemShut
  {NoStop}%
\bibitem [{\citenamefont {Hardy}(2011{\natexlab{a}})}]{hardy2011}%
  \BibitemOpen
  \bibfield  {author} {\bibinfo {author} {\bibfnamefont {L.}~\bibnamefont
  {Hardy}},\ }\enquote {\bibinfo {title} {Foliable operational structures for
  general probabilistic theories},}\ in\ \href {\doibase
  10.1017/CBO9780511976971.013} {\emph {\bibinfo {booktitle} {Deep Beauty:
  Understanding the Quantum World through Mathematical Innovation}}},\ \bibinfo
  {editor} {edited by\ \bibinfo {editor} {\bibfnamefont {H.}~\bibnamefont
  {Halvorson}}}\ (\bibinfo  {publisher} {Cambridge University Press},\ \bibinfo
  {address} {Cambridge},\ \bibinfo {year} {2011})\ pp.\ \bibinfo {pages}
  {409--442}\BibitemShut {NoStop}%
\bibitem [{\citenamefont {Hardy}(2011{\natexlab{b}})}]{Hardy-informational-2}%
  \BibitemOpen
  \bibfield  {author} {\bibinfo {author} {\bibfnamefont {L.}~\bibnamefont
  {Hardy}},\ }\href {http://arxiv.org/abs/1104.2066} {\bibfield  {journal}
  {\bibinfo  {journal} {arXiv:1104.2066 [quant-ph]}\ } (\bibinfo {year}
  {2011}{\natexlab{b}})}\BibitemShut {NoStop}%
\bibitem [{\citenamefont {Hardy}(2016)}]{hardy2013}%
  \BibitemOpen
  \bibfield  {author} {\bibinfo {author} {\bibfnamefont {L.}~\bibnamefont
  {Hardy}},\ }\enquote {\bibinfo {title} {Quantum theory: Informational
  foundations and foils},}\ \ (\bibinfo  {publisher} {Springer Netherlands},\
  \bibinfo {address} {Dordrecht},\ \bibinfo {year} {2016})\ Chap.\ \bibinfo
  {chapter} {Reconstructing Quantum Theory}, pp.\ \bibinfo {pages}
  {223--248}\BibitemShut {NoStop}%
\bibitem [{\citenamefont {Wilce}(2010)}]{Wilce-formalism}%
  \BibitemOpen
  \bibfield  {author} {\bibinfo {author} {\bibfnamefont {A.}~\bibnamefont
  {Wilce}},\ }\href {\doibase 10.1007/s10701-010-9410-x} {\bibfield  {journal}
  {\bibinfo  {journal} {Found. Phys.}\ }\textbf {\bibinfo {volume} {40}},\
  \bibinfo {pages} {434} (\bibinfo {year} {2010})}\BibitemShut {NoStop}%
\bibitem [{\citenamefont {Barnum}\ and\ \citenamefont
  {Wilce}(2011)}]{Barnum-2}%
  \BibitemOpen
  \bibfield  {author} {\bibinfo {author} {\bibfnamefont {H.}~\bibnamefont
  {Barnum}}\ and\ \bibinfo {author} {\bibfnamefont {A.}~\bibnamefont {Wilce}},\
  }\href {\doibase 10.1016/j.entcs.2011.01.002} {\bibfield  {journal} {\bibinfo
   {journal} {Electronic Notes in Theoretical Computer Science}\ }\textbf
  {\bibinfo {volume} {270}},\ \bibinfo {pages} {3} (\bibinfo {year} {2011})},\
  \bibinfo {note} {proceedings of the Joint 5th International Workshop on
  Quantum Physics and Logic and 4th Workshop on Developments in Computational
  Models (QPL/DCM 2008)}\BibitemShut {NoStop}%
\bibitem [{\citenamefont {Barnum}\ and\ \citenamefont
  {Wilce}(2016)}]{Barnum2016}%
  \BibitemOpen
  \bibfield  {author} {\bibinfo {author} {\bibfnamefont {H.}~\bibnamefont
  {Barnum}}\ and\ \bibinfo {author} {\bibfnamefont {A.}~\bibnamefont {Wilce}},\
  }\enquote {\bibinfo {title} {Post-classical probability theory},}\ in\ \href
  {\doibase 10.1007/978-94-017-7303-4\_11} {\emph {\bibinfo {booktitle}
  {Quantum Theory: Informational Foundations and Foils}}},\ \bibinfo {editor}
  {edited by\ \bibinfo {editor} {\bibfnamefont {G.}~\bibnamefont {Chiribella}}\
  and\ \bibinfo {editor} {\bibfnamefont {R.~W.}\ \bibnamefont {Spekkens}}}\
  (\bibinfo  {publisher} {Springer Netherlands},\ \bibinfo {address}
  {Dordrecht},\ \bibinfo {year} {2016})\ pp.\ \bibinfo {pages}
  {367--420}\BibitemShut {NoStop}%
\bibitem [{\citenamefont {Abramsky}\ and\ \citenamefont
  {Coecke}(2004)}]{Abramsky2004}%
  \BibitemOpen
  \bibfield  {author} {\bibinfo {author} {\bibfnamefont {S.}~\bibnamefont
  {Abramsky}}\ and\ \bibinfo {author} {\bibfnamefont {B.}~\bibnamefont
  {Coecke}},\ }in\ \href {\doibase 10.1109/LICS.2004.1319636} {\emph {\bibinfo
  {booktitle} {Proceedings of the 19th Annual IEEE Symposium on Logic in
  Computer Science}}}\ (\bibinfo {year} {2004})\ pp.\ \bibinfo {pages}
  {415--425}\BibitemShut {NoStop}%
\bibitem [{\citenamefont {Coecke}(2006)}]{Coecke-Kindergarten}%
  \BibitemOpen
  \bibfield  {author} {\bibinfo {author} {\bibfnamefont {B.}~\bibnamefont
  {Coecke}},\ }\href {\doibase 10.1063/1.2158713} {\bibfield  {journal}
  {\bibinfo  {journal} {AIP Conference Proceedings}\ }\textbf {\bibinfo
  {volume} {810}},\ \bibinfo {pages} {81} (\bibinfo {year} {2006})}\BibitemShut
  {NoStop}%
\bibitem [{\citenamefont {Coecke}(2010)}]{Coecke-Picturalism}%
  \BibitemOpen
  \bibfield  {author} {\bibinfo {author} {\bibfnamefont {B.}~\bibnamefont
  {Coecke}},\ }\href {\doibase 10.1080/00107510903257624} {\bibfield  {journal}
  {\bibinfo  {journal} {Contemp. Phys.}\ }\textbf {\bibinfo {volume} {51}},\
  \bibinfo {pages} {59} (\bibinfo {year} {2010})}\BibitemShut {NoStop}%
\bibitem [{\citenamefont {Selinger}(2011)}]{Selinger}%
  \BibitemOpen
  \bibfield  {author} {\bibinfo {author} {\bibfnamefont {P.}~\bibnamefont
  {Selinger}},\ }in\ \href {\doibase 10.1007/978-3-642-12821-9\_4} {\emph
  {\bibinfo {booktitle} {New Structures for Physics}}},\ \bibinfo {series}
  {Lecture Notes in Physics}, Vol.\ \bibinfo {volume} {813},\ \bibinfo {editor}
  {edited by\ \bibinfo {editor} {\bibfnamefont {B.}~\bibnamefont {Coecke}}}\
  (\bibinfo  {publisher} {Springer},\ \bibinfo {address} {Berlin, Heidelberg},\
  \bibinfo {year} {2011})\ pp.\ \bibinfo {pages} {289--356}\BibitemShut
  {NoStop}%
\bibitem [{\citenamefont {Coecke}\ \emph
  {et~al.}(2016{\natexlab{b}})\citenamefont {Coecke}, \citenamefont {Duncan},
  \citenamefont {Kissinger},\ and\ \citenamefont {Wang}}]{Coecke2016}%
  \BibitemOpen
  \bibfield  {author} {\bibinfo {author} {\bibfnamefont {B.}~\bibnamefont
  {Coecke}}, \bibinfo {author} {\bibfnamefont {R.}~\bibnamefont {Duncan}},
  \bibinfo {author} {\bibfnamefont {A.}~\bibnamefont {Kissinger}}, \ and\
  \bibinfo {author} {\bibfnamefont {Q.}~\bibnamefont {Wang}},\ }\enquote
  {\bibinfo {title} {Quantum theory: Informational foundations and foils},}\ \
  (\bibinfo  {publisher} {Springer Netherlands},\ \bibinfo {address}
  {Dordrecht},\ \bibinfo {year} {2016})\ Chap.\ \bibinfo {chapter} {Generalised
  Compositional Theories and Diagrammatic Reasoning}, pp.\ \bibinfo {pages}
  {309--366}\BibitemShut {NoStop}%
\bibitem [{\citenamefont {Coecke}\ and\ \citenamefont
  {Kissinger}(2017)}]{Coecke2017picturing}%
  \BibitemOpen
  \bibfield  {author} {\bibinfo {author} {\bibfnamefont {B.}~\bibnamefont
  {Coecke}}\ and\ \bibinfo {author} {\bibfnamefont {A.}~\bibnamefont
  {Kissinger}},\ }\href@noop {} {\emph {\bibinfo {title} {Picturing Quantum
  Processes: {A} First Course in Quantum Theory and Diagrammatic Reasoning}}}\
  (\bibinfo  {publisher} {Cambridge University Press},\ \bibinfo {address}
  {Cambridge},\ \bibinfo {year} {2017})\BibitemShut {NoStop}%
\bibitem [{\citenamefont {Chiribella}\ and\ \citenamefont
  {Yuan}(2016)}]{chiribella2016bridging}%
  \BibitemOpen
  \bibfield  {author} {\bibinfo {author} {\bibfnamefont {G.}~\bibnamefont
  {Chiribella}}\ and\ \bibinfo {author} {\bibfnamefont {X.}~\bibnamefont
  {Yuan}},\ }\href {\doibase 10.1016/j.ic.2016.02.006} {\bibfield  {journal}
  {\bibinfo  {journal} {Inform. Comput.}\ }\textbf {\bibinfo {volume} {250}},\
  \bibinfo {pages} {15} (\bibinfo {year} {2016})}\BibitemShut {NoStop}%
\bibitem [{\citenamefont {Daki\'c}\ and\ \citenamefont
  {Brukner}(2011)}]{Brukner}%
  \BibitemOpen
  \bibfield  {author} {\bibinfo {author} {\bibfnamefont {B.}~\bibnamefont
  {Daki\'c}}\ and\ \bibinfo {author} {\bibfnamefont {v.}~\bibnamefont
  {Brukner}},\ }\enquote {\bibinfo {title} {Quantum theory and beyond: Is
  entanglement special?}}\ in\ \href {\doibase 10.1017/CBO9780511976971.011}
  {\emph {\bibinfo {booktitle} {Deep Beauty: Understanding the Quantum World
  through Mathematical Innovation}}},\ \bibinfo {editor} {edited by\ \bibinfo
  {editor} {\bibfnamefont {H.}~\bibnamefont {Halvorson}}}\ (\bibinfo
  {publisher} {Cambridge University Press},\ \bibinfo {address} {Cambridge},\
  \bibinfo {year} {2011})\ pp.\ \bibinfo {pages} {365--392}\BibitemShut
  {NoStop}%
\bibitem [{\citenamefont {Masanes}\ and\ \citenamefont
  {M\"uller}(2011)}]{masanes}%
  \BibitemOpen
  \bibfield  {author} {\bibinfo {author} {\bibfnamefont {L.}~\bibnamefont
  {Masanes}}\ and\ \bibinfo {author} {\bibfnamefont {M.~P.}\ \bibnamefont
  {M\"uller}},\ }\href@noop {} {\bibfield  {journal} {\bibinfo  {journal} {New
  J. Phys.}\ }\textbf {\bibinfo {volume} {13}},\ \bibinfo {pages} {063001}
  (\bibinfo {year} {2011})}\BibitemShut {NoStop}%
\bibitem [{\citenamefont {Barnum}\ \emph {et~al.}(2014)\citenamefont {Barnum},
  \citenamefont {M\"uller},\ and\ \citenamefont
  {Ududec}}]{Barnum-interference}%
  \BibitemOpen
  \bibfield  {author} {\bibinfo {author} {\bibfnamefont {H.}~\bibnamefont
  {Barnum}}, \bibinfo {author} {\bibfnamefont {M.~P.}\ \bibnamefont
  {M\"uller}}, \ and\ \bibinfo {author} {\bibfnamefont {C.}~\bibnamefont
  {Ududec}},\ }\href {\doibase 10.1088/1367-2630/16/12/123029} {\bibfield
  {journal} {\bibinfo  {journal} {New J. Phys.}\ }\textbf {\bibinfo {volume}
  {16}},\ \bibinfo {pages} {123029} (\bibinfo {year} {2014})}\BibitemShut
  {NoStop}%
\bibitem [{\citenamefont {Birkhoff}(1946)}]{Birkhoff}%
  \BibitemOpen
  \bibfield  {author} {\bibinfo {author} {\bibfnamefont {G.}~\bibnamefont
  {Birkhoff}},\ }\href@noop {} {\bibfield  {journal} {\bibinfo  {journal}
  {Univ. Nac. Tucum{\'a}n Rev. Ser. A}\ }\textbf {\bibinfo {volume} {5}},\
  \bibinfo {pages} {147} (\bibinfo {year} {1946})}\BibitemShut {NoStop}%
\bibitem [{\citenamefont {Chiribella}\ and\ \citenamefont
  {Scandolo}(2015{\natexlab{b}})}]{Scandolo14}%
  \BibitemOpen
  \bibfield  {author} {\bibinfo {author} {\bibfnamefont {G.}~\bibnamefont
  {Chiribella}}\ and\ \bibinfo {author} {\bibfnamefont {C.~M.}\ \bibnamefont
  {Scandolo}},\ }\href {\doibase 10.1051/epjconf/20149503003} {\bibfield
  {journal} {\bibinfo  {journal} {EPJ Web of Conferences}\ }\textbf {\bibinfo
  {volume} {95}},\ \bibinfo {pages} {03003} (\bibinfo {year}
  {2015}{\natexlab{b}})}\BibitemShut {NoStop}%
\bibitem [{\citenamefont {Piron}(1976)}]{PironBook}%
  \BibitemOpen
  \bibfield  {author} {\bibinfo {author} {\bibfnamefont {C.}~\bibnamefont
  {Piron}},\ }\href@noop {} {\emph {\bibinfo {title} {Foundations of quantum
  physics}}},\ Mathematical Physics Monograph Series\ (\bibinfo  {publisher}
  {Benjamin-Cummings Publishing Company},\ \bibinfo {year} {1976})\BibitemShut
  {NoStop}%
\bibitem [{\citenamefont {Chiribella}\ and\ \citenamefont
  {Scandolo}(2015{\natexlab{c}})}]{QPL15}%
  \BibitemOpen
  \bibfield  {author} {\bibinfo {author} {\bibfnamefont {G.}~\bibnamefont
  {Chiribella}}\ and\ \bibinfo {author} {\bibfnamefont {C.~M.}\ \bibnamefont
  {Scandolo}},\ }in\ \href {\doibase 10.4204/EPTCS.195.8} {\emph {\bibinfo
  {booktitle} {{\rm Proceedings of the 12th International Workshop on} Quantum
  Physics and Logic, {\rm Oxford, U.K., July 15-17, 2015}}}},\ \bibinfo
  {series} {Electronic Proceedings in Theoretical Computer Science}, Vol.\
  \bibinfo {volume} {195},\ \bibinfo {editor} {edited by\ \bibinfo {editor}
  {\bibfnamefont {C.}~\bibnamefont {Heunen}}, \bibinfo {editor} {\bibfnamefont
  {P.}~\bibnamefont {Selinger}}, \ and\ \bibinfo {editor} {\bibfnamefont
  {J.}~\bibnamefont {Vicary}}}\ (\bibinfo {year} {2015})\ pp.\ \bibinfo {pages}
  {96--115}\BibitemShut {NoStop}%
\bibitem [{\citenamefont {Barnum}\ \emph {et~al.}(2015)\citenamefont {Barnum},
  \citenamefont {Barrett}, \citenamefont {Krumm},\ and\ \citenamefont
  {M\"uller}}]{Krumm-Muller}%
  \BibitemOpen
  \bibfield  {author} {\bibinfo {author} {\bibfnamefont {H.}~\bibnamefont
  {Barnum}}, \bibinfo {author} {\bibfnamefont {J.}~\bibnamefont {Barrett}},
  \bibinfo {author} {\bibfnamefont {M.}~\bibnamefont {Krumm}}, \ and\ \bibinfo
  {author} {\bibfnamefont {M.~P.}\ \bibnamefont {M\"uller}},\ }in\ \href
  {\doibase 10.4204/EPTCS.195.4} {\emph {\bibinfo {booktitle} {{\rm Proceedings
  of the 12th International Workshop on} Quantum Physics and Logic, {\rm
  Oxford, U.K., July 15-17, 2015}}}},\ \bibinfo {series} {Electronic
  Proceedings in Theoretical Computer Science}, Vol.\ \bibinfo {volume} {195},\
  \bibinfo {editor} {edited by\ \bibinfo {editor} {\bibfnamefont
  {C.}~\bibnamefont {Heunen}}, \bibinfo {editor} {\bibfnamefont
  {P.}~\bibnamefont {Selinger}}, \ and\ \bibinfo {editor} {\bibfnamefont
  {J.}~\bibnamefont {Vicary}}}\ (\bibinfo {year} {2015})\ pp.\ \bibinfo {pages}
  {43--58}\BibitemShut {NoStop}%
\bibitem [{\citenamefont {Krumm}(2015)}]{Krumm-thesis}%
  \BibitemOpen
  \bibfield  {author} {\bibinfo {author} {\bibfnamefont {M.}~\bibnamefont
  {Krumm}},\ }\href {http://arxiv.org/abs/1508.03299} {\bibfield  {journal}
  {\bibinfo  {journal} {arXiv:1508.03299 [quant-ph]}\ } (\bibinfo {year}
  {2015})},\ \bibinfo {note} {{M}aster's thesis}\BibitemShut {NoStop}%
\bibitem [{\citenamefont {Hardy}\ \emph {et~al.}(1929)\citenamefont {Hardy},
  \citenamefont {Littlewood},\ and\ \citenamefont
  {P{\'o}lya}}]{Hardy-Littlewood-Polya1929}%
  \BibitemOpen
  \bibfield  {author} {\bibinfo {author} {\bibfnamefont {G.~H.}\ \bibnamefont
  {Hardy}}, \bibinfo {author} {\bibfnamefont {J.~E.}\ \bibnamefont
  {Littlewood}}, \ and\ \bibinfo {author} {\bibfnamefont {G.}~\bibnamefont
  {P{\'o}lya}},\ }\href@noop {} {\bibfield  {journal} {\bibinfo  {journal}
  {Messenger Math.}\ }\textbf {\bibinfo {volume} {58}},\ \bibinfo {pages} {310}
  (\bibinfo {year} {1929})}\BibitemShut {NoStop}%
\bibitem [{\citenamefont {Scandolo}(2014)}]{Scandolo}%
  \BibitemOpen
  \bibfield  {author} {\bibinfo {author} {\bibfnamefont {C.~M.}\ \bibnamefont
  {Scandolo}},\ }\emph {\bibinfo {title} {Entanglement and thermodynamics in
  general probabilistic theories}},\ \href
  {http://tesi.cab.unipd.it/46015/1/Scandolo_carlo_maria.pdf} {Master's
  thesis},\ \bibinfo  {school} {Universit\`a degli Studi di Padova}, \bibinfo
  {address} {Italy} (\bibinfo {year} {2014})\BibitemShut {NoStop}%
\bibitem [{\citenamefont {Chiribella}(2015)}]{Chiribellatalk}%
  \BibitemOpen
  \bibfield  {author} {\bibinfo {author} {\bibfnamefont {G.}~\bibnamefont
  {Chiribella}},\ }\href {http://pirsa.org/15050072/} {\enquote {\bibinfo
  {title} {Towards an information-theoretic foundation of (quantum)
  thermodynamics},}\ } (\bibinfo {year} {2015}),\ \bibinfo {note} {talk given
  at the conference ``Information theoretic foundations of physics'', 11--15
  May 2015, Perimeter Institute, Waterloo, Canada}\BibitemShut {NoStop}%
\bibitem [{\citenamefont {Krumm}\ \emph {et~al.}(2017)\citenamefont {Krumm},
  \citenamefont {Barnum}, \citenamefont {Barrett},\ and\ \citenamefont
  {M\"uller}}]{Colleagues}%
  \BibitemOpen
  \bibfield  {author} {\bibinfo {author} {\bibfnamefont {M.}~\bibnamefont
  {Krumm}}, \bibinfo {author} {\bibfnamefont {H.}~\bibnamefont {Barnum}},
  \bibinfo {author} {\bibfnamefont {J.}~\bibnamefont {Barrett}}, \ and\
  \bibinfo {author} {\bibfnamefont {M.~P.}\ \bibnamefont {M\"uller}},\ }\href
  {http://stacks.iop.org/1367-2630/19/i=4/a=043025} {\bibfield  {journal}
  {\bibinfo  {journal} {New J. Phys.}\ }\textbf {\bibinfo {volume} {19}},\
  \bibinfo {pages} {043025} (\bibinfo {year} {2017})}\BibitemShut {NoStop}%
\bibitem [{\citenamefont {M\"uller}\ and\ \citenamefont
  {Masanes}(2013)}]{Muller3D}%
  \BibitemOpen
  \bibfield  {author} {\bibinfo {author} {\bibfnamefont {M.~P.}\ \bibnamefont
  {M\"uller}}\ and\ \bibinfo {author} {\bibfnamefont {L.}~\bibnamefont
  {Masanes}},\ }\href {\doibase 10.1088/1367-2630/15/5/053040} {\bibfield
  {journal} {\bibinfo  {journal} {New J. Phys.}\ }\textbf {\bibinfo {volume}
  {15}},\ \bibinfo {pages} {053040} (\bibinfo {year} {2013})}\BibitemShut
  {NoStop}%
\bibitem [{\citenamefont {Plenio}\ and\ \citenamefont
  {Virmani}(2007)}]{Entanglement-entropy1}%
  \BibitemOpen
  \bibfield  {author} {\bibinfo {author} {\bibfnamefont {M.~B.}\ \bibnamefont
  {Plenio}}\ and\ \bibinfo {author} {\bibfnamefont {S.}~\bibnamefont
  {Virmani}},\ }\href@noop {} {\bibfield  {journal} {\bibinfo  {journal}
  {Quant. Inf. Comp.}\ }\textbf {\bibinfo {volume} {7}},\ \bibinfo {pages} {1}
  (\bibinfo {year} {2007})}\BibitemShut {NoStop}%
\bibitem [{\citenamefont {Horodecki}\ \emph {et~al.}(2009)\citenamefont
  {Horodecki}, \citenamefont {Horodecki}, \citenamefont {Horodecki},\ and\
  \citenamefont {Horodecki}}]{Entanglement-entropy2}%
  \BibitemOpen
  \bibfield  {author} {\bibinfo {author} {\bibfnamefont {R.}~\bibnamefont
  {Horodecki}}, \bibinfo {author} {\bibfnamefont {P.}~\bibnamefont
  {Horodecki}}, \bibinfo {author} {\bibfnamefont {M.}~\bibnamefont
  {Horodecki}}, \ and\ \bibinfo {author} {\bibfnamefont {K.}~\bibnamefont
  {Horodecki}},\ }\href {\doibase 10.1103/RevModPhys.81.865} {\bibfield
  {journal} {\bibinfo  {journal} {Rev. Mod. Phys.}\ }\textbf {\bibinfo {volume}
  {81}},\ \bibinfo {pages} {865} (\bibinfo {year} {2009})}\BibitemShut
  {NoStop}%
\bibitem [{\citenamefont {Janzing}(2009)}]{Janzing2009}%
  \BibitemOpen
  \bibfield  {author} {\bibinfo {author} {\bibfnamefont {D.}~\bibnamefont
  {Janzing}},\ }\enquote {\bibinfo {title} {Entropy of entanglement},}\ in\
  \href {\doibase 10.1007/978-3-540-70626-7\_66} {\emph {\bibinfo {booktitle}
  {Compendium of Quantum Physics}}},\ \bibinfo {editor} {edited by\ \bibinfo
  {editor} {\bibfnamefont {D.}~\bibnamefont {Greenberger}}, \bibinfo {editor}
  {\bibfnamefont {K.}~\bibnamefont {Hentschel}}, \ and\ \bibinfo {editor}
  {\bibfnamefont {F.}~\bibnamefont {Weinert}}}\ (\bibinfo  {publisher}
  {Springer},\ \bibinfo {address} {Berlin, Heidelberg},\ \bibinfo {year}
  {2009})\ pp.\ \bibinfo {pages} {205--209}\BibitemShut {NoStop}%
\bibitem [{\citenamefont {Ryu}\ and\ \citenamefont {Takayanagi}(2006)}]{Ryu1}%
  \BibitemOpen
  \bibfield  {author} {\bibinfo {author} {\bibfnamefont {S.}~\bibnamefont
  {Ryu}}\ and\ \bibinfo {author} {\bibfnamefont {T.}~\bibnamefont
  {Takayanagi}},\ }\href {\doibase 10.1103/PhysRevLett.96.181602} {\bibfield
  {journal} {\bibinfo  {journal} {Phys. Rev. Lett.}\ }\textbf {\bibinfo
  {volume} {96}},\ \bibinfo {pages} {181602} (\bibinfo {year}
  {2006})}\BibitemShut {NoStop}%
\bibitem [{\citenamefont {Nishioka}\ \emph {et~al.}(2009)\citenamefont
  {Nishioka}, \citenamefont {Ryu},\ and\ \citenamefont {Takayanagi}}]{Ryu2}%
  \BibitemOpen
  \bibfield  {author} {\bibinfo {author} {\bibfnamefont {T.}~\bibnamefont
  {Nishioka}}, \bibinfo {author} {\bibfnamefont {S.}~\bibnamefont {Ryu}}, \
  and\ \bibinfo {author} {\bibfnamefont {T.}~\bibnamefont {Takayanagi}},\
  }\href {\doibase 10.1088/1751-8113/42/50/504008} {\bibfield  {journal}
  {\bibinfo  {journal} {J. Phys. A}\ }\textbf {\bibinfo {volume} {42}},\
  \bibinfo {pages} {504008} (\bibinfo {year} {2009})}\BibitemShut {NoStop}%
\bibitem [{\citenamefont {Eisert}\ \emph {et~al.}(2010)\citenamefont {Eisert},
  \citenamefont {Cramer},\ and\ \citenamefont {Plenio}}]{Area-law}%
  \BibitemOpen
  \bibfield  {author} {\bibinfo {author} {\bibfnamefont {J.}~\bibnamefont
  {Eisert}}, \bibinfo {author} {\bibfnamefont {M.}~\bibnamefont {Cramer}}, \
  and\ \bibinfo {author} {\bibfnamefont {M.~B.}\ \bibnamefont {Plenio}},\
  }\href {\doibase 10.1103/RevModPhys.82.277} {\bibfield  {journal} {\bibinfo
  {journal} {Rev. Mod. Phys.}\ }\textbf {\bibinfo {volume} {82}},\ \bibinfo
  {pages} {277} (\bibinfo {year} {2010})}\BibitemShut {NoStop}%
\bibitem [{\citenamefont {Petz}(2003)}]{Petz}%
  \BibitemOpen
  \bibfield  {author} {\bibinfo {author} {\bibfnamefont {D.}~\bibnamefont
  {Petz}},\ }\href {\doibase 10.1142/S0129055X03001576} {\bibfield  {journal}
  {\bibinfo  {journal} {Rev. Math. Phys.}\ }\textbf {\bibinfo {volume} {15}},\
  \bibinfo {pages} {79} (\bibinfo {year} {2003})}\BibitemShut {NoStop}%
\end{thebibliography}%

\appendix 

\section{Properties of the microcanonical state}\label{app:propertieschi}  

 \begin{prop}
 For every theory satisfying requirement~\ref{req:uniquep} and for every finite system   $\rA$ in the theory, the microcanonical state  $\chi_\rA$  is invariant under all reversible transformations of system $\rA$.  
 \end{prop}
 \begin{proof}For every reversible transformation $\map U$, one has  
 \begin{align}
\nonumber  \map U  \chi_\rA  &  =    \int p_\rA\left(   \d \psi\right)    \map U  \psi   \\
  \nonumber  &   =  \int p_\rA\left(   \d   \map U^{-1}  \psi'\right)    \psi'   \\
 \nonumber   &  =   \int p_\rA\left(   \d \psi'\right)   \psi'\\
   &  =  \chi_\rA \, ,
 \end{align}
the second equality following from the definition  $\psi'  : =   \map U \psi$, and the third equality following from  the invariance of the probability distribution $p_\rA$.
\end{proof} 

\begin{prop}
 For every theory satisfying requirement~\ref{req:uniquep} and for every finite system   $\rA$ in the theory,  there exists a set of reversible transformations $\left\{  \map U_i\right\}_{i=1}^r$ and a probability distribution $\left\{p_i\right\}_{i=1}^r$ such that 
\begin{align}\label{equilibration}
\sum_{i=1}^r   p_i    \map U_i  \alpha  =  \chi_\rA,
\end{align}  
for every deterministic pure  state $ \alpha $ of the system.   
 \end{prop}
 
 \begin{proof}The group of reversible transformations on system $\rA$ has a finite-dimensional representation on the state space of system $\rA$.  This representation defines a group of finite-dimensional matrices, call it $\widetilde {\grp G}_\rA$.  Note that the group   $\widetilde {\grp G}_\rA$ is compact, because it is closed and finite-dimensional. 
 Hence, one can construct the invariant  measure $\d \map U$ and define the transformation  
 \begin{align}\label{ta}
 \map T_\rA  :  = \int_{\widetilde {\grp  G}_\rA } \d \map U \: \map U.
 \end{align} 
 By construction, the transformation $\map T_\rA$  maps every deterministic pure state $\alpha$ into the microcanonical state:  indeed, one has  
 \begin{align}
 \nonumber \map T_\rA   \alpha   &  =     \int_{\widetilde {\grp  G}_\rA } \d \map U ~ \map U \alpha  \\
   \nonumber &  =       \int  p_\rA  (\d \psi) \,  \psi  \\
    & =  \chi_\rA \, ,
 \end{align}   
 the second equality following from the fact that  $p_\rA (\d \psi)$ is the probability distribution induced by the invariant measure on $\widetilde {\grp G}_\rA$.   Finally, since the matrices in  $\widetilde {\grp G}_\rA$ are finite-dimensional, the integral in Eq.~\eqref{ta} can be replaced by a finite sum.
\end{proof}

\section{Proof of lemma~\ref{prop:rarenoisy}\label{app:rarenoisy}}
\begin{proof}  
	Let $\map R \in \mathsf{DetTransf}\left( \rA\right) $ be a rational RaRe channel, written as  
	\begin{align}
	\map R   =   \sum_{i}  \frac {n_i}n \map U_i  ,  
	\end{align}  
	with $n_i\ge 0$ and $\sum_i n_i  =  n$.  Let $\rB$ be an $ n $-dimensional system, and pick the pure maximal set $\left\{  \beta_x\right\}_{x=1}^n$. Let $\map C$ be the channel from $\rA\otimes \rB  $ to $\rA$ defined by  
	\begin{align}
	\map C   =  \sum_{x=1}^n    \map V_x  \otimes \map   \beta_x^\dag ,   
	\end{align}
	where  $\left\{\map V_x\right\}_{x=1}^n$ are reversible transformations on $\rA$, chosen so that  $n_1$ of the channels are equal to $\map U_1$, $n_2$ are equal to $\map U_2$, and so on.  Since the theory satisfies Purification, the channel $\map C$ has a reversible extension \cite{Chiribella-purification,Chiribella14}, meaning that one has 
	\begin{align}
	\begin{aligned} \Qcircuit @C=1em @R=.7em @!R {  
		& \qw \poloFantasmaCn{\rA} & \multigate{1}{\mathcal C} & \qw \poloFantasmaCn{\rA} & \qw \\ 
		&  \qw \poloFantasmaCn{\rB}  &  \ghost{\mathcal C} &  & } \end{aligned}~ 
	= ~  \begin{aligned} \Qcircuit @C=1em @R=.7em @!R {  
		& \qw \poloFantasmaCn{\rA} & \multigate{2}{\mathcal U} & \qw \poloFantasmaCn{\rA} & \qw \\ 
		&  \qw \poloFantasmaCn{\rB}  &  \ghost{\mathcal U} &  &  \\
		\prepareC{\gamma}&  \qw \poloFantasmaCn{\rC} & \ghost{\mathcal U} &  \qw \poloFantasmaCn{\rC'}   & \measureD{u}   } \end{aligned}  ~ , 
	\end{align} 
	where $\rC$ and $\rC'$ are suitable systems, $\gamma$ is a suitable pure state, and $\map U$ is a reversible transformation.   Now, by construction we have  
	\begin{align}
	\nonumber  \begin{aligned} \Qcircuit @C=1em @R=.7em @!R {  
		& \qw \poloFantasmaCn{\rA} & \multigate{2}{\mathcal U} & \qw \poloFantasmaCn{\rA} & \qw \\ 
		\prepareC{\beta_x}&  \qw \poloFantasmaCn{\rB}  &  \ghost{\mathcal U} &  &  \\
		\prepareC{\gamma}&  \qw \poloFantasmaCn{\rC} & \ghost{\mathcal U} &  \qw \poloFantasmaCn{\rC'}   & \measureD{u}   } \end{aligned}   ~ & =  ~
	\begin{aligned} \Qcircuit @C=1em @R=.7em @!R {  
		& \qw \poloFantasmaCn{\rA} & \multigate{1}{\mathcal C} & \qw \poloFantasmaCn{\rA} & \qw \\ 
		\prepareC{\beta_x}&  \qw \poloFantasmaCn{\rB}  &  \ghost{\mathcal C} &  & } \end{aligned}  \\ 
	\nonumber   \\
	&= ~\begin{aligned} \Qcircuit @C=1em @R=.7em @!R {  
		& \qw \poloFantasmaCn{\rA} & \gate{\map V_x} & \qw \poloFantasmaCn{\rA} & \qw }
	\end{aligned} ~,
	\end{align}
	for every $x\in\left\{1,\ldots,  n\right\}$.     The above condition implies the relation \cite{Chiribella14}  
	\begin{equation}\label{Uprog}
	\begin{aligned} \Qcircuit @C=1em @R=.7em @!R {  
		& \qw \poloFantasmaCn{\rA} & \multigate{2}{\mathcal U} & \qw \poloFantasmaCn{\rA} & \qw \\ 
		\prepareC{\beta_x}&  \qw \poloFantasmaCn{\rB}  &  \ghost{\mathcal U} &  &  \\
		\prepareC{\gamma}&  \qw \poloFantasmaCn{\rC} & \ghost{\mathcal U} &  \qw \poloFantasmaCn{\rC'}   & \qw  } \end{aligned}    ~=  ~
	\begin{aligned} \Qcircuit @C=1em @R=.7em @!R {  
		& \qw \poloFantasmaCn{\rA} & \gate{\map V_x} & \qw \poloFantasmaCn{\rA} & \qw   \\
		 \prepareC{\gamma_x}   & \qw \poloFantasmaCn{\rC'} & \qw  &\qw&\qw}
	\end{aligned}  ~, 
	\end{equation}for some pure state $ \gamma_x $ of system $ \rC' $.
	Composing both sides with $\map V_x^{-1}$  on the left, and with $\map U^{-1}$ on the right we obtain  
	\begin{align}\label{Uprogdag}
	\begin{aligned} \Qcircuit @C=1em @R=.7em @!R {  
		& \qw \poloFantasmaCn{\rA} & \gate{\mathcal {V}_x^{-1}} & \qw \poloFantasmaCn{\rA} & \qw \\ 
		\prepareC{\beta_x}&  \qw \poloFantasmaCn{\rB}  & \qw &\qw &\qw \\
		\prepareC{\gamma}&  \qw \poloFantasmaCn{\rC} &   \qw   &\qw &\qw } \end{aligned}   ~ = ~
	\begin{aligned} \Qcircuit @C=1em @R=.7em @!R {  
		& \qw \poloFantasmaCn{\rA} & \multigate{2}{\map U^{-1}} & \qw \poloFantasmaCn{\rA} & \qw   \\
		&  & \pureghost{\map U^{-1}} & \qw \poloFantasmaCn{\rB} & \qw   \\
		\prepareC{\gamma_x} & \qw \poloFantasmaCn{\rC'} & \ghost{\map U^{-1}} & \qw \poloFantasmaCn{\rC} & \qw   
	}
	\end{aligned}  ~. 
	\end{align} 
	Combining Eqs.~\eqref{Uprog} and \eqref{Uprogdag} we obtain the relation  
	\begin{align}\label{doubleprog}
	\begin{aligned} \Qcircuit @C=1em @R=.7em @!R {  
		& \qw \poloFantasmaCn{\rA} & \multigate{2}{\mathcal U} & \qw \poloFantasmaCn{\rA} & \qw  &  \qw &\qw   &\qw \\ 
		\prepareC{\beta_x}&  \qw \poloFantasmaCn{\rB}  &  \ghost{\mathcal U} &  &   &  \pureghost{\map U^{-1}}  &    \qw \poloFantasmaCn{\rB}  & \qw  \\
		\prepareC{\gamma}&  \qw \poloFantasmaCn{\rC} & \ghost{\mathcal U} &  \qw \poloFantasmaCn{\rC'}   & \qw   &    \ghost{\map U^{-1}}  &    \qw \poloFantasmaCn{\rC}  & \qw  \\    
		&  \qw   \poloFantasmaCn{\rA} & \qw& \qw & \qw  &    \multigate{-2}{\map U^{-1}}  &    \qw \poloFantasmaCn{\rA}  & \qw    } \end{aligned}    ~   =   ~
	\begin{aligned} \Qcircuit @C=1em @R=.7em @!R {  
		& \qw \poloFantasmaCn{\rA} & \gate{\map V_x}  &\qw   \poloFantasmaCn{\rA}&\qw \\ 
		\prepareC{\beta_x}&  \qw \poloFantasmaCn{\rB}  &  \qw &\qw &\qw \\
		\prepareC{\gamma}&  \qw \poloFantasmaCn{\rC} & \qw &\qw &\qw \\ 
		& \qw   \poloFantasmaCn{\rA} & \gate{\map V_x^{-1}}  &    \qw \poloFantasmaCn{\rA}  & \qw    } \end{aligned} ~.  
	\end{align}
	At this point, we define the pure transformation  
	\begin{align}
	\begin{aligned} \Qcircuit @C=1em @R=.7em @!R {  
		& \qw \poloFantasmaCn{\rA} & \multigate{2}{\mathcal P} &\qw   \poloFantasmaCn{\rA}&\qw \\ 
		&  \qw \poloFantasmaCn{\rB}  &  \ghost{\mathcal P}  &    \qw \poloFantasmaCn{\rB}  & \qw  \\
		&  \qw   \poloFantasmaCn{\rA} &\ghost{\map P}  &    \qw \poloFantasmaCn{\rA}  & \qw    } \end{aligned}          
	~:  =~
	\begin{aligned} \Qcircuit @C=1em @R=.7em @!R {  
		& \qw \poloFantasmaCn{\rA} & \multigate{2}{\mathcal U} & \qw \poloFantasmaCn{\rA} & \qw  &  \qw &\qw   &\qw \\ 
		&  \qw \poloFantasmaCn{\rB}  &  \ghost{\mathcal U} &  &   &  \pureghost{\map U^{-1}}  &    \qw \poloFantasmaCn{\rB}  & \qw  \\
		\prepareC{\gamma}&  \qw \poloFantasmaCn{\rC} & \ghost{\mathcal U} &  \qw \poloFantasmaCn{\rC'}   & \qw   &    \ghost{\map U^{-1}}  &    \qw \poloFantasmaCn{\rC}  & \measureD{\gamma^\dag}  \\    
		&  \qw   \poloFantasmaCn{\rA} & \qw& \qw & \qw  &    \multigate{-2}{\map U^{-1}}  &    \qw \poloFantasmaCn{\rA}  & \qw    } \end{aligned}      ~.
	\end{align}
	From Eq.~\eqref{doubleprog} we obtain that $\map P$ satisfies the relation  
	\begin{align}
	\begin{aligned} \Qcircuit @C=1em @R=.7em @!R {  
		& \qw \poloFantasmaCn{\rA} & \multigate{2}{\mathcal P} &\qw   \poloFantasmaCn{\rA}&\qw \\ 
		\prepareC{\beta_x}&  \qw \poloFantasmaCn{\rB}  &  \ghost{\mathcal P}  &    \qw \poloFantasmaCn{\rB}  & \qw  \\
		&  \qw   \poloFantasmaCn{\rA} &\ghost{\map P}  &    \qw \poloFantasmaCn{\rA}  & \qw    } \end{aligned}         ~   =   ~
	\begin{aligned} \Qcircuit @C=1em @R=.7em @!R {  
		& \qw \poloFantasmaCn{\rA} & \gate{\map V_x}  &\qw   \poloFantasmaCn{\rA}&\qw \\ 
		\prepareC{\beta_x}&  \qw \poloFantasmaCn{\rB}  &  \qw &\qw &\qw\\
		& \qw   \poloFantasmaCn{\rA} & \gate{\map V_x^{-1}}  &    \qw \poloFantasmaCn{\rA}  & \qw    } \end{aligned}  ~, 
	\end{align}  
	for all values of $x$.   Using this relation and the expression of $ \chi_\rB $ in terms of the $ \beta_x $'s, we can reconstruct $ \mathcal{R} $ from $ \mathcal{P} $:
	\begin{align}
	\nonumber 
	&\begin{aligned} \Qcircuit @C=1em @R=.7em @!R {  
		& \qw \poloFantasmaCn{\rA} & \multigate{2}{\mathcal P} &\qw   \poloFantasmaCn{\rA}&\qw \\ 
		\prepareC{\chi}&  \qw \poloFantasmaCn{\rB}  &  \ghost{\mathcal P}  &    \qw \poloFantasmaCn{\rB}  & \measureD{u}  \\
		\prepareC{\chi}&  \qw   \poloFantasmaCn{\rA} &\ghost{\map P}  &    \qw \poloFantasmaCn{\rA}  & \measureD{u}    } \end{aligned}             ~ =  \frac 1n   \sum_{x=1}^n ~
	\begin{aligned}    \Qcircuit @C=1em @R=.7em @!R {  
		& \qw \poloFantasmaCn{\rA} & \multigate{2}{\mathcal P} &\qw   \poloFantasmaCn{\rA}&\qw \\ 
		\prepareC{\beta_x}&  \qw \poloFantasmaCn{\rB}  &  \ghost{\mathcal P}  &    \qw \poloFantasmaCn{\rB}  & \measureD{u}  \\
		\prepareC{\chi}&  \qw   \poloFantasmaCn{\rA} &\ghost{\map P}  &    \qw \poloFantasmaCn{\rA}  & \measureD{u}    }  \end{aligned}        \\ 
	\nonumber  &\\
	\nonumber &  \qquad =    \frac 1n  \sum_{x=1}^n~  
	\begin{aligned} \Qcircuit @C=1em @R=.7em @!R {  
		&\qw \poloFantasmaCn{\rA} & \gate{\map V_x}  &\qw   \poloFantasmaCn{\rA}&\qw \\ 
		\prepareC{\beta_x}&  \qw \poloFantasmaCn{\rB}  &\qw &\qw  &\measureD{u} & \\
		\prepareC{\chi} & \qw   \poloFantasmaCn{\rA} & \gate{\map V_x^{-1}}  &    \qw \poloFantasmaCn{\rA}  & \measureD{u}   } \end{aligned}  \\
	\nonumber  &\\
	\nonumber &   \qquad =    \frac 1n  \sum_{x=1}^n  ~
	\begin{aligned} \Qcircuit @C=1em @R=.7em @!R {  
		&\qw \poloFantasmaCn{\rA} & \gate{\map V_x}  &\qw   \poloFantasmaCn{\rA}&\qw  } \end{aligned}  \\
	\nonumber & \\
	&   \qquad=
	\begin{aligned} \Qcircuit @C=1em @R=.7em @!R {  
		&\qw \poloFantasmaCn{\rA} & \gate{\map R}  &\qw   \poloFantasmaCn{\rA}&\qw  } \end{aligned}    ~, 
	\label{simplenoisy}\end{align}where we have used the fact that $ \sum_{x=1}^{n}\map V_x=\sum_i n_i \map U_i$.
	Finally, let us show that $\map P$ is a channel. To this end, it is enough to show that $ u\mathcal{P}=u $ \cite{Chiribella-purification}. This property is satisfied if and only if $ \left( u\middle|\mathcal{P}\middle|\chi\right) =1 $, because every state lies in some convex decomposition of $ \chi $ \cite{Chiribella-purification}.  By the condition of informational equilibrium and Eq.~\eqref{simplenoisy}, we have
	\begin{align}
	\nonumber\begin{aligned} \Qcircuit @C=1em @R=.7em @!R {  
		\multiprepareC{2}{\chi}& \qw \poloFantasmaCn{\rA} & \multigate{2}{\mathcal P} &\qw   \poloFantasmaCn{\rA}&\measureD{u} \\ 
		\pureghost{\chi}&  \qw \poloFantasmaCn{\rB}  &  \ghost{\mathcal P}  &    \qw \poloFantasmaCn{\rB}  & \measureD{u}  \\
		\pureghost{\chi}&  \qw   \poloFantasmaCn{\rA} &\ghost{\map P}  &    \qw \poloFantasmaCn{\rA}  & \measureD{u}    } \end{aligned}~&=~\begin{aligned} \Qcircuit @C=1em @R=.7em @!R {  
		\prepareC{\chi}& \qw \poloFantasmaCn{\rA} & \multigate{2}{\mathcal P} &\qw   \poloFantasmaCn{\rA}&\measureD{u} \\ 
		\prepareC{\chi}&  \qw \poloFantasmaCn{\rB}  &  \ghost{\mathcal P}  &    \qw \poloFantasmaCn{\rB}  & \measureD{u}  \\
		\prepareC{\chi}&  \qw   \poloFantasmaCn{\rA} &\ghost{\map P}  &    \qw \poloFantasmaCn{\rA}  & \measureD{u}    }\end{aligned} \\
		&=~\nonumber\begin{aligned} \Qcircuit @C=1em @R=.7em @!R {  \prepareC{\chi}
		&\qw \poloFantasmaCn{\rA} & \gate{\map R}  &\qw   \poloFantasmaCn{\rA}&\measureD{u}  }\end{aligned}\\
	    &=1~,
	\end{align}
	so $\map P$ is a channel. Since every pure channel on a fixed system (here $ \mathrm{A}\otimes\mathrm{B}\otimes\mathrm{A} $) is reversible \cite{Chiribella-purification}, $\map P$ is reversible.
	Hence, Eq.~\eqref{simplenoisy} shows that $\map R$ is a basic noisy operation, with environment $\rE  =  \rB\otimes \rA$.   \end{proof} 

\section{Proof of theorem~\ref{prop:unital channels}\label{app:unital channels}}
\begin{proof}
	Let $\rho=\sum_{j=1}^{d}p_{j}\alpha_{j}$ and $\sigma=\sum_{j=1}^{d}q_{j}\alpha'_{j}$
	be diagonalisations of $\rho$ and $\sigma$, respectively. We first show that  $\rho  \succeq_{\mathsf{Unital}}  \sigma$ implies $\mathbf{p} \succeq \mathbf{q}$.
	Suppose that one has $\sigma=\mathcal{D}\rho$, where $\mathcal{D}$ is a unital
	channel. Then 
	\begin{equation}
	\sum_{j=1}^{d}q_{j}\alpha'_{j}=\sum_{j=1}^{d}p_{j}\mathcal{D}\alpha_{j}.
	\end{equation}
	Applying $\alpha_{i}'^{\dagger}$ to both sides, we obtain 
	\begin{align}
	\nonumber
	q_{i}  &=\sum_{j=1}^{d}p_{j}\left(\alpha_{i}'^{\dagger}\middle|\mathcal{D}\middle|\alpha_{j}\right)  \\
	\label{matrixD} &  =  \sum_{j=1}^{d}D_{ij}p_{j}  \, , \qquad   D_{ij}:=\left(\alpha_{i}'^{\dagger}\middle|\mathcal{D}\middle|\alpha_{j}\right) \, .
	\end{align}
	Now, the  $D_{ij}$'s are the entries of a doubly stochastic
	matrix $D$ (lemma~\ref{lem:channelmatrix}). Hence, Eq.~\eqref{matrixD} implies that  $\mathbf{p}$ majorises $\mathbf{q}$. 
	
	Conversely, suppose that $\st p  \succeq \st q$ and let $D$ be a doubly stochastic matrix such that $\st q  =  D \st p$. Define the measure-and-prepare channel 
	\begin{align}
	\map D    =  \sum_{j=1}^{d}   \rho_j      \alpha_j^\dag       \qquad \rho_j  :  =    \sum_{i=1}^d  D_{ij}   \alpha'_i   \, .
	\end{align}
	By construction, one has 
	\begin{align}
	\nonumber \map D \rho &  =    \sum_{j=1}^d      \rho_j     \left(\alpha^\dag_j\middle|\rho\right)\\
	\nonumber  &  =      \sum_{i=1}^{d}  \alpha_i' \sum_{j=1}^{d}    D_{ij}      p_j  \\
	\nonumber  &   =  \sum_{i=1}^{d}    q_i  \alpha_i' \\
	&  = \sigma \, .
	\end{align}
	Now, the channel $\map D$ is unital by lemma~\ref{lem:matrixchannel}. Hence, $\rho$  can be converted into $\sigma$ by a unital channel. 
\end{proof}

\section{Operational features of Doubled Quantum Theory}\label{app:operationaldoubled}  

Here we summarise the key operational features of Doubled Quantum Theory.

 \subsection{Doubled Quantum Theory violates Local Tomography}\label{app:double local tomography}
 An equivalent formulation of Local Tomography is that the dimension of the vector space spanned by the states of a composite system is equal to the product of the dimensions of the vector spaces spanned by the states of the components \cite{Hardy-informational-1,Chiribella-purification}.   The equality fails to hold in Doubled Quantum Theory, where the dimension of the global vector space is strictly larger than the product of the dimensions of the individual vector spaces.   To see why this is the case, note that the block diagonal states of the form~\eqref{eq:state doubled} span  a vector space of dimension $D  :=  2 d^2$, where $d$ is the dimension of the Hilbert spaces $\spc H_0$ and $\spc H_1$.  Given two systems $\rA$ and $\rB$, the product of the individual dimensions is 
 \begin{align}
\nonumber D_\rA  D_\rB & = (2  d_\rA^2 ) \cdot   (2  d_\rB^2 )\\
   & =  \left( 2  d_\rA d_\rB\right) ^2 .
\end{align} 
On the other hand,  each of the Hilbert spaces $\spc H^{\mathrm{AB}}_0 $ and  $\spc H^{\mathrm{AB}}_1 $  in Eq.~\eqref{hhahhb}  has dimension  $d_{\mathrm{AB}}   =  2  d_\rA d_\rB$.  Hence, the vector space spanned by the states of the composite system has dimension   
\begin{align}
\nonumber 
D_{\mathrm{AB}}    &=    2  d_{\rA\rB}^2  \\  
&=  2\left(  2 d_\rA    d_\rB\right) ^2 \, ,
\end{align} that is, twice  the dimensions of the vector space spanned by the product states.

 \subsection{Doubled Quantum Theory satisfies Purification} 

A generic state of a generic  system $\left( \spc H_0, \spc H_1\right) $  can be diagonalised as 
\begin{align}
\rho=  \left(  \sum_{i=1}^d  \lambda_i    \ket{\varphi_{i0}}\bra{\varphi_{i0}}  \right)   \oplus   \left( \sum_{j=1}^d \mu_j   \ket{\psi_{j1}}\bra{\psi_{j1}} \right)  ,
\end{align}
where $\left\{    \ket{\varphi_{i0}}\right\}_{i=1}^d$ is an orthonormal basis for $\spc H_0$ and $\left\{    \ket{\psi_{j1}}\right\}_{j=1}^d$ is an orthonormal basis for $\spc H_1$.  
  The state can be purified e.g.\ by adding one copy of system  $\left( \spc H_0, \spc H_1\right) $.  
 Since the composite system has two superselection sectors, there will be two types of purification:  purifications in the even  subspace $\spc H_0^{\rA\rB}$   and purifications in the odd subspace $\spc H_1^{\rA\rB}$.  
 A purification in the subspace $\spc H_0^{\rA\rB}$ will have the form 
  \begin{align}
  \ket{\Psi_0}   =   \left(  \sum_{i=1}^d  \sqrt{\lambda_i}  \ket{\varphi_{i0}}\ket{ \alpha_{i0}}   \right)    +   \left( \sum_{j=1}^d   \sqrt{\mu_j }   \ket{\psi_{j1}}\ket{\beta_{j1}} \right)   ,
  \end{align}
  where $\left\{    \ket{\alpha_{i0}}\right\}_{i=1}^d$  is an orthonormal basis for $\spc H_0$ and $\left\{    \ket{\beta_{j1}}\right\}_{j=1}^d$ is an orthonormal basis for $\spc H_1$.  
 A purification in the subspace $\spc H_1^{\rA\rB}$ will have the form 
  \begin{align}
  \ket{\Psi_1}   =   \left(  \sum_{i=1}^d   \sqrt{\lambda_i}    \ket{\varphi_{i0}}\ket{ \alpha'_{i1}}   \right)    +   \left( \sum_{j=1}^d   \sqrt{\mu_j }  \ket{\psi_{j1}}\ket{\beta'_{j0}} \right)  ,
  \end{align}
  where $\left\{    \ket{\alpha'_{i1}}\right\}_{i=1}^d$  is an orthonormal basis for $\spc H_1$ and $\left\{    \ket{\beta'_{j0}}\right\}_{j=1}^d$ is an orthonormal basis for $\spc H_0$.  
 Note that any two such   purifications   are equivalent under local unitary transformations: indeed, one has 
 \begin{align}
 \ket{\Psi_1}     = \left( I  \otimes   U\right)    \ket{\Psi_0}  ,
 \end{align} 
where $U$ is the unitary matrix defined by  
\begin{align}
U =    \left(  \sum_{i=1}^d   \ket{\alpha'_{i1}}\bra{\alpha_{i0}}  \right)    +   \left( \sum_{j=1}^d     \ket{\beta'_{j0}}\bra{\beta_{j1} } \right) . 
\end{align}
The same arguments apply to purifications within the same sector and to purifications where the purifying system is not a copy of the original system.    In summary, every state can be purified and every two purifications with the same purifying system are equivalent under local unitaries.  
    

\subsection{Doubled Quantum Theory satisfies Causality, Pure Sharpeness, and Purity Preservation}
Causality is immediate: for every system, the only deterministic effect is the identity matrix.  Pure Sharpness is also immediate:  every rank-one projector is a pure sharp effect.  As to Purity Preservation, note that the only pure transformations are quantum operations of the single-Kraus form $\map Q  \left( \cdot\right)   =  Q  \cdot Q^\dag$.  
 Clearly, the composition of two single-Kraus operations (both in parallel and in sequence) is a single-Kraus operation.  In other words, the composition of two pure transformations is pure.

\section{Permutability vs Strong-Symmetry: the example of the square bit}\label{app:square}

Consider the square bit \cite{Barrett}. Here the state space is a
square, and the pure states are its vertices. The group of reversible
transformations is the symmetry group of the square, which is the
dihedral group $D_{4}$. 
Every pair of vertices is a set of perfectly distinguishable pure
states. Fig.~\ref{fig:square} shows the situation
for the pure states
\begin{equation}
\alpha_{1}=\left(\begin{array}{c}
-1\\
1\\
1
\end{array}\right)\qquad\alpha_{2}=\left(\begin{array}{c}
-1\\
-1\\
1
\end{array}\right)\qquad\alpha_{3}=\left(\begin{array}{c}
1\\
-1\\
1
\end{array}\right),
\end{equation}
where the third component gives the normalisation. The pure observation-test
$\left\{ a_{1},a_{2}\right\} $, where
\begin{equation}
a_{1}=\frac{1}{2}\left(\begin{array}{ccc}
0 & 1 & 1\end{array}\right)\qquad a_{2}=\frac{1}{2}\left(\begin{array}{ccc}
0 & -1 & 1\end{array}\right),
\end{equation}
is the perfectly distinguishing test for the two sets $\left\{ \alpha_{1},\alpha_{2}\right\} $ and
$\left\{ \alpha_{1},\alpha_{3}\right\} $.

Now, since every set of perfectly distinguishable pure states has
two elements, the only non-trivial permutation of the elements of
such a set is the transposition. This permutation can be implemented
by considering the reflection through the axis of the segment connecting
the two points. 
Hence the square bit satisfies Permutability.  On the other hand, the square bit does  \emph{not} satisfy Strong Symmetry. A counterexample is shown in Fig.~\ref{fig:square}. Consider
the two maximal sets $\left\{ \alpha_{1},\alpha_{2}\right\} $ and
$\left\{ \alpha_{1},\alpha_{3}\right\} $. There are no reversible
transformations mapping the former to the latter because no symmetries
of the square map a side to a diagonal.

\section{Proof of proposition~\ref{prop:permstrong}\label{app:permstrong}}

\begin{proof}

The implication   ``Strong Symmetry   $\Rightarrow$ Permutability'' follows immediately from the definitions.  The implication ``Strong Symmetry    $\Rightarrow$ Reversible Controllability'' was proved by Lee and Selby  \cite{Control-reversible} using Causality, Purification, and the property that the product of two pure states is pure, which is guaranteed by our Purity Preservation axiom.  
Hence, we only need to prove the implications ``Permutability $\Rightarrow$ Strong Symmetry''  and   ``Reversible Controllability   $\Rightarrow$ Strong Symmetry''  

Let us prove that Permutability implies Strong Symmetry. The first
part of the proof is similar to the proof of theorem 30 of Ref.~\cite{Hardy-informational-2}.
Consider two maximal sets of perfectly distinguishable pure states
$\left\{ \varphi_{i}\right\} _{i=1}^{d}$ and $\left\{ \psi_{i}\right\} _{i=1}^{d}$.
Assuming Permutability, we will show that there exists a reversible
channel $\mathcal{U}$ such that $\psi_{i}=\mathcal{U}\varphi_{i}$,
for all $i=1,\ldots,d$. First of all, note that the states $\left\{ \varphi_{i}\otimes\psi_{j}\right\} $
are pure (by Purity Preservation) and perfectly distinguishable. Then
Permutability implies there exists a reversible transformation $\mathcal{U}$
such that for all $i=1,\ldots,d$ \cite{hardy2013}\begin{equation}
\begin{aligned} \Qcircuit @C=1em @R=.7em @!R { & \prepareC{\varphi_i} & \qw \poloFantasmaCn{\rA} & \multigate{1}{\mathcal U} & \qw \poloFantasmaCn{\rA} & \qw \\ & \prepareC{\psi_1} & \qw \poloFantasmaCn{\rA} & \ghost{\mathcal U} & \qw \poloFantasmaCn{\rA} & \qw } \end{aligned} ~= \!\!\!\! \begin{aligned} \Qcircuit @C=1em @R=.7em @!R { & \prepareC{\varphi_1} & \qw \poloFantasmaCn{\rA} &\qw \\ & \prepareC{\psi_i} & \qw \poloFantasmaCn{\rA} &\qw } \end{aligned}~.
\end{equation}Applying the pure effect $\varphi_{1}^{\dagger}$ to both sides of
the equation we obtain
\begin{equation}\label{eq:connectP}
\begin{aligned} \Qcircuit @C=1em @R=.7em @!R { & \prepareC{\varphi_i} & \qw \poloFantasmaCn{\rA} & \gate{\mathcal P} & \qw \poloFantasmaCn{\rA} & \qw } \end{aligned} ~=\!\!\!\! \begin{aligned} \Qcircuit @C=1em @R=.7em @!R { & \prepareC{\psi_i} & \qw \poloFantasmaCn{\rA} &\qw } \end{aligned}~,
\end{equation}with\begin{equation}
\begin{aligned} \Qcircuit @C=1em @R=.7em @!R { &\qw \poloFantasmaCn{\rA} & \gate{\mathcal P} & \qw \poloFantasmaCn{\rA} &\qw } \end{aligned} ~:=\!\!\!\! \begin{aligned} \Qcircuit @C=1em @R=.7em @!R { & & \qw \poloFantasmaCn{\rA} & \multigate{1}{\mathcal U} & \qw \poloFantasmaCn{\rA} & \measureD{\varphi_1^\dag} \\ & \prepareC{\psi_1} & \qw \poloFantasmaCn{\rA} & \ghost{\mathcal U} & \qw \poloFantasmaCn{\rA} & \qw } \end{aligned}~.
\end{equation}By construction, $\mathcal{P}$ is pure (by Purity Preservation) and
occurs with probability 1 on all the states $\left\{ \varphi_{i}\right\} _{i=1}^{d}$.  Moreover, the diagonalisation $\chi   =  \frac{1}{d}\sum_{i=1}^{d}\varphi_{i}$ implies that $\map P$ occurs with probability 1 on every state because $ \left( u\middle|\map P\middle|\chi\right)  =1 $ \cite{Chiribella-purification}. Since $\map P$ is a pure deterministic transformation on $ \rA $, it must be reversible \cite{Chiribella-purification}.  Hence, Eq.~\eqref{eq:connectP} proves  that the states $\left\{ \varphi_{i}\right\} _{i=1}^{d}$
can be reversibly transformed into the states $\left\{ \psi_{i}\right\} _{i=1}^{d}$.   
In short, Permutability implies Strong Symmetry. 

Let us prove now that Reversible Controllability implies Strong Symmetry.  Let $\left\{ \varphi_{i}\right\} _{i=1}^{d}$ and $\left\{ \psi_{i}\right\} _{i=1}^{d}$  be two pure maximal sets of a generic system $\rA$.  Since reversible transformations act transitively on pure states, for every $ i\in\left\lbrace 1,\ldots,d\right\rbrace  $, one can find a reversible transformation $\map U_i$ that maps $\psi_1$ into $\psi_i$, in formula 
\begin{align}\label{1cont}
\map U_i  \psi_1   =  \psi_i .   
\end{align}    
Moreover, Reversible Controllability implies that we can find a reversible transformation $\map U$ such that 
\begin{align}\label{2cont}
\begin{aligned} \Qcircuit @C=1em @R=.7em @!R { & \prepareC{\varphi_i} & \qw \poloFantasmaCn{\rA} & \multigate{1}{\mathcal U} & \qw \poloFantasmaCn{\rA} & \qw \\ &    & \qw \poloFantasmaCn{\rA} & \ghost{\mathcal U} & \qw \poloFantasmaCn{\rA} & \qw  } \end{aligned} ~=\!\!\!\!
\begin{aligned} \Qcircuit @C=1em @R=.7em @!R { & \prepareC{\varphi_i} & \qw \poloFantasmaCn{\rA} & \qw  & \qw &\qw  \\
&  & \qw \poloFantasmaCn{\rA} & \gate{\map U_i} &   \qw \poloFantasmaCn{\rA}   & \qw  } \end{aligned}
\end{align}
for every $i\in\left\lbrace 1,\ldots,d\right\rbrace $. Likewise, for every $i\in\left\lbrace 1,\ldots,d\right\rbrace $, one can always find a reversible transformation $\map V_i$ that transforms $\varphi_i$ into $\varphi_1$, in formula  
\begin{align}\label{3cont}
\map V_i \varphi_i   =  \varphi_1 . 
\end{align}    
And again, one can find a reversible transformation $\map V$  such that  
 \begin{align}\label{4cont}
\begin{aligned} \Qcircuit @C=1em @R=.7em @!R { & \qw \poloFantasmaCn{\rA} & \multigate{1}{\mathcal V} & \qw \poloFantasmaCn{\rA} & \qw \\  \prepareC{\psi_i} & \qw \poloFantasmaCn{\rA} & \ghost{\mathcal V} & \qw \poloFantasmaCn{\rA} & \qw  } \end{aligned}~ =
\!\!\!\!\begin{aligned} \Qcircuit @C=1em @R=.7em @!R { &  & \qw \poloFantasmaCn{\rA} & \gate{\map V_i}  &  \qw \poloFantasmaCn{\rA}  &\qw  \\
& \prepareC{\psi_i} & \qw \poloFantasmaCn{\rA} & \qw &  \qw &\qw } \end{aligned}
\end{align}
for every $i\in\left\lbrace 1,\ldots,d\right\rbrace$.      Combining Eqs.~(\ref{1cont}--\ref{4cont}), we obtain  
\begin{align}
\begin{aligned} \Qcircuit @C=1em @R=.7em @!R { & \prepareC{\varphi_i} & \qw \poloFantasmaCn{\rA} & \multigate{1}{\mathcal U} & \qw \poloFantasmaCn{\rA} & \multigate{1}{\map V} & \qw \poloFantasmaCn{\rA}  &\qw  \\ &   \prepareC{\psi_1} & \qw \poloFantasmaCn{\rA} & \ghost{\mathcal U} & \qw \poloFantasmaCn{\rA} & \ghost{\map V }  &\qw \poloFantasmaCn{\rA}  &\qw  } \end{aligned}~ =
\!\!\!\!\begin{aligned} \Qcircuit @C=1em @R=.7em @!R { & \prepareC{\varphi_1} & \qw \poloFantasmaCn{\rA} & \qw  \\
& \prepareC{\psi_i} & \qw \poloFantasmaCn{\rA} & \qw  } \end{aligned}
\end{align}
for every $i$.   Hence, one has\begin{equation}\label{eq:connectP2}
\begin{aligned} \Qcircuit @C=1em @R=.7em @!R { & \prepareC{\varphi_i} & \qw \poloFantasmaCn{\rA} & \gate{\mathcal P} & \qw \poloFantasmaCn{\rA} & \qw } \end{aligned} ~=\!\!\!\! \begin{aligned} \Qcircuit @C=1em @R=.7em @!R { & \prepareC{\psi_i} & \qw \poloFantasmaCn{\rA} &\qw } \end{aligned}~,
\end{equation}with\begin{equation}
\begin{aligned} \Qcircuit @C=1em @R=.7em @!R { &\qw \poloFantasmaCn{\rA} & \gate{\mathcal P} & \qw \poloFantasmaCn{\rA} &\qw } \end{aligned} ~:=\!\!\!\! \begin{aligned} \Qcircuit @C=1em @R=.7em @!R { & & \qw \poloFantasmaCn{\rA} & \multigate{1}{\mathcal U} & \qw \poloFantasmaCn{\rA} &   \multigate{1}{\mathcal V} & \qw \poloFantasmaCn{\rA} &  \measureD{\varphi_1^\dag} \\ & \prepareC{\psi_1} & \qw \poloFantasmaCn{\rA} & \ghost{\mathcal U} & \qw \poloFantasmaCn{\rA}  & \ghost{\mathcal V} & \qw \poloFantasmaCn{\rA} & \qw } \end{aligned}~.
\end{equation}
By the same argument used in the first part of the proof, we conclude that $\map P$ is a reversible transformation. Hence, Eq.~\eqref{eq:connectP2} implies that the set $\left\{ \varphi_{i}\right\} _{i=1}^{d}$ can be reversibly converted into the set $\left\{ \psi_{i}\right\} _{i=1}^{d}$.    In short, Reversible Controllability implies Strong Symmetry.
\end{proof}

\section{Proof that sharp theories with purification and unrestricted reversibility satisfy the Local Exchangeability axiom}\label{app:localexchange}

The aim of this appendix is to prove the following proposition: 

\begin{prop}\label{prop:local exchangeability}
Every sharp theory with purification  and unrestricted reversibility satisfies Local Exchangeability.  
\end{prop}

\begin{proof}[Proof of proposition~\ref{prop:local exchangeability}]
Let $\Psi \in\Pur\St_1 \left( \rA\otimes \rB\right) $ be a generic pure state and let $\rho_\rA$ and $\rho_\rB$ its marginal states, diagonalised as 
\begin{equation}    \rho_\rA   =   \sum_{i=1}^r     p_i   \alpha_i   \qquad \textrm{and} \qquad \rho_\rB  =   \sum_{i=1}^{r} p_i   \beta_i  ,\end{equation}
where $  p_i>0 $ for all $ i=1,\dots,r $, and $ r\leq\min\left\lbrace d_\rA,d_\rB\right\rbrace  $. Here we are invoking a result of Ref.~\cite{TowardsThermo}, where we showed that the marginals of a pure bipartite state have the same spectrum (up to vanishing elements).   Now, we extend the set of eigenstates of $\rho_\rA$ and $\rho_\rB$ to two pure maximal sets. Without loss of generality assume $ d_\rA \leq d_\rB $.   By  the Permutability axiom, there must exist a reversible transformation $\map U  \in  \Det\Transf \left( \rB\otimes \rA  ,\rA\otimes \rB\right) $ such that  
\begin{equation}    \map U    \left(    \beta_1\otimes \alpha_i\right)    =  \alpha_1 \otimes \beta_i  ,\qquad \forall i\in  \left\{  1,\ldots ,  d_\rA\right\} .\end{equation} 
Similarly, there must exist a reversible transformation 
$\map V  \in  \Det\Transf \left( \rB\otimes \rA  ,\rA\otimes \rB\right) $ such that  
\begin{equation}    \map V   \left(   \beta_i\otimes \alpha_1\right)    =  \alpha_i \otimes \beta_1  ,\qquad \forall i\in  \left\{  1,\ldots ,  d_\rA\right\} .\end{equation} 
  At this point, we define  the pure transformations  
  \begin{equation}
  \begin{aligned} \Qcircuit @C=1em @R=.7em @!R { &\qw \poloFantasmaCn{\rA} & \gate{\mathcal P} & \qw \poloFantasmaCn{\rB} &\qw } \end{aligned} ~:=\!\!\!\! \begin{aligned} \Qcircuit @C=1em @R=.7em @!R { & \prepareC{\beta_1} & \qw \poloFantasmaCn{\rB} & \multigate{1}{\mathcal U} & \qw \poloFantasmaCn{\rA} & \measureD{\alpha_1^\dag} \\ &  & \qw \poloFantasmaCn{\rA} & \ghost{\mathcal U} & \qw \poloFantasmaCn{\rB} & \qw } \end{aligned}~,
  \end{equation}
  
  \begin{equation}
  \begin{aligned} \Qcircuit @C=1em @R=.7em @!R { &\qw \poloFantasmaCn{\rB} & \gate{\mathcal Q} & \qw \poloFantasmaCn{\rA} &\qw } \end{aligned} ~:=\!\!\!\! \begin{aligned} \Qcircuit @C=1em @R=.7em @!R { & & \qw \poloFantasmaCn{\rB} & \multigate{1}{\mathcal V} & \qw \poloFantasmaCn{\rA} & \qw \\ & \prepareC{\alpha_1} & \qw \poloFantasmaCn{\rA} & \ghost{\mathcal V} & \qw \poloFantasmaCn{\rB} & \measureD{\beta_1^\dag} } \end{aligned}~.
  \end{equation}
and the pure state 
\begin{equation}
\begin{aligned} \Qcircuit @C=1em @R=.7em @!R { & \multiprepareC{1}{\Psi'} & \qw \poloFantasmaCn{\rB}  &\qw \\ & \pureghost{\Psi'} & \qw \poloFantasmaCn{\rA} & \qw  }\end{aligned} ~:= \!\!\!\! \begin{aligned}\Qcircuit @C=1em @R=.7em @!R { & \multiprepareC{1}{\Psi} & \qw \poloFantasmaCn{\rA} & \gate{\mathcal P} & \qw \poloFantasmaCn{\rB} &\qw \\ & \pureghost{\Psi} & \qw \poloFantasmaCn{\rB} & \gate{\mathcal Q} & \qw \poloFantasmaCn{\rA} &\qw }\end{aligned}~,
\end{equation}
where the purity of $\map P$, $\map Q$, and $\Psi'$ follows from Purity Preservation.  Like in the proof of proposition~\ref{prop:permstrong}, we can prove that   $\map P$ and $\map Q$ are in fact channels, so $ u_\rB\map P=u_\rA $ and $ u_\rA\map Q=u_\rB $. Hence $\Psi'$  and $ {\tt SWAP}  \Psi  $ have the same marginals.  Then, the uniqueness of purification applied to both systems $ \rA $ and $ \rB $ (viewed as purifying systems of one another) implies that there  exist two reversible transformations $\map W_\rA$ and $\map W_B$ such that 
\begin{equation}
\begin{aligned}\Qcircuit @C=1em @R=.7em @!R { &\multiprepareC{1}{\Psi}  & \qw \poloFantasmaCn{\rA} &\multigate{1}{\tt SWAP}  &   \qw \poloFantasmaCn{\rB}  & \qw   \\ & \pureghost{\Psi}  & \qw \poloFantasmaCn{\rB} & \ghost{\tt SWAP}  &   \qw \poloFantasmaCn{\rA}  & \qw   }  
\end{aligned}~=\!\!\!\!\begin{aligned}\Qcircuit @C=1em @R=.7em @!R { & \multiprepareC{1}{\Psi'} & \qw \poloFantasmaCn{\rB} & \gate{\mathcal W_\rB} & \qw \poloFantasmaCn{\rB} &\qw \\ & \pureghost{\Psi'} & \qw \poloFantasmaCn{\rA} & \gate{\mathcal W_\rA} & \qw \poloFantasmaCn{\rA} &\qw }\end{aligned} ~=
\end{equation}
\begin{equation}
=\!\!\!\!\begin{aligned}\Qcircuit @C=1em @R=.7em @!R { & \multiprepareC{1}{\Psi} & \qw \poloFantasmaCn{\rA} & \gate{\map P} & \qw \poloFantasmaCn{\rB} & \gate{\mathcal W_\rB} & \qw \poloFantasmaCn{\rB} &\qw \\ & \pureghost{\Psi} & \qw \poloFantasmaCn{\rB} & \gate{\map Q} & \qw \poloFantasmaCn{\rA} & \gate{\mathcal W_\rA} & \qw \poloFantasmaCn{\rA} &\qw }\end{aligned}~.
\end{equation}
Hence, we have shown that there exist two local \emph{pure} channels $\map C  :  =   \map W_\rB \map P$ and $\map D : =  \map W_\rA  \map Q$ that reproduce the action of the swap channel on the state $\Psi$.
\end{proof}
Note that Local Exchangeability is implemented in this setting by \emph{pure} channels.

In passing, we also mention that the validity of Local Exchangeability implies that every state admits  a \emph{symmetric purification}, in the following sense:

\begin{defn}\cite{Chiribella-Scandolo-entanglement} 
	Let $\rho$ be a state of system $\mathrm{A}$ and let $\Psi$ be
	a pure state of $\mathrm{A}\otimes\mathrm{A}$. We say that $\Psi$
	is a \emph{symmetric purification} of $\rho$ if 
	\begin{equation}
	\begin{aligned}\Qcircuit @C=1em @R=.7em @!R { & \multiprepareC{1}{\Psi} & \qw \poloFantasmaCn{\rA} & \qw  \\ & \pureghost{\Psi} & \qw \poloFantasmaCn{\rA} & \measureD{u}}\end{aligned}~=\!\!\!\!\begin{aligned}\Qcircuit @C=1em @R=.7em @!R {& \prepareC {\rho} & \qw \poloFantasmaCn{\rA} & \qw }\end{aligned}~,
	\end{equation}and\begin{equation}
	\begin{aligned}\Qcircuit @C=1em @R=.7em @!R { & \multiprepareC{1}{\Psi} & \qw \poloFantasmaCn{\rA} & \measureD{u} \\ &\pureghost{\Psi} &\qw \poloFantasmaCn{\rA} &\qw}\end{aligned}~=\!\!\!\!\begin{aligned}\Qcircuit @C=1em @R=.7em @!R {& \prepareC {\rho} & \qw \poloFantasmaCn{\rA} & \qw }\end{aligned}~.
	\end{equation}
	\end{defn}  

 With the above notation, we have the following 
\begin{prop}
In every sharp theory with purification and unrestricted reversibility, every state of every finite system admits a symmetric purification.
\end{prop}
The existence of a symmetric purification for every state is guaranteed by theorem 3 of Ref.~\cite{Chiribella-Scandolo-entanglement}.

\end{document}